\documentclass[12pt]{article}
\usepackage{amsmath}
\usepackage{amsthm} 
\usepackage{amsfonts}
\usepackage{booktabs}
\usepackage{dsfont}
\usepackage{graphicx,psfrag,epsf}
\usepackage{enumitem}
\usepackage{url} 
\usepackage{xcolor}
\usepackage{multirow,multicol}
\usepackage{tabularx}
\usepackage[title]{appendix}
\usepackage[section]{placeins}
\usepackage{stackengine}
\usepackage{setspace}
\usepackage[font={footnotesize}]{subcaption}
\usepackage{lscape}
\usepackage{lipsum}
\usepackage{float}

\usepackage[
colorlinks=true,
pagebackref=false,
pdfauthor={Timo Dimitriadis and Julie Schnaitmann},%
pdftitle={Forecast Encompassing Tests for the Expected Shortfall},%
citecolor=blue,
linkcolor=blue,
urlcolor=blue,
filecolor=blue
]{hyperref}      
\usepackage[authoryear,round]{natbib}

\numberwithin{equation}{section}

\newcommand{\blind}{0}

\newcolumntype{R}{>{\raggedleft\arraybackslash}X}
\newcolumntype{C}{>{\centering\arraybackslash}X}

\DeclareMathOperator{\Var}{Var}
\DeclareMathOperator*{\argmin}{arg\,min}

\newcommand{\bfc}{\hat{\boldsymbol{f}_t}}
\newcommand{\bq}{\hat{\boldsymbol{q}}_t}
\newcommand{\be}{\hat{\boldsymbol{e}}_t}

\newcommand{\bW}{\boldsymbol{W}_t}

\newcommand\ubar[1]{\stackunder[1.2pt]{$#1$}{\rule{.8ex}{.075ex}}}

\newtheoremstyle{boldtitle}
{}
{}
{}
{}
{\bfseries}
{.}
{.5em}
{{\thmname{#1 }}{\thmnumber{#2}}{\thmnote{ (#3)}}}

\theoremstyle{boldtitle}

\newtheorem{theorem}{Theorem}[section]
\newtheorem{lemma}[theorem]{Lemma}
\newtheorem{definition}[theorem]{Definition}
\newtheorem{proposition}[theorem]{Proposition}

\newtheorem{example}[theorem]{Example}
\newtheorem{assumption}[theorem]{Assumption}

\begin{document}
	
	\def\spacingset#1{\renewcommand{\baselinestretch}%
		{#1}\small\normalsize} \spacingset{1}

	
	\if0\blind
	{
		\title{\bf Forecast Encompassing Tests for the Expected Shortfall}
		\author{Timo Dimitriadis\footnote{Corresponding Author. Schloß-Wolfsbrunnenweg 35, 69118 Heidelberg, Germany; Tel.  +49 176 47796302. 
}\vspace{.2cm}\\
			Heidelberg Institute for Theoretical Studies \\
			Institute of Economics, Universität Hohenheim \\
			\texttt{\href{mailto:timo.dimitriadis@h-its.org}{timo.dimitriadis@h-its.org}}
			\vspace{.5cm} \\
			and 
			\vspace{.5cm} \\
			Julie Schnaitmann\vspace{.2cm}\\
			Department of Economics, University of Konstanz \\
			\texttt{\href{mailto:julie.schnaitmann@uni-konstanz.de}{julie.schnaitmann@uni-konstanz.de}}
		}
		\maketitle
	} \fi
	
	\if1\blind
	{
		\bigskip
		\bigskip
		\bigskip
		\begin{center}
			{\LARGE\bf Forecast Encompassing Tests for the Expected Shortfall} \\
			\vspace{1cm}
			\today 
		\end{center}
		\medskip
	} \fi
	
	\bigskip
	\begin{abstract}
		We introduce new forecast encompassing tests for the risk measure Expected Shortfall (ES). 
		The ES currently receives much attention through its introduction into the Basel III Accords, which stipulate its use as the primary market risk measure for the international banking regulation.
		We utilize joint loss functions for the pair ES and Value at Risk to set up three ES encompassing test variants.
		The tests are built on misspecification robust asymptotic theory and we investigate the finite sample properties of the tests in an extensive simulation study.
		We use the encompassing tests to illustrate the potential of forecast combination methods for different financial assets.
	\end{abstract}
	
	\noindent%

	\spacingset{1.45} 
	
	\section{Introduction}

	Through the recent introduction of Expected Shortfall (ES) as the primary market risk measure for the international banking regulation in the Basel III Accords \citep{Basel2016, Basel2017}, there is a great demand for reliable methods for evaluating and comparing the predictive ability of competing ES forecasts.
	The ES at probability level $\alpha \in (0,1)$ is defined as the expectation of the returns smaller than the respective $\alpha$-quantile (the Value at Risk, VaR), where $\alpha$ is usually chosen to be 2.5\% as proposed by the Basel Accords.
	The ES is replacing the VaR in the banking regulation as it overcomes several shortcomings of the latter such as being not coherent and its inability to capture tail risks beyond the $\alpha$-quantile \citep{Artzner1999,Danielsson2001,Basel2013}. 
	While the empirical properties favor the ES over the VaR as a risk measure, the ES lacks elicitability, which implies that no strictly consistent loss functions exist.
	The non-elicitability of the ES is overcome by considering the pair VaR and ES which are \textit{jointly elicitable}, i.e.\ there exist joint loss functions for the VaR and the ES \citep{Fissler2016}.
	This discovery triggered a rapidly growing branch of literature in developing forecasting methods and forecast evaluation techniques for the ES, see \cite{Patton2019}, \cite{DimiBayer2019}, \cite{BayerDimi2019}, \cite{Taylor2019}, \cite{Barendse2020},  \cite{FisslerZiegelGneiting2016} and \cite{Nolde2017} among others.

	A desirable tool for the comparison of ES forecasts are encompassing tests, which however build upon the existence of strictly consistent loss functions.
	Given two competing forecasts A and B, forecast encompassing tests the null hypothesis that forecast A \textit{performs not worse} than any (linear) combination of these forecasts. 
	This is carried out by testing whether the optimal combination weight of forecast B deviates significantly from zero.\footnote{For the classical theory on forecast encompassing see \cite{HendryRichard1982}, \cite{MizonRichard1986}, \cite{Diebold1989}, \cite{Ericsson1993}, \cite{harvey1998tests}, \cite{clark2001tests}, \cite{GiacominiKomunjer2005}, \cite{NewboldHarvey2007} and \cite{ClementsHarvey2009} among others.}
	This null hypothesis allows for the convenient interpretation that forecast B does not add any information to forecast A and thus, forecast A is superior to forecast B. 
	The existence of appropriate loss functions is inevitable for encompassing tests for two reasons.
	First, the \textit{superior performance} of competing forecasts is defined in the statistical sense by using strictly consistent loss functions.
	Second, loss and identification functions are crucial for M- or GMM-estimation of the optimal forecast combination weights through an appropriate regression framework for the risk measure under consideration.
	
	In this paper, we introduce novel encompassing tests for the ES based on the joint loss functions for the VaR and ES developed in \cite{Fissler2016}.
	We introduce the following three test variants for the ES.
	First, we propose to jointly test forecast encompassing for the VaR and ES, henceforth denoted the \textit{joint VaR and ES encompassing test}.
	We introduce a second test variant, denoted the  \textit{auxiliary ES encompassing test}, which estimates the optimal combination weights for the vector of the VaR and ES, however, only tests the parameters associated with the ES.
	While incorporating both, VaR and ES forecasts, this variant only tests encompassing of the ES forecasts.
	The third variant overcomes the tests' dependence on VaR forecasts and tests encompassing of competing ES forecasts stand-alone, which comes at the cost of a potential model misspecification.
	We henceforth call this test the \textit{strict ES encompassing test}.
	This variant is particularly relevant due to the current set of rules established by the Basel Committee of Banking Supervision, which only imposes the financial institutions to report ES forecasts \citep{Basel2016, Basel2017}.
	Only this test variant can be applied in situations where the person evaluating the forecasts merely has forecasts for the ES at hand. 	
	However, in situations where both, the VaR and ES forecasts (stemming from the same model or forecasting procedure) are available, application of the joint or auxiliary tests is generally recommended.
	
	We implement the encompassing tests through M-estimation of the optimal combination weights \citep{Patton2019, DimiBayer2019} and in an environment with asymptotically non-vanishing estimation uncertainty of the forecasting procedures \citep{GiacominiKomunjer2005, GiacominiWhite2006}.
	As the strict ES encompassing test is potentially subject to model misspecification, 
	we derive the asymptotic distribution of the test statistics in a general setting which allows for misspecified models.
	This generalizes the asymptotic theory of \cite{Patton2019}, \cite{DimiBayer2019} and \cite{BayerDimi2019} to potentially misspecified (and nonlinear) models.
	We base the Wald test statistics of the encompassing tests on a misspecification-robust covariance estimator.
	Our implementation further introduces a \textit{link} or \textit{combination function} which captures the different linear and nonlinear forecast combination methods in the existing encompassing testing literature, see \cite{ClementsHarvey2009} and \cite{ClementsHarvey2010} among others.
	
	We analyze the finite sample behavior of our encompassing tests and the effect of the potential model misspecification in an extensive simulation study using models from various model classes associated with the ES.
	For this, we consider classical GARCH models, the GAS (generalized auto-regressive score) models with time-varying higher moments of \cite{Creal2013}, the GAS models for the VaR and ES of \cite{Patton2019} and the ES-CAViaR models of \cite{Taylor2019}.
	Data stemming from the latter three model classes induces some model misspecification for the strict ES encompassing test, which allows us to evaluate the effect the misspecification has on our tests.
	We find that all tests exhibit approximately correct size and good power properties for all considered simulations.
	This also holds for the strict ES encompassing test which demonstrates that this test is robust to the degree of model misspecification we usually encounter in financial applications.
	
	Tests for forecast encompassing are commonly used to establish a theoretical basis for forecast combinations in cases when encompassing is rejected for both forecasts \citep{ClementsHarvey2009, NewboldHarvey2007, GiacominiKomunjer2005}.
	This implies that neither of the forecasts stand-alone performs as good as an optimal forecast combination, which indicates that a forecast combination incorporates more information than the individual forecasts.
	\cite{GiacominiKomunjer2005}, \cite{Timmermann2006}, \cite{Halbleib2012} and \cite{Taylor2019FCcomb} advocate general forecast combination methods for multiple reasons and particularly for risk measures with small probability levels, as it is customary for the VaR and the ES.
	
	We apply our encompassing tests to ES forecasts from classical GARCH and GAS models, but also from the recently developed dynamic ES models of \cite{Taylor2019} and \cite{Patton2019} for daily returns of the IBM stock, the S$\&$P 500 and the DAX 30 indices.
	The test results imply that for the IBM stock, forecast combination methods outperform the stand-alone forecasting models in many instances.
	In comparison, this pattern seems to be less pronounced for the S$\&$P 500 and the DAX 30 indices, which are already well diversified through their versatile composition.
	Thus, classical diversification gains \citep{Timmermann2006} of forecast combination methods might be less pronounced for stock indices.
	The two ES based test variants exhibits very similar results, which further indicates that the strict ES test is robust against potential misspecifications in financial settings.

	The classical idea of forecast encompassing goes back to \cite{HendryRichard1982}, \cite{ChongHendry1986} and \cite{MizonRichard1986} and is developed for mean forecasts under the squared loss function.
	Broad reviews on encompassing testing are provided e.g.\ by \cite{NewboldHarvey2007} and \cite{ClementsHarvey2009}.
	\cite{HarveyNewbold2000} extend the encompassing technique which classically focuses on two competing forecasts to encompassing of multiple forecasts.
	\cite{GiacominiKomunjer2005} develop (conditional) encompassing of quantile forecasts and focus on encompassing tests for \textit{methods} instead of \textit{models}.
	\cite{ClementsHarvey2010} generalize encompassing tests to probabilistic forecasts by relying on strictly consistent scoring rules.
	\cite{GiacominiKomunjer2005} and \cite{ClementsHarvey2010} investigate extensions of encompassing to more complicated functionals of the conditional distribution.
	Our work pursues this path by developing encompassing tests for the ES as a prominent example of higher-order elicitable functionals where only joint loss functions for vector-valued functionals are available.
	Our testing approach can be adapted to further higher-order elicitable functionals such as the pair mean, variance and the Range Value at Risk \citep{Cont2010, Embrechts2018, FisslerZiegel2019}.
	
	The rest of the paper is organized as follows. 
	In Section \ref{sec:Theory}, we introduce encompassing tests for the ES and derive the asymptotic distribution of the associated test statistics under model misspecification.
	Section \ref{sec:SimulationStudy} presents an extensive simulation study analyzing the size and power properties of our tests.
	In Section \ref{sec:EmpiricalApplication}, we apply the testing procedure to daily financial returns of the IBM stock and the S\&P 500 and DAX 30 indices and Section \ref{sec:Conclusion} concludes.
	All proofs are deferred to Appendix \ref{sec:Proofs}.
	Technical details of the proofs and additional results are provided in the supplementary material.

	\section{Theory}
	\label{sec:Theory}
	
	We consider a stochastic process $Z = \left\{ Z_t: \Omega \to \mathbb{R}^{l+1}, l \in \mathbb{N}, t= 1,\dots,T \right\}$, which is defined on some common and complete probability space $(\Omega, \mathcal{F}, \mathbb{P})$, where $\mathcal{F} = \left\{ \mathcal{F}_t, t = 1, \dots, T \right\}$ and $\mathcal{F}_t = \sigma \left\{ Z_s, s\le t \right\}$.
	We partition the stochastic process as $Z_t = (Y_t,X_t)$, where $Y_t: \Omega \to \mathbb{R}$ is an absolutely continuous random variable of interest and $X_t: \Omega \to \mathbb{R}^{l}$ is a vector of explanatory variables.
	We denote the conditional distribution of $Y_{t+1}$ given the information set $\mathcal{F}_t$ by $F_t$.
	Accordingly, $\mathbb{E}_t$, $\Var_t$ and $h_t$ denote the expectation, variance and density corresponding to $F_t$.
	Following \cite{GiacominiKomunjer2005}, we consider ($\mathcal{F}_t$-measurable) one-step ahead forecasts, henceforth denoted by $\hat f_{t}, \hat q_{t}$ and $\hat e_{t}$, which are generated by a function $f \big( \gamma_{t,m}, Z_t,Z_{t-1},\dots \big)$, which is fixed over time.
	For this, $\gamma_{t,m}$ denotes the (estimated) model parameters at time $t$ or alternatively the semi- or non-parametric estimator used in the construction of the forecasts.
	This construction allows for both, fixed forecasting schemes, where the model parameters $\gamma_{t,m}$ are only estimated once, and rolling window forecasting schemes, where the parameters $\gamma_{t,m}$ are re-estimated in each step.	
	We denote general competing forecasts by $\bfc = (\hat f_{1,t}, \hat f_{2,t})$, specific VaR (quantile) forecasts by $\bq = (\hat q_{1,t}, \hat q_{2,t})$ and ES forecasts by $\be = (\hat e_{1,t}, \hat e_{2,t})$.
	
	In the context of evaluating point forecasts, an important property of risk measures (or more general statistical functionals) is \textit{elicitability} \citep{Gneiting2011}.
	Elicitability means that there exist strictly consistent loss functions, i.e.\ loss functions $\rho(Y,f)$ depending on the random variable $Y \sim F$ and the issued forecast $f$, whose expectation $\mathbb{E} \left[ \rho(Y, \cdot) \right]$ is uniquely minimized by the true risk measure
	$\Gamma(F)$.
	Using such a loss function, one can assess the quality of issued forecasts by comparing their average losses induced by the realizations of the predicted variable.
	Evaluating forecasts through strictly consistent loss functions has the desired impact that it incentivizes financial institutions to truthfully report their correct forecasts \citep{Gneiting2011, FisslerZiegelGneiting2016}.
	As a direct consequence, the literature on tests for forecast comparison and forecast rationality evolves around the associated loss functions, see \cite{MizonRichard1986}, \cite{DieboldMariano1995}, \cite{ElliottKomunjerTimmermann2005}, \cite{GiacominiKomunjer2005}, \cite{GiacominiWhite2006}, \cite{PattonTimmermann2007}, \cite{ClementsHarvey2010}, \cite{Gneiting2011} and \cite{Patton2011} among many others.

	Many important statistical functionals such as the variance, the ES, the minimum, the maximum and the mode are not elicitable, i.e.\ no strictly consistent loss functions exist \citep{Gneiting2011, Heinrich2014, Fissler2016}.
	This deficiency calls for generalized approaches in many academic disciplines.
	We built our test procedure for the ES on such an approach, which considers multiple functionals stacked as vectors and considers joint elicitability.
	\cite{Fissler2016} show that the ES is jointly elicitable with the VaR by constructing strictly consistent joint loss functions for this pair, which we utilize in our encompassing approach.
	
	In the following section, we formally introduce the concept of forecast encompassing in the classical case of one-dimensional, real-valued and elicitable functionals.
	Subsequently, we make use of the higher-order elicitability of the ES and generalize the encompassing approach to ES forecasts in Section \ref{sec:EncompassingTestsES}.

	\subsection{The Encompassing Principle}
	\label{sec:EncompassingPrinciple}
	
	Following e.g. \cite{HendryRichard1982}, \cite{MizonRichard1986}, \cite{Diebold1989} and \cite{GiacominiKomunjer2005}, we formally introduce the classical concept of linear forecast encompassing for one-dimensional, real-valued and elicitable functionals.
	We assume that two competing forecasters predict the variable of interest $Y_{t+1}$ and issue one-step ahead point forecasts  $\bfc = \big( \hat f_{1,t}, \hat f_{2,t} \big)$ for a given functional $\Gamma(F_t)$.\footnote{While we focus our approach on one-step ahead forecasts, extensions to multi-step ahead forecasts are straight-forward by employing a HAC-type estimator for the asymptotic covariance.} 
	In order to conduct the forecast evaluation in an out-of-sample fashion, we divide the sample size $T$ in an in-sample part of size $m$ and an out-of-sample part of size $n$ such that $T = m+n$.
	The in-sample period is used to generate the forecasts $\hat f_{1,t}$ and $\hat f_{2,t}$ as described in the beginning of Section \ref{sec:Theory}, while the out-of-sample period is used for the evaluation of the forecasts.
	This procedure poses little restrictions on how to generate the forecasts and allows for parametric, semiparametric or nonparametric techniques and for nested and non-nested forecasting procedures \citep{GiacominiKomunjer2005}.

	Let $\rho \big(Y_{t+1}, \hat f_t \big)$ be a strictly consistent loss function for $\Gamma(\cdot)$.
	Then, we say that forecast $\hat f_{1,t}$ encompasses $\hat f_{2,t}$ at time $t$, if
	\begin{align}
	\label{eqn:GeneralFcEncpConditionH0}
	\mathbb{E} \left[ \rho \big(Y_{t+1}, \hat f_{1,t} \big) \right]
	\le 
	\mathbb{E} \left[ \rho \big(Y_{t+1}, \theta_{1} \hat f_{1,t} + \theta_{2} \hat f_{2,t} \big) \right],
	\end{align}
	for all $\big(\theta_{1},\theta_{2} \big) \in \Theta \subseteq{\mathbb{R}^2}$.
	Equation (\ref{eqn:GeneralFcEncpConditionH0}) implies that, in terms of the loss induced by $\rho$, the forecast $\hat f_{1,t}$ is at least as good as any (linear) combination of $\hat f_{1,t}$ and $\hat f_{2,t}$.
	Hence, forecast $\hat f_{2,t}$ does not add any information on $Y_{t+1}$ which is not already incorporated in $\hat f_{1,t}$.
	We define $\big( \theta^\ast_{1}, \theta^\ast_{2} \big)$ as the optimal combination parameters which minimize the expected loss,
	\begin{align}
	\label{eqn:ThetaStarDef}
	\big( \theta^\ast_{1}, \theta^\ast_{2} \big) = \underset{(\theta_{1},\theta_{2})  \in \Theta}{\operatorname{arg \, min}} \; \mathbb{E} \left[ \rho \big(Y_{t+1}, \theta_{1} \hat f_{1,t} + \theta_{2} \hat f_{2,t} \big) \right].
	\end{align}
	By definition, it holds that $\mathbb{E} \left[ \rho \big(Y_{t+1}, \theta_{1} \hat f_{1,t} + \theta_{2} \hat f_{2,t} \big) \right]
	\ge
	\mathbb{E} \left[ \rho \big(Y_{t+1}, \theta^\ast_{1} \hat f_{1,t} + \theta^\ast_{2} \hat f_{2,t} \big) \right]$
	for all $\big(\theta_{1},\theta_{2} \big) \in \Theta$.
	In particular, this implies that
	\begin{align}
	\label{eqn:GeneralFcEncpConditionMEstimation}
	\mathbb{E} \left[ \rho \big(Y_{t+1}, \hat f_{1,t} \big) \right]
	\ge
	\mathbb{E} \left[ \rho \big(Y_{t+1}, \theta^\ast_{1} \hat f_{1,t} + \theta^\ast_{2} \hat f_{2,t} \big) \right].
	\end{align}
	Combining (\ref{eqn:GeneralFcEncpConditionH0}) and (\ref{eqn:GeneralFcEncpConditionMEstimation}) yields the following definition of forecast encompassing.
	\begin{definition}[Linear Forecast Encompassing for Elicitable Functionals]
		\label{def:FcEnc1dim}
		We say that forecast $\hat f_{1,t}$ encompasses $\hat f_{2,t}$ at time $t$  with respect to the loss function $\rho$
		if and only if
		\begin{align}
		\label{eqn:UncondEncomp}
		\mathbb{E} \left[ \rho \big(Y_{t+1}, \hat f_{1,t} \big) \right]
		=
		\mathbb{E} \left[ \rho \big(Y_{t+1}, \theta^\ast_{1} \hat f_{1,t} + \theta^\ast_{2} \hat f_{2,t} \big) \right],
		\end{align}
		which is equivalent to $\big( \theta^\ast_{1}, \theta^\ast_{2} \big) = \big(1, 0\big)$.
	\end{definition}
	
	Tests for forecast encompassing are carried out through the following steps.
	First, we regress the realizations $Y_{t+1}$ onto the forecasts $\hat f_{1,t}$ and $\hat f_{2,t}$ using an appropriate regression technique for the functional under consideration in order to obtain the estimated combination (or encompassing) parameters $\hat \theta_n$ and their asymptotic distribution.
	Then, we test whether these parameters equal one and zero respectively.

	As discussed e.g.\ in \cite{ClementsHarvey2009} and  \cite{ClementsHarvey2010}, there exist several different testing specifications available for the encompassing principle, which differ in terms of the admissible specifications of the linear (or nonlinear) forecast combination formula.
	We generalize  and unify these approaches by introducing a general \textit{link} or \textit{combination function},
	\begin{align}
	\label{eqn:DefLinkFct}
	g: \mathfrak{F} \times \Theta \to \mathbb{R}, \qquad (\bfc, \theta) \mapsto g(\bfc,\theta),
	\end{align}
	which maps the forecasts and the respective parameters onto a linear or nonlinear forecast combination and where $\mathfrak{F}$ denotes the random space of the issued forecasts.
	For this, the function $g$ and the parameter space $\Theta$ have to be chosen such that there exists a $\theta_0 \in \Theta$, such that $g(\bfc,\theta_0) = \hat f_{1,t}$ almost surely, which enables testing whether $\hat f_{1,t}$ alone captures the full information provided by any forecast combination through testing the parametric restriction $\theta^\ast = \theta_0$.
	\begin{definition}[General Forecast Encompassing for Elicitable Functionals]
		\label{def:GeneralFcEnc1dim}
		We say that forecast $\hat f_{1,t}$ encompasses $\hat f_{2,t}$ at time $t$ with respect to the loss function $\rho$ and with respect to the link function $g$ if and only if 
		\begin{align}
		\label{eqn:GeneralFcEncpConditionH0LinkFct}
		\mathbb{E} \left[ \rho \big(Y_{t+1}, \hat f_{1,t} \big) \right]
		=
		\mathbb{E} \left[ \rho \big(Y_{t+1}, g( \bfc, \theta^\ast) \big) \right],
		\end{align}
		which is equivalent to $\theta^\ast = \theta_0$.
	\end{definition}
	This general definition unifies the following existing specifications of forecast encompassing, but also allows for more general linear and nonlinear specifications, see e.g.\ \cite{Ericsson1993}, \cite{ClementsHarvey2009} and  \cite{ClementsHarvey2010}.
	\begin{example}
		\label{exmpl:LinkFunctions}
		Prominent examples for linear and nonlinear forecast encompassing are the following link functions and associated null hypotheses,
		\begin{enumerate}[label={(\arabic*)}]
			\item 
			$g( \bfc, \theta) = \theta_1 + \theta_2 \hat f_{1,t} + \theta_3 \hat f_{2,t}$ and  $\mathbb{H}_0: (\theta_2^\ast, \theta_3^\ast) = (1,0)$ or $\mathbb{H}_0:  (\theta_1^\ast, \theta_2^\ast, \theta_3^\ast) = (0,1,0)$,
			
			\item 
			$g( \bfc, \theta) = \theta_1 + \theta_2 \hat f_{1,t} + (1-\theta_2) \hat f_{2,t}$ and  $\mathbb{H}_0:  \theta_2^\ast = 1$ or $\mathbb{H}_0:  (\theta_1^\ast, \theta_2^\ast) = (0,1)$,
			
			\item 
			$g( \bfc, \theta) = \theta_1 + \hat f_{1,t} + \theta_2 \hat f_{2,t}$ and  $\mathbb{H}_0: \theta_2^\ast = 0$ or  $\mathbb{H}_0: (\theta_1^\ast, \theta_2^\ast) = (0,0)$,
			
			\item 
			$g( \bfc, \theta) = \theta_1 \hat f_{1,t} + \theta_2 \hat f_{2,t}$ and $\mathbb{H}_0: (\theta_1^\ast, \theta_2^\ast) = (1,0)$,
			
			\item 
			$g( \bfc, \theta) = \theta_1 \hat f_{1,t} + (1-\theta_1) \hat f_{2,t}$ and  $\mathbb{H}_0: \theta_1^\ast = 1$,
			
			\item 
			$g( \bfc, \theta) = \hat f_{1,t} + \theta_1 \hat f_{2,t}$ and $\mathbb{H}_0: \theta_1^\ast = 0$,
			
			\item 
			$g( \bfc, \theta) = \theta_1 \pm \exp \big( \theta_2 \log(\pm \hat f_{1,t}) + \theta_3 \log(\pm \hat f_{2,t}) \big)$ and $\mathbb{H}_0: (\theta_2^\ast, \theta_3^\ast) = (1,0)$.
		\end{enumerate}
	\end{example}

	\subsection{Forecast Encompassing for the Expected Shortfall}
	\label{sec:EncompassingTestsES}
	
	In this section, we consider encompassing tests for the ES. 
	For absolutely continuous distributions $F_t$, the ES is formally defined as
	\begin{align}
	ES_{t,\alpha} ( Y_{t+1} ) = \mathbb{E}_t \left[ Y_{t+1} | Y_{t+1} \le Q_{t,\alpha}(Y_{t+1}) \right],
	\end{align}
	where $Q_{t,\alpha}(Y_{t+1})$ denotes the conditional $\alpha$-quantile of $Y_{t+1}$ given $\mathcal{F}_t$.
	As discussed in the previous section, the main ingredient of forecast encompassing tests is the specification of the underlying loss function, which has to be associated with the risk measures we consider forecasts for.
	As such loss functions do not exist for the ES stand-alone, we utilize a strictly consistent joint loss function for the pair consisting of the ES and the VaR, given by \cite{Fissler2016} as
	\begin{align} 
	\begin{aligned}
	\label{eqn:JointLossESReg0Hom}
	\rho(Y,q_\alpha,e_\alpha) = - \frac{1}{e_\alpha} \left( e_\alpha - q_\alpha + \frac{(q_\alpha - Y) \mathds{1}_{\{Y \le q_\alpha\}}}{\alpha}  \right) + \log(-e_\alpha),
	\end{aligned}
	\end{align}	
	where the arguments $Y$, $q_\alpha$ and $e_\alpha$ denote the return realization, the quantile and the ES respectively.
	As this loss function exhibits the desirable property of having loss differences which are homogeneous of order zero, it is often denoted as the FZ0-loss function, see e.g. \cite{Patton2019}.
	While there exist infinitely many strictly consistent loss functions for the pair VaR and ES, the recent literature seems to agree upon this choice: \cite{DimiBayer2019} find that it exhibits a stable numerical performance in M-estimation and empirically yields relatively efficient parameter estimates.
	\cite{Nolde2017} discuss the desirable property of homogeneity of these loss functions and \cite{Patton2019}, \cite{BayerDimi2019} and \cite{Taylor2019} use this loss function to estimate dynamic ES models.

	Following the specification of a link function in (\ref{eqn:DefLinkFct}), we introduce the quantile- and ES-specific link functions
	\begin{alignat}{3}
	\label{eqn:DefLinkFctQES}
	&g^q: \mathfrak{Q} \times \Theta^\beta \to \mathbb{R}&, \qquad &(\bq, \beta) \mapsto g^q(\bq, \beta), \\
	&g^e: \mathfrak{E} \times \Theta^\eta \to \mathbb{R}&, \qquad &(\be, \eta) \mapsto g^e(\be, \eta),
	\end{alignat}
	where $ \mathfrak{Q}$ and $\mathfrak{E}$ denote the random spaces of the VaR and ES forecasts, $\Theta^\beta \subseteq \mathbb{R}^{k_\beta}$ and  $\Theta^\eta \subseteq \mathbb{R}^{k_\eta}$ such that $\Theta = \Theta^\beta \times \Theta^\eta$ and $k_\beta + k_\eta  = k \in \mathbb{N}$.
	We assume that the functions $g^q$, $g^e$ and the parameter space $\Theta$ are chosen such that there exist values $\beta_0 \in \Theta^\beta$ and $\eta_0 \in \Theta^\eta$, such that $g^q(\bq, \beta_0) = \hat q_{1,t}$ and $g^e(\be, \eta_0) = \hat e_{1,t}$ almost surely.
	
	In the following, we introduce the concept of joint forecast encompassing for the pair consisting of the VaR and the ES.
	Analogously to (\ref{eqn:ThetaStarDef}), we define the optimal combination parameters for the VaR and ES as
	\begin{align}
		\label{eqn:ThetaStarQESDef}
		\theta^\ast = (\beta^\ast, \eta^\ast)
		= \underset{(\beta, \eta)  \in \Theta}{\operatorname{arg \, min}} \; \mathbb{E} \left[ \rho \big(Y_{t+1}, g^q(\bq, \beta), g^e(\be, \eta) \big) \right].
	\end{align}
	
	\begin{definition}[Joint VaR and ES Forecast Encompassing]
		\label{def:QuantileESEncompassing}
		Let $\big( \hat q_{1,t}, \hat e_{1,t}\big)$ and $\big( \hat q_{2,t}, \hat e_{2,t}\big)$ denote pair-wise competing forecasts for the pair consisting of the conditional quantile and ES of $F_t$.
		We say that $\big( \hat q_{1,t}, \hat e_{1,t}\big)$  encompasses $\big( \hat q_{2,t}, \hat e_{2,t}\big)$  at time $t$ with respect to the link functions $g^q$ and $g^e$ if and only if
		\begin{align}
		\label{eqn:UncondEncompVaRES}
		\mathbb{E} \left[ \rho \big( Y_{t+1}, \hat q_{1,t}, \hat e_{1,t} \big) \right]
		= \mathbb{E} \left[ \rho  \big( Y_{t+1}, g^q(\bq, \beta^\ast), g^e(\be, \eta^\ast) \big) \right],
		\end{align}
		where the loss function $\rho$ is given in (\ref{eqn:JointLossESReg0Hom}).
		This holds if and only if $\big( \beta^\ast, \eta^\ast \big) = \big( \beta_0, \eta_0 \big)$.
	\end{definition}
	
	We test whether the sequence of joint quantile and ES forecasts $\big( \hat q_{1,t}, \hat e_{1,t}\big)$ encompasses the sequence  $\big( \hat q_{2,t}, \hat e_{2,t}\big)$ for all $t=m,\dots,T-1$ by estimating the parameters of the following semiparametric regression,
	\begin{align}
	\begin{aligned}
	\label{eqn:JointRegressionJointESEncTest}
	Y_{t+1} = g^q(\bq, \beta) + u_{t+1}^q, \qquad \text{ and } \qquad
	Y_{t+1} = g^e(\be, \eta) + u_{t+1}^e,
	\end{aligned}
	\end{align}
	where $Q_\alpha(u_{t+1}^q|\mathcal{F}_t) = 0$ and $ES_\alpha(u_{t+1}^e|\mathcal{F}_t) = 0$ almost surely  for all $t=m,\dots,T-1$ by using the M-estimation technique introduced in \cite{Patton2019} and \cite{DimiBayer2019}.
	We then test for $\big( \beta^\ast, \eta^\ast \big) = \big( \beta_0, \eta_0 \big)$ using a Wald type test statistic.
	
	Definition \ref{def:QuantileESEncompassing} develops a \textit{joint} encompassing test for the VaR and ES, which is reasonable given the joint elicitability property of the VaR and ES.
	However, a further objective of this paper is to construct encompassing tests for the ES stand-alone, which we do in the following.
	\begin{definition}[Auxiliary ES Forecast Encompassing]
		\label{def:AuxESEncompassingTest}
		Let $\big( \hat q_{1,t}, \hat e_{1,t}\big)$ and $\big( \hat q_{2,t}, \hat e_{2,t}\big)$ denote competing forecasts for the pair consisting of the conditional quantile and ES of $F_t$.
		We say that $\hat e_{1,t}$ \textit{auxiliarily encompasses} $\hat e_{2,t}$ at time $t$ with respect to the link functions $g^q$ and $g^e$ if and only if
		\begin{align}
			\label{eqn:AuxEncompassing}
		\mathbb{E} \left[ \rho \big( Y_{t+1}, g^q(\bq, \beta^\ast), \hat e_{1,t} \big) \right]
		= \mathbb{E} \left[ \rho \big( Y_{t+1}, g^q(\bq, \beta^\ast), g^e(\be, \eta^\ast) \big) \right],
		\end{align}
		that is, if and only if $\eta^\ast = \eta_0$.\footnote{It is important to notice that the optimal combination parameter $\beta^\ast$ on the left-hand side of (\ref{eqn:AuxEncompassing}) is given by (\ref{eqn:ThetaStarQESDef}) and not in the sense of an optimal combination parameter of a \textit{restricted} model.}
	\end{definition}
	This parameter restriction is tested using a Wald type test statistic based on the estimates of the regression setup given in (\ref{eqn:JointRegressionJointESEncTest}).
	As we do not test the quantile specific parameters $\beta^\ast$, we do not impose that the underlying quantile forecast also encompasses its competitor under this null hypothesis.
	Hence, even though this test is based on the joint regression, it only tests encompassing of the ES forecasts.
	We call this test \textit{auxiliary ES encompassing test} as it still depends on the auxiliary quantile forecasts which are used for the estimation of the optimal combination parameters.
	
	Given that both, the VaR and ES forecasts are available, application of either the joint or auxiliary test is the most plausible approach given their joint elicitability.
	However, even though the emphasis of the auxiliary encompassing test is on the ES, it still requires quantile forecasts for the implementation of the parameter estimation.
	This can be problematic for two reasons.
	First, the quantile forecasts are still used in the estimation procedure and thus have an indirect effect on the parameter estimates of the ES specific parameters.
	E.g., the previous tests are not applicable for ES forecasts which are based on the same VaR forecasts, as this implies perfect collinearity of the quantile regressors.
	Second, the auxiliary test is only applicable in the setup where the person applying the test has access to the quantile forecasts.
	In the current implementation of the regulatory framework of the Basel Committee \citep{Basel2016, Basel2017}, the banks are only obligated to report their ES forecasts (at probability level $2.5\%$), but not the corresponding VaR forecasts.
	Thus, the accompanying VaR forecasts, which the ES forecasts are internally based on, are in general not available to the regulator who has to decide on an adequate risk management of the financial institution at hand. 
	
	In order to account for these scenarios, we further introduce the \textit{strict ES encompassing test}, which only requires ES forecasts in the following.
	For this, we slightly modify the definition of (\ref{eqn:ThetaStarQESDef}) by replacing $g^q(\bq, \beta)$ through $g^q(\be, \beta)$,\footnote{Note that the parameters denoted $\theta^\ast$ in (\ref{eqn:ThetaStarQESDef}) and in (\ref{eqn:ThetaStarStrictESDef}) can generally differ.}
	\begin{align}
		\label{eqn:ThetaStarStrictESDef}
		\theta^\ast = (\beta^\ast, \eta^\ast)
		= \underset{(\beta, \eta)  \in \Theta}{\operatorname{arg \, min}} \; \mathbb{E} \left[ \rho \big(Y_{t+1}, g^q(\be, \beta), g^e(\be, \eta) \big) \right].
	\end{align}
	\begin{definition}[Strict ES Forecast Encompassing]
		\label{def:StrictESEncompassing}
		Let $\hat e_{1,t}$ and $\hat e_{2,t}$ denote competing ES forecasts of the underlying predictive distribution $F_t$.
		We say that $\hat e_{1,t}$ strictly encompasses $ \hat e_{2,t}$ at time $t$ with respect to the link functions $g^q$ and $g^e$ if and only if
		\begin{align}
		\label{eqn:StrictEncompassing}
		\mathbb{E} \left[ \rho \big( Y_{t+1}, g^q(\be, \beta^\ast), \hat e_{1,t} \big) \right]
		= \mathbb{E} \left[ \rho \big( Y_{t+1}, g^q(\be, \beta^\ast), g^e(\be, \eta^\ast) \big) \right],
		\end{align}
		that is, if and only if $\eta^\ast = \eta_0$.
	\end{definition}
	
	We test whether $\hat e_{1,t}$ strictly encompasses $\hat e_{2,t}$ for all $t=m,\dots,T-1$ by setting up the slightly transformed regression
	\begin{align}
	\begin{aligned}
	\label{eqn:JointRegressionMisspecifiedESEncTest}
	Y_{t+1} = g^q(\be, \beta) + u_{t+1}^q, \qquad \text{ and } \qquad
	Y_{t+1} = g^e(\be, \eta) + u_{t+1}^e,
	\end{aligned}
	\end{align}
	where $Q_\alpha(u_{t+1}^q|\mathcal{F}_t) = 0$ and $ES_\alpha(u_{t+1}^e|\mathcal{F}_t) = 0$ almost surely  for all $t=m,\dots,T-1$.
	The crucial difference between this test and the joint and auxiliary encompassing tests is that instead of using the quantile forecasts $\bq$ in the quantile link function $g^q$, we use the ES forecasts $\be$ for both, the quantile and ES link functions $g^q$ and $g^e$.
	We argue that this can be seen as a best feasible solution due to the lack of loss functions for the ES stand-alone together with the necessity of developing forecast evaluation methods for the ES stand-alone due to the current setup of the Basel III regulatory framework \citep{Basel2016, Basel2017}.
	
	The underlying idea of this test is mainly motivated by pure scale models, i.e.\ $Y_{t+1} = \sigma_{t} u_{t+1}$, $u_{t+1} \sim F(0,1)$, which is still the most frequently used class of models for risk management with the GARCH and stochastic volatility models as prime examples.
	For this model class, the VaR and ES forecasts are perfectly colinear, $\hat e_t = \frac{\xi_\alpha}{z_\alpha} \hat q_t$, where $z_\alpha$ and $\xi_\alpha$ are the $\alpha$-quantile and $\alpha$-ES of the distribution $F(0,1)$.
	Hence, the quantile model $g^q(\be, \beta) = g^q(\bq \xi_\alpha / z_\alpha, \beta) = g^q(\bq, \tilde \beta)$ is correctly specified, but with transformed quantile parameters $\tilde \beta$.\footnote{For the prominent case of linear encompassing link formulas $g^q(\cdot)$, it holds that $\tilde \beta = \beta z_\alpha / \xi_\alpha $.}
	As we only test on the ES-specific parameters $\eta$ as described in Definition \ref{def:StrictESEncompassing}, our test is invariant to this (often linear) transformation of the parameter $\beta$ and thus, it is correctly specified for pure scale models.

	In the general case, the quantile equation can possibly be misspecified.
	Thus, we provide asymptotic theory under general model misspecification for the M-estimator in the following section.
	The potential model misspecification might bias the pseudo-true parameters and challenge the interpretability of the test decision, but we argue that this effect is negligible for this setup.
	First, the misspecification is only \textit{slight} in the sense that daily financial return data is approximated well by pure scale processes.
	Second, the misspecification is \textit{indirect} in the sense that while the quantile parameters are potentially misspecified, we only test the ES parameters, which are influenced by the misspecification only indirectly through the joint estimation.
	Furthermore, we illustrate that the performance of our strict ES encompassing test is not negatively influenced by more general data generating processes in the simulation study in Section \ref{sec:SimulationStudy} by considering GAS models with time-varying higher moments of \cite{Creal2013} and the dynamic ES models of \cite{Patton2019} and \cite{Taylor2019}.

	Tests for equal (superior) predictive ability in the sense of \cite{DieboldMariano1995}, \cite{clark2001tests}, \cite{GiacominiWhite2006}, \cite{West2006} and the model confidence set approach of \cite{Hansen2011} can be seen as a general alternative to encompassing tests.
	As these tests are directly based on the average loss difference, they can only test the predictive ability of the VaR and ES jointly.
	In contrast, encompassing tests are based on the regression coefficients of the semiparametric quantile and ES models and hence, only indirectly on the respective loss function. 
	This fundamental difference allows for stand-alone encompassing tests for ES forecasts, which constitutes a great advantage for ES encompassing tests.
	
	Strictly speaking, strict consistency of loss functions only implies that the \textit{optimal} forecast exhibits the \textit{smallest possible} loss in expectation.
	In reality however, competing forecasts are often misspecified due to estimation error or misspecified forecasting models.
	\cite{Patton2019MisspecFC} shows that then, the ranking induced by the loss functions can be sensitive towards the choice of (strictly consistent) loss functions or even misleading.
	\cite{Holzmann2014} show that for competing forecasts which are based on nested information sets and which are correctly specified given their underlying (but usually incomplete) information set (auto-calibrated), applying any strictly consistent loss function results in a correct ranking of the forecasts.
	In our case of testing forecast encompassing, we indeed build on nested information sets as it obviously holds that $\sigma \big\{ \hat f_{1,t},  \hat f_{2,t} \big\} \supseteq \sigma \big\{ \hat f_{1,t} \big\}$.
	Thus, by further assuming that the issued forecasts are auto-calibrated given the forecaster's information set, we can conclude that the ranking implied by (\ref{eqn:GeneralFcEncpConditionH0}) is indeed the correct one and invariant towards the choice of strictly consistent loss functions.

	\subsection{Asymptotic Theory under Model Misspecification}
	\label{sec:AsymptoticTheoryMisspec}
	
	In the following, we use the short notation $g^e_t(\eta) = g^e(\be, \eta)$  and $g^q_t(\beta) = g^q(\bq, \beta)$ (or $g^q_t(\beta) = g^q(\be, \beta)$ in the case of the strict test).
	We define the M-estimator as
	\begin{align}
	\label{eqn:DefQn}
	\hat \theta_n = \underset{\theta \in \Theta}{\argmin} \, Q_n(\theta), \quad \text{ where }  \quad
	Q_n(\theta) = \frac{1}{n} \sum_{t=m}^{T-1} \rho \big( Y_{t+1}, g^q_t(\beta) , g^e_t(\eta)  \big),
	\end{align}
	and the pseudo-true parameter as\footnote{ The pseudo-true parameter can generally depend on the issued loss function, i.e. in this case on the zero homogeneous choice in (\ref{eqn:JointLossESReg0Hom}).}
	\begin{align}
	\label{eqn:DefQn0}
	\theta^\ast_n =  \underset{\theta \in \Theta}{\argmin} \, Q_n^0(\theta), \quad \text{ where }  \quad
	Q_n^0(\theta) = \frac{1}{n} \sum_{t=m}^{T-1} \mathbb{E} \left[ \rho \big( Y_{t+1}, g^q_t(\beta) , g^e_t(\eta)  \big) \right].
	\end{align} 
	When the link (regression) functions $g^q(\cdot)$ and $g^e(\cdot)$ are correctly specified, we get that the pseudo-true parameter $\theta^\ast_n$ equals the classical true regression parameter and it is independent of the sample size $n$.
	We further define the corresponding identification functions, which are almost surely the derivative of the loss function $\rho$ with respect to $\theta$,
	\begin{align}
	\label{eqn:DefPsiEncmp}
	\psi \big( Y_{t+1}, g^q_t(\beta) , g^e_t(\eta)  \big)
	= 
	\begin{pmatrix}
	- \frac{\nabla g^q_t(\beta)}{\alpha g^e_t(\eta)} \left( \mathds{1}_{\{Y_{t+1} \le g^q_t(\beta) \}} - \alpha \right) \\
	\frac{\nabla g^e_t(\eta) }{g^e_t(\eta)^2}  \left( g^e_t(\eta) - g^q_t(\beta)+ \frac{1}{\alpha} (g^q_t(\beta) - Y_{t+1}) \mathds{1}_{\{Y_{t+1} \le g^q_t(\beta) \}} \right)
	\end{pmatrix}.
	\end{align}
	We restrict our attention to processes which satisfy the following conditions.
	\begin{assumption}
		\label{assu:AsymptoticTheory}
		We assume that 
		\begin{enumerate}[label=(\alph*)]
			\item
			\label{cond:AlphaMixing}
			the process $Z_t$ is strong mixing of size $- r/(  r-2)$ for some $ r>2$,
			
			\item 
			\label{cond:ParameterSpace}
			the parameter space $\Theta = \Theta^\beta \times \Theta^\eta \subseteq \mathbb{R}^k$ is compact and non-empty,		
			
			\item 
			\label{cond:UniqueMinimum}
			the pseudo-true parameter $\theta_n^\ast$ defined in (\ref{eqn:DefQn0}) is in the interior of $\Theta$ and is the unique minimizer of the objective function $Q_n^0(\theta)$ and the sequence $\mathbb{E}_t \big[ \psi \big( Y_{t+1}, g^q_t(\beta) , g^e_t(\eta) \big) \big]$, defined in (\ref{eqn:DefPsiEncmp}) is uncorrelated,
			
			\item 
			\label{cond:AbsContDistribution}
			the distribution of $Y_{t+1}$ given $\mathcal{F}_t$, denoted by $F_t$ is absolutely continuous with continuous and strictly positive density $h_t$, which is bounded from above almost surely on the whole support of $F_t$ and Lipschitz continuous,
			
			\item 
			\label{cond:ESModelBounded}
			for all $\theta$ in a neighborhood of $\theta_n^\ast$, it holds that $\left| \frac{1}{g^e_t(\eta) } \right| \le K < \infty$ for some constant $K > 0$, 
			
			\item 
			\label{cond:FullRankConditionNormality}
			the link functions $g_t^q(\beta)$ and $g_t^e(\eta)$ are $\mathcal{F}_t$-measurable, twice continuously differentiable in $\theta = (\beta,\eta)$ on $\operatorname{int}(\Theta)$ almost surely and if $\mathbb{P} \big( g_t^q(\beta_1) = g_t^q(\beta_2) \cap g_t^e(\eta_1) = g_t^e(\eta_2) \big) = 1$, then $\theta_1 = \theta_2$,
			
			\item 
			\label{cond:AsyCovPositiveDefinite}
			the matrices $\Lambda_n$ and $\Sigma_n$, defined in Proposition \ref{prop:AsymptoticNormality} are positive definite with a determinant bounded away from zero for all $n$ sufficiently large,

			\item 
			\label{cond:MomentCondition}
			it holds that
			$g_t^q(\beta) \le Q$, $\nabla g_t^q(\beta) \le Q_1$, $H_t^q(\beta) \le Q_2$, $\nabla H_t^q(\beta) \le Q_3$, and 
			$g_t^e(\eta) \le E$, $\nabla g_t^e(\eta) \le E_1$, $H_t^e(\eta) \le E_2$, $\nabla H_t^e(\eta) \le E_3$, for all $\theta$ in a neighborhood of $\theta_n^\ast$, where the random variables $Q, E, Q_1, E_1, Q_2, E_2, Q_3, E_3$ are all $\mathcal{F}_t$-measurable and for some $r>2$ (from condition \ref{cond:AlphaMixing}), the following moments are bounded
			(i) $\mathbb{E} [ Q_1^{r+1} ]$, 
			(ii) $\mathbb{E} [ E_1^{r+1} ]$, 
			(iii) $\mathbb{E} [ Q_2^{(r+1)/2} ]$, 
			(iv) $\mathbb{E} [ E_2^{(r+1)/2} ]$, 
			(v) $\mathbb{E} [ E_1 Q_2 ]$, 
			(vi) $\mathbb{E} [ Q_1 Q_2 ]$, 
			(vii) $\mathbb{E} [ Q_1 E_2 ]$, 
			(viii) $\mathbb{E} [ Q_1^2 E_1 ]$, 
			(ix) $\mathbb{E} [ E E_1^3 ]$, 
			(x) $\mathbb{E} [ E E_3 ]$, 
			(xi) $\mathbb{E} [ E E_1 E_2 ]$, 
			(xii) $\mathbb{E} [ Q E_1 E_2 ]$, 
			(xiii) $\mathbb{E} [ Q E_1^3 ]$, 
			(xiv) $\mathbb{E} [  Q_1 Q^{r} E_1^{r}]$,  
			(xv) $\mathbb{E} [  E_1^{r-1} E_2 |Y_t]^{r} ]$, 
			(xvi) $\mathbb{E} [  E_1^{r+1} |Y_t]^{r} ]$,
			(xvii) $\mathbb{E} [ Y_t^{2r}]$,

			\item 
			\label{cond:ExactQuantileMatches}
			for any $n$, the term $\sup_{\beta \in \Theta^\beta} \sum_{t=m}^{T-1} \mathds{1}_{\{ Y_{t+1} = g^q_t(\beta)  \}}$ is almost surely bounded from above.
			
		\end{enumerate}
	\end{assumption}
	
	The following propositions show consistency and asymptotic normality of the M-estimator under potential model misspecification.
	\begin{proposition}
		\label{prop:Consistency}
		Given the conditions in Assumption \ref{assu:AsymptoticTheory}, it holds that $\hat \theta_n - \theta_n^\ast \stackrel{\mathbb{P}}{\longrightarrow} 0$.
	\end{proposition}

	\begin{proposition}
		\label{prop:AsymptoticNormality}
		Given the conditions in Assumption \ref{assu:AsymptoticTheory}, it holds that
		\begin{align}
		\Omega_n^{-1/2} (\theta^\ast_n) \sqrt{n} \big( \hat \theta_n - \theta^\ast_n \big) \stackrel{d}{\longrightarrow} \mathcal{N}(0,I_k),
		\end{align}
		with $\Omega_n(\theta^\ast_n) =  \Lambda_n^{-1}(\theta^\ast_n) \, \Sigma_n(\theta^\ast_n) \, \Lambda_n^{-1}(\theta^\ast_n)$,
		where
		$\Lambda_n(\theta^\ast_n) = \begin{pmatrix} \Lambda_{n,qq}(\theta^\ast_n) & \Lambda_{n,qe}(\theta^\ast_n) \\ \Lambda_{n,eq}(\theta^\ast_n) & \Lambda_{n,ee}(\theta^\ast_n) \end{pmatrix}$,
		and $\Sigma_n(\theta^\ast_n) =  \frac{1}{n} \sum_{t=m}^{T-1} \mathbb{E} \left[ \psi \big( Y_{t+1}, g^q_t( \beta^\ast_n) , g^e_t( \eta^\ast_n)  \big) \cdot \psi \big( Y_{t+1}, g^q_t( \beta^\ast_n) , g^e_t( \eta^\ast_n)  \big)^\top \right]$.
		Furthermore, the components of $\Lambda_n(\theta^\ast_n)$ are given by
		\begin{align}
		\Lambda_{n,qq}(\theta^\ast_n) &= - \frac{1}{n} \sum_{t=m}^{T-1} \mathbb{E}  \left[ 
		\frac{H^q_{t} ( \beta^\ast_n)}{\alpha g_t^e(\eta^\ast_n)} \big(F_t(g_t^q(\beta^\ast_n)) - \alpha \big)
		+ \frac{ \nabla g_t^q(\beta^\ast_n) \nabla g_t^q(\beta^\ast_n)^\top  }{\alpha g_t^e(\eta^\ast_n)} {h_t}(g_t^q( \beta^\ast_n)) \right], \\
		\Lambda_{n,qe}(\theta^\ast_n) &= \Lambda_{n,eq}(\theta^\ast_n)^\top = \frac{1}{n} \sum_{t=m}^{T-1} \mathbb{E} \left[ 
		\frac{  \nabla g_t^q(\beta^\ast_n) \nabla g_t^e(\eta^\ast_n)^\top }{\alpha g_t^e(\eta^\ast_n)^2} \big(F_t(g_t^q(\beta^\ast_n)) - \alpha \big) \right], \\
		\Lambda_{n,ee}(\theta^\ast_n) &= \frac{1}{n} \sum_{t=m}^{T-1} \mathbb{E}   \left[ 
		\frac{\nabla g_t^e(\eta^\ast_n) \nabla g_t^e(\eta^\ast_n)^\top}{g^e_t( \eta^\ast_n)^2}  
		+ \left( \frac{H^e_{t}( \eta^\ast_n)}{g^e_t( \eta^\ast_n)^2}  -2 \frac{\nabla g_t^e(\eta^\ast_n) \nabla g_t^e(\eta^\ast_n)^\top}{g^e_t( \eta^\ast_n)^3} \right) \right. \times \\
		&\qquad \qquad \quad \left. \left( g^e_t( \eta^\ast_n)  - g^q_t( \beta^\ast_n)+ \frac{1}{\alpha}  (g^q_t( \beta^\ast_n)- Y_{t+1}) \mathds{1}_{\{Y_{t+1} \le g^q_t( \beta^\ast_n) \}} \right) \right],
		\end{align}
		where $H^q_{t} (\beta)$  and $H^e_{t} ( \eta)$ are the Hessian matrices of $g_t^q(\beta)$ and $g_t^e(\eta)$ respectively.
	\end{proposition}
	The two preceding propositions extend the asymptotic theory of \cite{Patton2019} to the case of possibly misspecified models, and the misspecification theory for linear models of \cite{BayerDimi2019} to nonlinear models.
	The proofs in Appendix \ref{sec:Proofs} combine, extend and go along the lines of the ideas of \cite{Engle2004} and \cite{Patton2019}.
	The conditions closely resemble the regularity conditions of \cite{Patton2019}.
	As we further allow for model misspecification, we impose the unique minimization condition \ref{cond:UniqueMinimum} and slightly strengthen the moment conditions \ref{cond:MomentCondition}.
	In the baseline case of linear encompassing link functions $g^q$ and $g^e$, the required moment conditions simplify to those given in \cite{BayerDimi2019}. 

	For the estimation of the asymptotic covariance matrix $\widehat{\Omega}_{n}$ under possible model misspecification, we follow the approach of \cite{DimiBayer2019} and \cite{BayerDimi2019}.
	We deal with the three nuisance quantities in $\widehat{\Omega}_{n}$ as follows.
	In order to estimate the density quantile function $h_t(g_t^q(\beta^\ast_n))$, we follow the \textit{nid}-estimator of \cite{Hendricks1992}. 
	As the degree of misspecification in the investigated financial time series is small \citep{BayerDimi2019}, we  approximate $F_t(g_t^q(\beta^\ast_n)) \approx \alpha$.
	For the conditional truncated variance, $\Var_t \big( g^q_t( \beta^\ast_n)- Y_{t+1} \big| Y_{t+1} \le g^q_t( \beta^\ast_n) \big)$, we employ the \textit{scl-sp} estimator of \cite{DimiBayer2019}.
	
	We now consider the asymptotic distributions of our three ES encompassing tests proposed in Section \ref{sec:EncompassingTestsES}  under the null hypotheses and for general link functions, where we test certain $s$-dimensional ($s \in \mathbb{N}, s \le k$) sub-vectors of $\theta$.
	For this, let $R \in \mathbb{R}^{k \times s}$ be a selection matrix whose columns consist of $k$-dimensional Cartesian unit (column) vectors $e_j \in \mathbb{R}^k$, which are zero apart from a one in dimension $j$.
	E.g., when $g_t^q(\beta)$ and $g_t^e(\eta)$ equal the linear link functions with intercept, given in the first point of Example \ref{exmpl:LinkFunctions}, $\theta = (\beta_1, \beta_2, \beta_3, \eta_1, \eta_2, \eta_3)$.
	Then, for the strict and auxiliary ES encompassing tests, $R = (e_5, e_6)$ and for the joint test $R = (e_2,e_3,e_5, e_6)$.
	These choices pick the respective parameters from $\theta$.
	Then, we define the respective test statistics by
	\begin{align}
		Z_{R,n} &= n \big( \hat \theta_n R - \theta_n^\ast R \big)  \, R^\top \widehat{\Omega}_{n}^{-1} R \,  \big( \hat \theta_n R - \theta_n^\ast R \big)^\top.
	\end{align}
	
	\begin{theorem}[ES Encompassing Tests]
		\label{thm:DistributionEncmpTestStatistics}
		Given the conditions of Assumption \ref{assu:AsymptoticTheory} and given that $\widehat{\Omega}_{n} - \Omega_{n} \stackrel{\mathbb{P}}{\to} 0$, under the respective null hypotheses given in Definition \ref{def:QuantileESEncompassing} - \ref{def:StrictESEncompassing}, it holds that
		\begin{align}
			Z_{R,n} \stackrel{d}{\longrightarrow} \chi^2_s.
		\end{align}
	\end{theorem}
	For linear link functions, this theorem implies that the limiting $\chi^2$ distribution of the joint test has four degrees of freedom, while the one of the strict and auxiliary tests has two degrees of freedom.

	An important application of these ES encompassing tests is in the context of selecting the best-performing forecast,  i.e.\ selecting at time $T$ a superior forecasting method for the future.
	This is particularly relevant as the ES is recently introduced into the Basel regulations without having proper forecast selection procedures at hand.
	Following \cite{GiacominiKomunjer2005}, we propose the following decision rule.
	We test the two encompassing hypotheses $\mathbb{H}_{0}^{(1)}$: $\hat e_{1,t}$ encompasses $\hat e_{2,t}$ and  $\mathbb{H}_{0}^{(2)}$: $\hat e_{2,t}$ encompasses $\hat e_{1,t}$ for $t = m,\dots,T-1$.
	Then, there are four possible scenarios: 
	(1) if neither $\mathbb{H}_{0}^{(1)}$ nor $\mathbb{H}_{0}^{(2)}$ are rejected, the test is not helpful for forecast selection.
	(2) If $\mathbb{H}_{0}^{(1)}$ is rejected while $\mathbb{H}_{0}^{(2)}$ is not rejected, we can conclude that forecast $\hat e_{2,t}$ does add information to forecast $\hat e_{1,t}$, while we cannot conclude the reverse. 
	Thus, we decide to use the forecasting method of $\hat e_{2,t}$.
	(3) If $\mathbb{H}_{0}^{(2)}$ is rejected while $\mathbb{H}_{0}^{(1)}$ is not rejected, the same logic applies inversely and we use the forecasting method of $\hat e_{1,t}$.
	(4) If both, $\mathbb{H}_{0}^{(1)}$and $\mathbb{H}_{0}^{(2)}$are rejected, the test delivers statistical evidence that both forecasts contain exclusive information and that a forecast combination outperforms the stand-alone forecasts.
	Consequently, we use a combined forecast $\hat e_{c,t} = g^e(\be, \hat \eta_n)$
	where the estimated combination parameters $\hat \eta_{n}$ are obtained from the M-estimator proposed here.

	Testing forecast encompassing \textit{conditional}  on some information set $\tilde{\mathcal{G}}_t = \sigma\{ \bW \}$ based on some $\mathcal{F}_t$-measurable vector of instruments $\bW$ in the sense of \cite{GiacominiKomunjer2005} can be facilitated through estimating the regression parameters through (overidentified) GMM-estimation instead of M-estimation.
	However, for the strict ES test, this approach requires asymptotic theory under model misspecification for the overidentified GMM estimator based on \textit{nonsmooth} objective functions.
	While such theory is available for smooth moment conditions (see e.g. \cite{HallInoue2003} and \cite{HansenLee2019}), its generalization to nonsmooth objective functions is not straight-forward and thus, we leave \textit{conditional} ES encompassing tests based on misspecified GMM-estimation for future research.
	The moment conditions of our unconditional approach can be interpreted as \textit{conditional} encompassing with respect to the instruments $\nabla g^q_t(\beta)$ and $\nabla g^e_t(\eta)$.
	In the classical baseline case of linear forecast encompassing, these instruments simplify to $\bq$ and $\be$ and thus, our approach tests \textit{conditional} encompassing with respect to the information set  $\mathcal{G}_t = \sigma\{ 1, \bq, \be \} \subseteq \mathcal{F}_t$, which in most cases already captures the most relevant information which is available.

	\section{Simulation Study}
	\label{sec:SimulationStudy}
	
	In this section, we evaluate the size and power properties of our three proposed ES encompassing tests and compare them to the VaR encompassing test of \cite{GiacominiKomunjer2005}.
	For this, we describe the simulation setup in Section \ref{sec:DGPs} and we report and discuss the simulation results in Section \ref{sec:SimulationResults}.
	Section \ref{sec:SimulationExtensions} considers three extensions of the simulation setup with respect to additional data generating processes (DGPs), loss and link functions.

	\subsection{The Simulation Setup}
	\label{sec:DGPs}

	We employ the three encompassing tests based on the linear link functions
	$g^q( \bfc, \beta) = \beta_1 + \beta_2 \hat f_{1,t} + \beta_3 \hat f_{2,t}$ and $g^e( \be, \eta) = \eta_1 + \eta_2 \hat e_{1,t} + \eta_3 \hat e_{2,t}$, where $\bfc = \bq$ for the joint and auxiliary tests and $\bfc = \be$ for the strict test, together with the parameter space $\Theta = \{ \theta = (\beta, \eta) \in \mathbb{R}^6: ||\theta|| \le K \}$.\footnote{ We choose the constant $K$ large enough such that the parameter estimation is not restricted in realistic settings but the parameter space $\Theta$ is indeed convex.}
	For the respective encompassing tests, in each case we test the following two opposing hypotheses:
	\begin{alignat}{5}
		\begin{aligned}
		\label{eqn:SimStudyHypotheses}
		&\text{Joint}:& \; &\mathbb{H}_0^{(1)}: (\beta_2^\ast, \beta_3^\ast, \eta_2^\ast, \eta_3^\ast) = (1,0,1,0),& \;
		&\mathbb{H}_0^{(2)}: (\beta_2^\ast, \beta_3^\ast, \eta_2^\ast, \eta_3^\ast) = (0,1,0,1), \\
		&\text{Str \& Aux}:& \;  &\mathbb{H}_0^{(1)}: (\eta_2^\ast, \eta_3^\ast) = (1,0),& \; 
		&\mathbb{H}_0^{(2)}: (\eta_2^\ast, \eta_3^\ast) = (0,1), \\
		&\text{VaR}:& \; &\mathbb{H}_0^{(1)}: (\beta_2^\ast, \beta_3^\ast) = (1,0),& \; 
		&\mathbb{H}_0^{(2)}: (\beta_2^\ast, \beta_3^\ast) = (0,1).
		\end{aligned}
	\end{alignat}

	In the following, we describe two DGPs where for the first, both forecasting models stem from classical GARCH models  while the second considers two joint GAS models for the VaR and ES of \cite{Patton2019}.
	For both model classes, we simulate data as a convex combination of two distinct models with a flexible convex combination weight $\pi \in [0,1]$.
	This implies that for $\pi = 0$, the first model encompasses the second, while for $\pi=1$, the inverse holds.
	For all intermediate parameters $\pi \in (0,1)$, the data stems from a linear combination and both forecast encompassing null hypotheses should be rejected which indicates that a forecast combination method is preferred.

\subsection*{The GARCH DGP}

	The two GARCH models, which are calibrated to daily IBM returns, are given by $\tilde Y_{j,t+1} = \hat \sigma_{j,t} u_{t+1}$, for $j=1,2$, where $u_{t+1} \stackrel{iid}{\sim} \mathcal{N}(0,1)$ and the two distinct volatility specifications are given by
	\begin{align}
	\hat \sigma_{1,t}^2 &= 0.042 + 0.053 \tilde Y_{1,t}^2 + 0.925 \hat \sigma_{1,t-1}^2, \qquad \text{and} \label{eqn:GARCHModel} \\
	\hat \sigma_{2,t}^2 &= 0.044 + \big(0.024 + 0.058 \cdot \mathds{1}_{\{ \tilde Y_{2,t} \le 0 \}} \big) \tilde Y_{2,t}^2 + 0.923 \hat \sigma_{2,t-1}^2.
	\label{eqn:GJRGARCHModel}
	\end{align}
	For both models, we obtain VaR and ES forecasts by $\hat q_{j,t} = z_\alpha \hat \sigma_{j,t}$ and $\hat e_{j,t} = \xi_\alpha \hat \sigma_{j,t}$, for $j=1,2$, where $z_\alpha$ and $\xi_\alpha$ are the $\alpha$-quantile and $\alpha$-ES of the standard normal distribution.
	Notice that the time index $t$ on $\hat \sigma_{j,t}$ indicates that it is a $\mathcal{F}_t$-measurable forecast for 
	time $t+1$.
	While the first specification in (\ref{eqn:GARCHModel}) is a classical GARCH(1,1) model \citep{Bollerslev1986}, the second specification in (\ref{eqn:GJRGARCHModel}) follows the GJR-GARCH model of \cite{Glosten1993}, which allows for a leverage effect.
	We simulate data from the convex combination of these processes, $Y_{t+1} = \big( (1-\pi) \hat \sigma_{1,t} + \pi \hat \sigma_{2,t} \big) u_{t+1}$ for 21 equally spaced values of $\pi \in [0,1]$, where $u_{t+1} \stackrel{iid}{\sim} \mathcal{N}(0,1)$. 

 \subsection*{The VaR/ES GAS DGP}

	In the second simulation setup, we implement the one-factor (1F) and two-factor (2F) GAS models for the VaR and ES of \cite{Patton2019}.
	The 1F-GAS model evolves as
	\begin{align}
	\hat q_{1,t} &= -1.164 \exp(\hat \kappa_{t}) \qquad \text{ and } \qquad 
	\hat e_{1,t} = -1.757 \exp(\hat \kappa_{t}), \quad \text{ where }\\
	\hat \kappa_{t} &=  0.995 \hat \kappa_{t-1} + \frac{0.007}{\hat e_{1,t-1}} \left( \frac{\tilde Y_{1,t}}{\alpha}  \mathds{1}_{\{\tilde Y_{1,t} \le \hat q_{1,t-1} \}}  -  \hat e_{1,t-1} \right).
	\end{align}
	The 2F-GAS model follows the specification
	\begin{align}
	\begin{pmatrix} \hat q_{2,t} \\ \hat e_{2,t}  \end{pmatrix}
	= \begin{pmatrix} -0.009 \\ -0.010  \end{pmatrix}
	+ \begin{pmatrix} 0.993 & 0 \\ 0 & 0.994  \end{pmatrix}
	\begin{pmatrix} \hat q_{2,t-1} \\ \hat e_{2,t-1}  \end{pmatrix}
	+ \begin{pmatrix} -0.358 & -0.351 \\ -0.003 & -0.003  \end{pmatrix} \lambda_t,
	\end{align}
	where the forcing variable is given by $\lambda_t = \big( \hat q_{2,t-1} ( \alpha - \mathds{1}_{\{\tilde Y_{2,t} \le \hat q_{2,t-1} \}} ) , \, \mathds{1}_{\{\tilde Y_{2,t} \le \hat q_{2,t-1} \}}  \tilde Y_{2,t}  / \alpha - \hat e_{2,t-1} \big)^\top$.
	For both models, $j=1,2$, we simulate $\tilde Y_{j,t+1}  \sim \mathcal{N} \big( \hat \mu_{j,t}, \hat \sigma_{j,t}^2 \big)$, where the conditional mean and standard deviations are given by
	$\hat \mu_{j,t} = \hat q_{j,t} - z_\alpha \frac{\hat e_{j,t} -  \hat q_{j,t}}{\xi_\alpha - z_\alpha}$ and $\hat \sigma_{j,t} = \frac{\hat e_{j,t} -  \hat q_{j,t}}{\xi_\alpha - z_\alpha}$,
	such that $Q_\alpha(\tilde Y_{j,t+1}|\mathcal{F}_t) = \hat q_{j,t}$ and  $ES_\alpha(\tilde Y_{j,t+1}|\mathcal{F}_t) = \hat e_{j,t}$ almost surely.
	The parameter values for this model are obtained from Table 8 of \cite{Patton2019} and correspond to calibrated parameters to daily S\&P 500 returns.
	
	In order to simulate returns which follow a convex combination of these two conditional distributions, we simulate Bernoulli draws $\pi_{t+1} \sim \operatorname{Bern}(\pi)$ for 21 equally spaced values of $\pi \in [0,1]$, and let $Y_{t+1}  = (1-\pi_{t+1}) \tilde Y_{1,t+1}  + \pi_{t+1} \tilde Y_{2,t+1}$.
	Thus, for $\pi=0$, $Y_{t+1}$ follows the 1F-GAS model, for $\pi = 1$, $Y_{t+1}$ follows the 2F-GAS model and for $\pi \in (0,1)$, $Y_{t+1}$ follows some convex combination of these two models.\footnote{
	While generating returns stemming from convex combinations of GARCH-type volatility models is straight-forward by using convex combinations of the conditional volatilities, this is not as simple for the more general GAS models considered in this section. 
	Consequently, we use this more involved approach based on Bernoulli draws in order to generate these convex model combinations.}
	
	Both models in the GARCH DGP generate data from a pure scale (volatility) process resulting in perfectly colinear VaR and ES forecasts.
	In contrast, the more general VaR/ES GAS models in the second DGP generate VaR and ES forecasts which are not colinear and consequently introduce misspecification in the quantile model of the strict ES encompassing test.
	As the utilized parameters are calibrated to daily financial returns, these models reflect a realistic degree of misspecification encountered in practical risk management.

	\subsection{Simulation Results}
	\label{sec:SimulationResults}
	
	Table \ref{tab:Size} reports the empirical sizes of the three different ES encompassing tests introduced in Section \ref{sec:Theory} together with the VaR encompassing test of \cite{GiacominiKomunjer2005} at a 10$\%$ nominal significance level based on $2000$ Monte Carlo replications.
	Table S.1 and S.2 in the supplementary material present equivalent results for nominal sizes of 1$\%$ and 5$\%$.
	The column panel $\mathbb{H}_0^{(1)}$ indicates that we test whether model 1 encompasses model 2, while the panel $\mathbb{H}_0^{(2)}$ indicates the reverse.	

	\begin{table}[tbh]
		\footnotesize
		\centering
		\caption{Empirical Sizes of the Forecast Encompassing Tests.}
		\label{tab:Size}
		\begin{tabularx}{0.95\linewidth}{l l @{\hspace{0.7cm}} RRRR l @{\hspace{0.5cm}} RRRR }
			\toprule
			 & & \multicolumn{4}{c}{$\mathbb{H}_0^{(1)}$} & & \multicolumn{4}{c}{$\mathbb{H}_0^{(2)}$} \\
			\cmidrule(lr){3-6} \cmidrule(lr){8-11}
			& & Str ES & Aux ES & VaR ES & VaR & & Str ES & Aux ES & VaR ES & VaR \\ 
			\midrule
			\midrule
			$n$ & &\multicolumn{9}{c}{GARCH} \\
			\cmidrule(lr){3-11}
			$500$  &        & 15.25  & 15.20  & 18.35  & 22.75  &        & 14.40  & 14.65  & 18.80  & 22.50 \\
$1000$ &        & 11.55  & 11.10  & 15.60  & 20.10  &        & 12.30  & 12.70  & 17.80  & 22.85 \\
$2500$ &        & 11.45  & 11.55  & 16.35  & 18.80  &        & 11.00  & 11.25  & 14.60  & 17.55 \\
$5000$ &        & 10.05  & 10.25  & 13.10  & 15.35  &        &  9.75  & 10.15  & 13.90  & 15.75 \\

			\midrule
			\midrule
			$n$ & & \multicolumn{9}{c}{VaR/ES GAS} \\
			\cmidrule(lr){3-11}
			$500$  &        & 29.35  & 29.75  & 24.15  & 27.85  &        & 21.70  & 21.15  & 23.10  & 27.80 \\
$1000$ &        & 22.75  & 21.85  & 19.55  & 23.95  &        & 18.15  & 18.60  & 18.15  & 22.80 \\
$2500$ &        & 16.05  & 15.80  & 16.20  & 18.35  &        & 12.65  & 13.50  & 15.65  & 19.65 \\
$5000$ &        & 13.50  & 13.60  & 14.05  & 16.60  &        & 10.60  & 11.35  & 14.10  & 17.95 \\

			\bottomrule 
			\addlinespace
			\multicolumn{11}{p{.93\linewidth}}{\textit{Notes:} This table presents the empirical sizes (in $\%$) of our three forecast encompassing tests for the ES together with a VaR encompassing test of \cite{GiacominiKomunjer2005} for a nominal size of $10\%$.
			The results are shown for the two DGPs described in Section \ref{sec:DGPs} in the horizontal panels, for both tested hypotheses in the vertical panels and for different sample sizes.}
		\end{tabularx}
	\end{table}

	We find that the two ES encompassing tests (the strict and auxiliary test) are well-sized, especially in large samples for both DGPs and for both null hypotheses.
	While the joint VaR and ES test is slightly oversized, the VaR test exhibits even larger sizes.
	This behavior is especially remarkable as the ES is considerably further in the tail than the VaR at the same probability level and hence, harder to estimate and test.
	This pattern can be explained by the fact that the asymptotic covariance of the two tests involving the VaR is subject to estimation of the density quantile function $h_t(g_t^q( \beta^\ast_n))$, which is hard to estimate for small probability levels \citep{Koenker1978, GiacominiKomunjer2005, DimiBayer2019}.

	We further find that the strict and the auxiliary tests behave almost identically. 
	This also holds for the VaR/ES GAS DGP for which the regression model of the strict ES encompassing test is potentially misspecified.\footnote{
			Figure S.1 in the supplementary material plots the ratio of the VaR and ES forecasts for a simulated return series of the 2F-GAS model (and further misspecified models described in Section S.1 in the supplementary material).
			Following the discussion after Definition \ref{def:StrictESEncompassing}, this ratio mainly governs the degree of misspecification the regression model of the strict ES encompassing test is subject to.
			For the 2F-GAS model, the ratio of the VaR and ES forecasts fluctuates approximately between 0.7 and 0.85, while it equals 0.84 for the location-scale approaches under normality.
			This demonstrates that the VaR/ES GAS simulation designs, which are calibrated to real financial data, generate some moderate degree of model misspecification in the regression model.}
	This suggests that the approximation error induced by the misspecification in the strict ES test is negligible for realistic financial settings.
	Remarkably, in the vast majority of cases, the strict ES test exhibits better size properties than the correctly specified joint VaR and ES and the VaR encompassing tests.
	
	\begin{figure}[htb]
		\includegraphics[width=\linewidth]{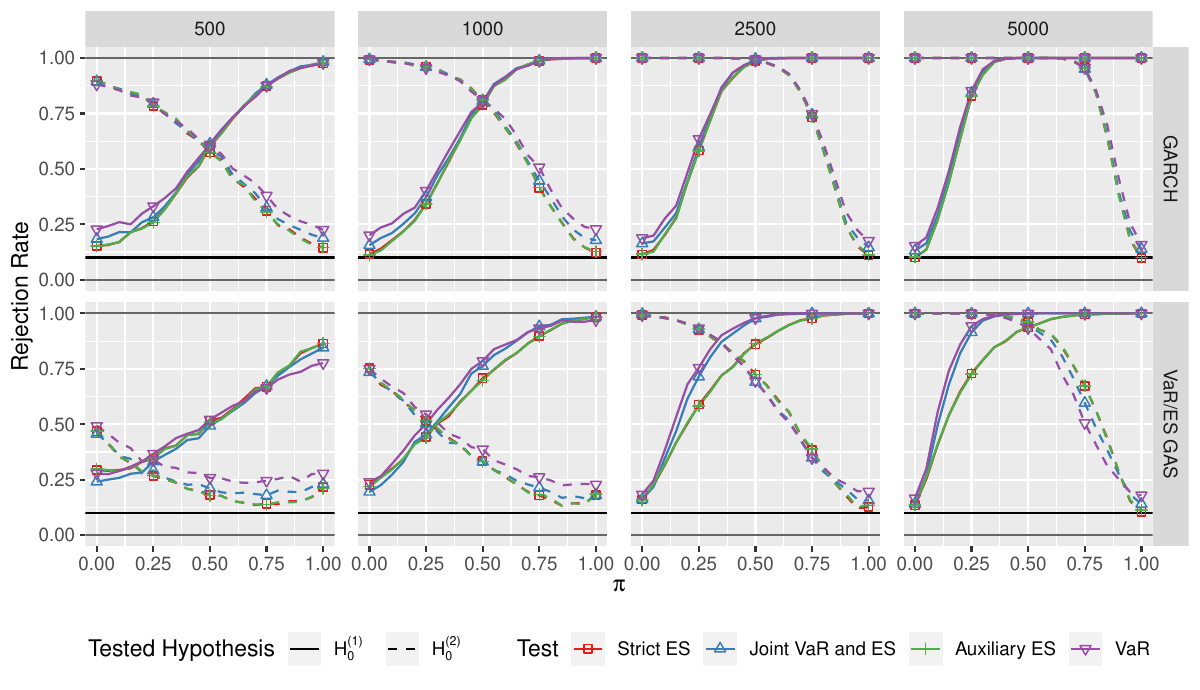}
		\caption{This figure shows power curves (empirical rejection frequencies) for the encompassing tests with a nominal size of $10\%$ and for the two DGPs described in Section \ref{sec:DGPs} in the plot rows.
		The plot columns depict different sample sizes, while the colors indicate the four different encompassing tests and the line types refer to the two tested null hypotheses.
		An ideal test exhibits a rejection frequency of $10\%$ for $\pi=0$ and for $\mathbb{H}_0^{(1)}$ (and inversely for $\pi=1$ and $\mathbb{H}_0^{(2)}$) and as sharply increasing rejection rates as possible for increasing (decreasing) values of $\pi$. }
		\label{fig:sim_power1}
	\end{figure}
	
	We present power curves (empirical rejection rates) for both DGPs and different sample sizes in the individual plot panels in Figure \ref{fig:sim_power1}.
	In each plot, we depict the respective power curves for our three ES encompassing tests and the VaR encompassing test of \cite{GiacominiKomunjer2005} for both null hypotheses and for a nominal significance level of $10\%$ based on $2000$ Monte Carlo replications. 
	We observe increasing power for both DGPs, both tested null hypotheses and for all four encompassing tests for increasing (decreasing) values of the combination parameter $\pi$.
	We find that while the VaR and joint VaR and ES tests are considerably oversized, they produce a similar test power compared to the strict and auxiliary ES encompassing tests, especially for larger (smaller) values of $\pi$.
	Again, the strict and auxiliary ES encompassing tests are almost indistinguishable, which implies that the strict test is robust against the misspecification induced by the VaR/ES GAS models.
	Interestingly, we find that the power curves for the VaR/ES GAS specifications are slightly asymmetric implying that the tests react differently to the specifications of different numbers of driving factors in the GAS models.

	\subsection{Extensions of the Simulation Setup}
	\label{sec:SimulationExtensions}

	In this section, we consider three extensions of our simulation setup. 
	First, in Section \ref{sec:SimulationDGP}, we present two additional DGPs.
	Second, in Section \ref{sec:SimulationLosses}, we analyze the behavior of our tests under different strictly consistent loss functions for the pair VaR and ES.
	Third, in Section \ref{sec:SimulationLinkFunctions}, we employ two additional link functions by testing for forecast encompassing for affine and nonlinear forecast combinations.

	\subsubsection{Forecast Encompassing under Different Data Generating Processes}
	\label{sec:SimulationDGP}
	
	In this subsection, we consider two additional DGPs, namely the GAS-$t$ model of \cite{Creal2013} and the ES-CAViaR models of \cite{Taylor2019}, which are described in detail in Section S.1 in the supplementary material.
	Both models go beyond the class of pure scale models and consequently generate model misspecification in the quantile regression equation of the strict ES encompassing test.
	Figure \ref{fig:sim_power_AddDGPs} presents power curves for these two DGPs and Table S.3 in the supplementary material reports the corresponding test sizes.
	
	\begin{figure}[tbh]
		\includegraphics[width=\linewidth]{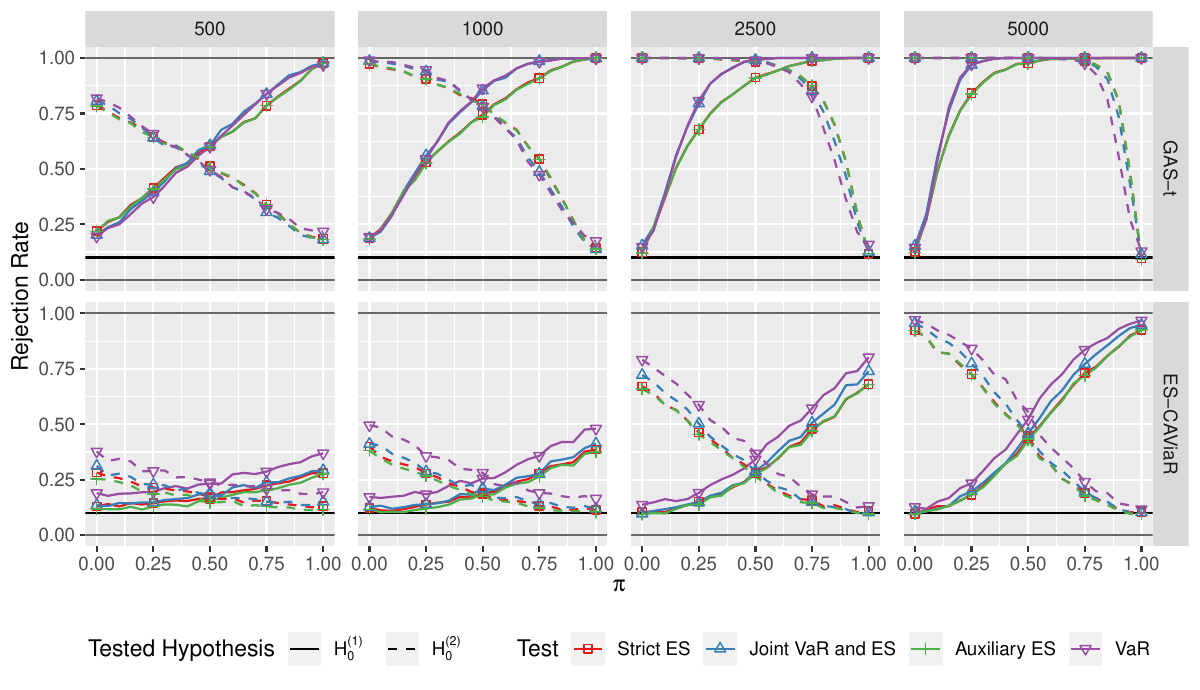}
		\caption{This figure shows power curves (empirical rejection frequencies) for the encompassing tests with a nominal size of $10\%$. The {plot rows} correspond to the GAS-$t$ and the ES-CAViaR DGPs described in Section S.1 in the supplementary material.
		The {plot columns} depict different sample sizes, while the colors indicate the four different encompassing tests and the line types refer to the two tested null hypotheses.}
		\label{fig:sim_power_AddDGPs}
	\end{figure}
	
	The results of these two DGPs qualitatively confirm the simulation results of Section \ref{sec:SimulationResults}.
	The three ES specific tests exhibit accurate empirical test sizes, especially in large samples and the strict and the auxiliary tests generally exhibit better size properties than the joint VaR and ES and the stand-alone VaR encompassing tests.
	Increasing the sample size results in increasing power for both DGPs and for all considered encompassing tests.
	Noteworthy, all tests show considerably lower power for the ES-CAViaR DGP compared to the other DGPs.
	This result is comparable with the power results of \cite{GiacominiKomunjer2005} as this DGP is a slightly modified version of their DGP.
	In summary, these two additional DGPs demonstrate that the new ES encompassing tests perform well for a variety of realistic data generating processes.

	\subsubsection{Forecast Encompassing under Different Loss Functions}
	\label{sec:SimulationLosses}
	
	The three encompassing test specifications for the ES presented in Section \ref{sec:EncompassingTestsES} are built on the zero-homogeneous joint loss function for the VaR and ES given in (\ref{eqn:JointLossESReg0Hom}).
	While this loss function is the most popular choice in the recent literature on semiparametric ES models \citep{Patton2019, BayerDimi2019, Taylor2019}, there exists an entire class of joint loss functions for the VaR and ES, proposed by \cite{Fissler2016} as
	\begin{align} 
		\begin{aligned}
			\label{eqn:JointLossESRegGeneral}
			\rho(Y,q_\alpha,e_\alpha) &= \big( \mathds{1}_{\{Y \le q_\alpha\}} - \alpha \big) \mathfrak{g}(q_\alpha) - \mathds{1}_{\{Y \le q_\alpha\}}  \mathfrak{g}(Y) \\
			&+ \phi'(e_\alpha) \left( e_\alpha - q_\alpha + \frac{(q_\alpha - Y) \mathds{1}_{\{Y \le q_\alpha\}}}{\alpha}  \right) - \phi(e_\alpha) + a(Y),
		\end{aligned}
	\end{align}	
	where the function $\mathfrak{g}$ is twice continuously differentiable and increasing, $\phi$ is three times continuously differentiable, strictly increasing and strictly convex, and $a$ and $\mathfrak{g}$ are integrable functions \citep{Fissler2016}.
	The loss function in (\ref{eqn:JointLossESReg0Hom}) is a special case of (\ref{eqn:JointLossESRegGeneral}) for the choices $\mathfrak{g}(z) = 0$,  $\phi(z) = -\log(-z)$ and $a(z) = 0$.
	
	The ES encompassing tests can generally be set up by using any choice of (\ref{eqn:JointLossESRegGeneral}) (fulfilling certain further weak regularity conditions). We consider two additional specifications in the following.
	Following the theory of homogeneous loss functions \citep{Nolde2017} and the numerical performance in linear regression settings \citep{DimiBayer2019}, we fix $\mathfrak{g}(z) = 0$ and in addition to $\phi(z) = -\log(-z)$, we employ the choices $\phi(z) = 1/\sqrt{-z}$ and $\phi(z) = - 1/z$.\footnote{Strictly speaking, the asymptotic theory in Theorem \ref{thm:DistributionEncmpTestStatistics} only covers the M-estimator based on the loss function in (\ref{eqn:JointLossESReg0Hom}).
		However, the proofs and the resulting asymptotic covariance matrices are easily extended to the general case by combining the methods of this paper with the extension of \cite{DimiBayer2019} to the case of general loss functions as given in (\ref{eqn:JointLossESRegGeneral}).}

	\begin{figure}[tbhp]
		\includegraphics[width=\linewidth]{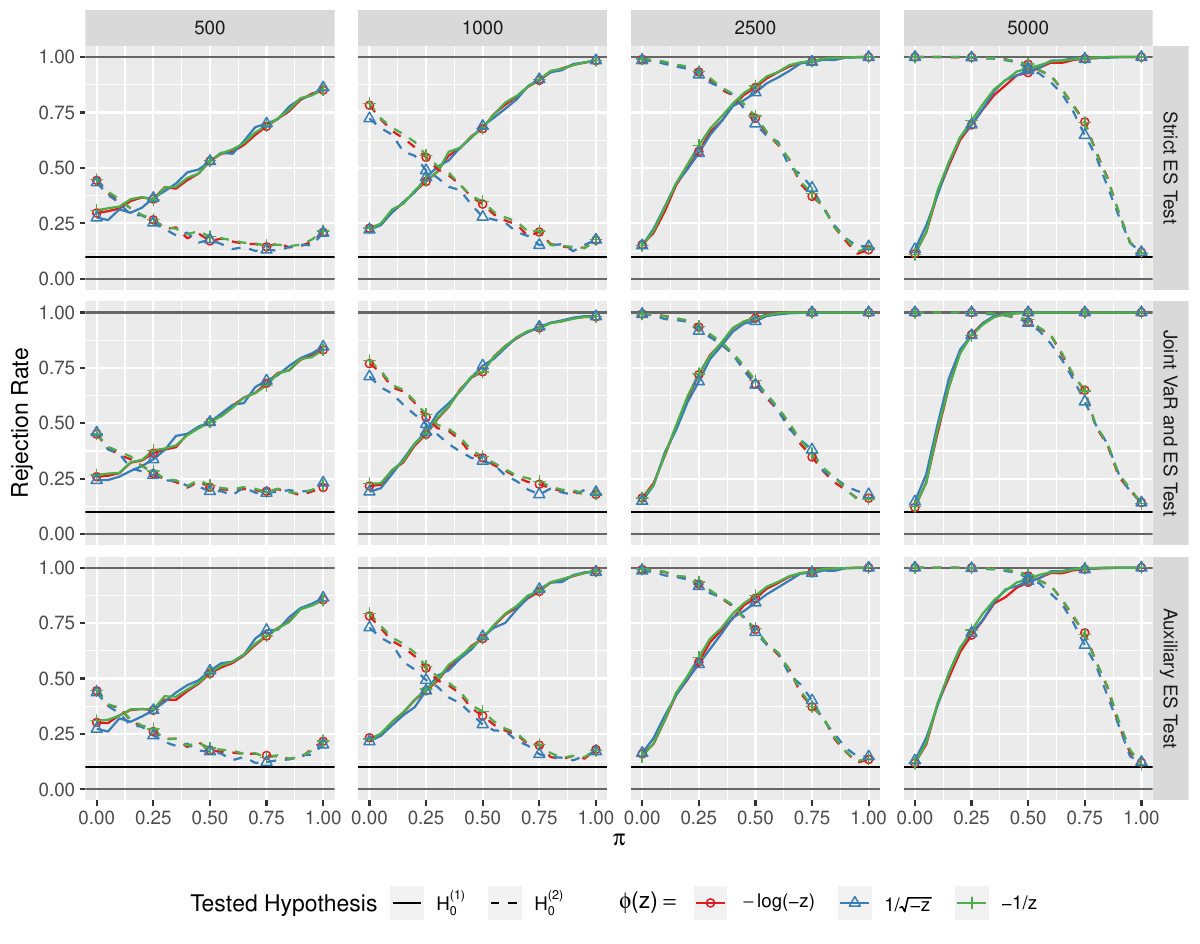}		
		\caption{This figure shows power curves (empirical rejection frequencies) for the three ES specific encompassing tests with a nominal size of $10\%$ for the VaR/ES GAS DGP described in Section S.1 in the supplementary material, where the different colors represent the different loss functions employed in the tests, given in and below (\ref{eqn:JointLossESRegGeneral}).
		The {plot rows}  depict different sample sizes, the {plot columns} show the three different ES encompassing tests and the line types refer to the tested null hypotheses.}		
		\label{fig:sim_power_AddLosses}
	\end{figure}

	Figure \ref{fig:sim_power_AddLosses} shows rejection rates for the VaR/ES GAS DGP for different sample sizes and the three encompassing test specifications for the ES, where the different line colors represent the three different loss specifications.
	Table S.4 in the supplementary material reports the corresponding test sizes.
	We find that the tests based on the three different loss functions perform almost indistinguishably, especially in large samples.
	This result  is not unexpected (especially in large samples where the expectation is well approximated by the sample mean) as  \cite{Fissler2016} show that all loss functions in the class in (\ref{eqn:JointLossESRegGeneral}) are uniquely minimized by the true VaR and ES.
	
	For the potentially misspecified strict ES encompassing test, the pseudo-true parameter defined in (\ref{eqn:DefQn0}) may theoretically depend on the underlying loss function and notice that the VaR/ES GAS DGP used in Figure \ref{fig:sim_power_AddLosses} allows for such a model misspecification.
	However, we find that neither the rejection rates of the strict ES encompassing test are affected by employing different loss functions, nor are the rejection rates of the two correctly specified encompassing tests in Figure \ref{fig:sim_power_AddLosses}.
	This result implies that the potentially different pseudo-true parameters are almost entirely unaffected by the misspecification.

	\subsubsection{Forecast Encompassing under Different Link Functions}
	\label{sec:SimulationLinkFunctions}
	
	In this section, we employ two additional link function specifications.
	First, we consider an affine combination including an intercept\footnote{We include an intercept as this stabilizes the performance of the associated quantile regression. 
	We do not include classical convex combinations (where $0 \le \beta_2 \le 1$) as our theoretical framework does not allow for testing on the boundary (see e.g. \cite{Andrews1999} for details).}
	\begin{align}
		\label{eqn:AffineLinkFunctions}
		g^q( \bfc, \beta) = \beta_1 + \beta_2 \hat f_{1,t} + (1- \beta_2) \hat f_{2,t}
		\quad \text{ and } \quad
		g^e( \be, \eta) = \eta_1 + \eta_2 \hat e_{1,t} + (1- \eta_2) \hat e_{2,t},
	\end{align}
	where $\bfc = \bq$ for the joint and auxiliary tests and $\bfc = \be$ for the strict test. 
	For the joint test, the first null hypothesis is given by $\mathbb{H}_0^{(1)}: (\beta_2^\ast, \eta_2^\ast) = (1,1)$ while for the strict and auxiliary tests, it is given by $\mathbb{H}_0^{(1)}: \eta_2^\ast = 1$.
	The second, opposing null hypotheses, $\mathbb{H}_0^{(2)}$ are obtained by replacing the ones by zeros.
	For the affine link functions, we employ the same DGPs as for the encompassing tests based on linear link functions.

	Furthermore we employ the nonlinear link functions, where $\bfc$ is given as in (\ref{eqn:AffineLinkFunctions}) and
	\begin{align}
		\begin{aligned}
		\label{eqn:NonlinearLinkFunctions}
		g^q( \bfc, \beta) &= \beta_1 - \exp \big( \beta_2 \log(- \hat f_{1,t}) + \beta_3 \log(- \hat f_{2,t}) \big),
		\quad \text{ and }  \\
		g^e( \be, \eta) &= \eta_1 - \exp \big( \eta_2 \log(- \hat e_{1,t}) + \eta_3 \log(- \hat e_{2,t}) \big),
		\end{aligned}
	\end{align}
	and we test the same null hypotheses as for the linear link functions described in (\ref{eqn:SimStudyHypotheses}).
	For the nonlinear link function, we employ a slightly modified GARCH DGP.
	As in Section \ref{sec:DGPs}, let $\hat \sigma_{1,t}$ and $\hat \sigma_{2,t}$ denote the conditional volatilities of the GARCH and GJR-GARCH models.
	Then, we simulate data according to
	\begin{align}
		\label{eqn:LogGARCHDGP}
		Y_{t+1} = \exp \big(  (1-\pi)  \log (\hat \sigma_{1,t} ) +  \pi \log (\hat \sigma_{2,t}) \big) \cdot u_{t+1}
	\end{align}
	for an equally spaced grid of 21 values for $\pi \in [0,1]$ and where  $u_{t+1} \stackrel{iid}{\sim}  \mathcal{N}(0,1)$.
	This ensures that for $\pi=0$, $\hat q_{1,t}$ and $\hat e_{1,t}$ are the correct VaR and ES forecasts, and vice versa for $\pi=1$.
	For any $\pi \in (0,1)$, the true VaR and ES are given by a combination according to the nonlinear link functions in (\ref{eqn:NonlinearLinkFunctions}).
	
	\begin{figure}[htb]
		\includegraphics[width=\linewidth]{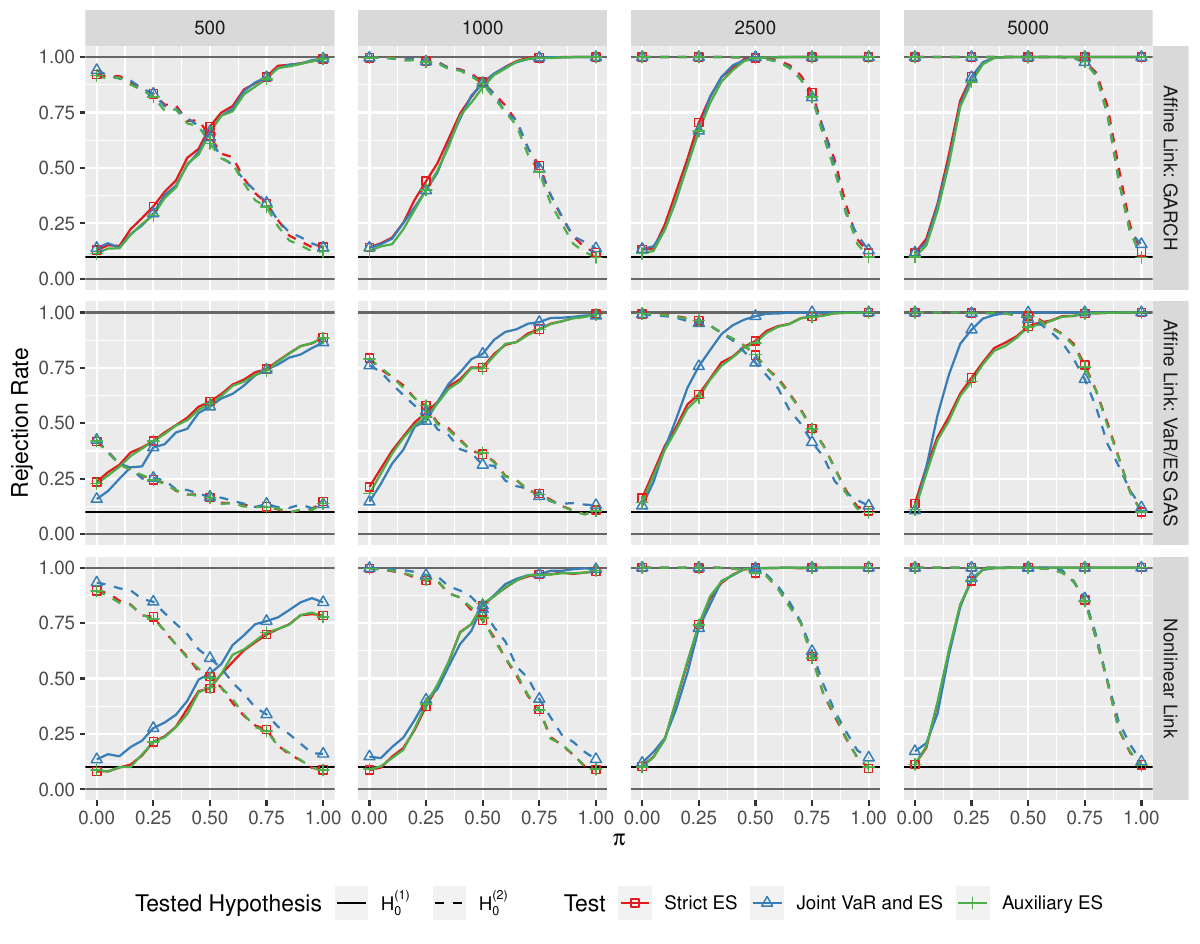}
		\caption{This figure shows power curves (empirical rejection frequencies) for the encompassing tests with a nominal size of $10\%$ for two additional link functions.
		The first two rows of plots present results for the affine link functions given in (\ref{eqn:AffineLinkFunctions}) and for the GARCH and VaR/ES GAS DGPs.
		The third plot row presents results for the nonlinear link functions given in  (\ref{eqn:NonlinearLinkFunctions}) and the respective nonlinear GARCH DGP.
		The {plot columns} depict different sample sizes, while the colors indicate the three ES specific encompassing tests and the line types refer to the two {tested null hypotheses}.}
		\label{fig:sim_power_AddLinks}
	\end{figure}

	Figure \ref{fig:sim_power_AddLinks} presents power curves for these two additional link functions, where for the affine combination, we consider both, the GARCH and the VaR/ES GAS DGPs.
	Table S.5 in the supplementary material presents the associated test sizes. 
	We find that the power curves for the affine combinations in Figure \ref{fig:sim_power_AddLinks} are comparable to the linear specifications, as given in Figure \ref{fig:sim_power1} for both considered DGPs.
	Moreover, the encompassing tests based on the nonlinear link function in the third row of plots perform similar to the linear encompassing tests in the GARCH setting.
	This extension of the simulation setup shows that our ES encompassing tests can be applied based on a variety of different link functions.

	\section{Empirical Application}
	\label{sec:EmpiricalApplication}
	
	We use close-to-close returns from the IBM stock, the S$\&$P 500 index and the DAX 30 index from June 1st, 2000 until May 31st, 2019, which amounts to a total of $T = 4779$ daily observations.
	We use a fixed forecasting scheme, i.e.\ 
	the model parameters are estimated once on the first $m=2000$ in-sample observations.
	These parameter estimates are used to generate the VaR and ES  forecasts in a rolling-window fashion for the remaining out-of-sample period of $n = 2779$ days.
	Following the suggestion of the Basel III Accords, we use the probability level $\alpha = 2.5\%$ for the VaR and the ES.

	For the analysis, we consider the following competing forecasting models.
	First, we employ the Historical Simulation (HS) model which generates VaR and ES forecasts by computing the empirical quantile and ES at level $\alpha$ of the past 250 trading days.
	The second model is the RiskMetrics (RM) model, which models the conditional volatility as an IGARCH equation with fixed parameter values, $\hat \sigma_{t}^2 = 0.94 \hat \sigma_{t-1}^2 + 0.06 Y_{t}^2$ and Gaussian residuals.
	Third, we use the GJR-GARCH(1,1)-$t$ model of \cite{Glosten1993} with Student-$t$ residuals.
	The fourth model is given by the Student-$t$-GAS model with time-varying variance and degrees of freedom introduced in Section S.1 in the supplementary material.
	The fifth and sixth model are the one and two factor GAS models for the VaR and ES of \cite{Patton2019} set out in Section \ref{sec:DGPs} and estimated by minimizing the strictly consistent loss function for the VaR and ES given in (\ref{eqn:JointLossESReg0Hom}).
	The last two models are the two dynamic ES-CAViaR models of \cite{Taylor2019} described in Section S.1 in the supplementary material.
	Table S.6 in the supplementary material shows the correlations of the respective VaR and ES forecasts of these models.
	We find that no pair of forecasts is perfectly correlated, which is crucial for the applicability of the encompassing tests as implied by condition \ref{cond:FullRankConditionNormality} of Assumption \ref{assu:AsymptoticTheory}.

	We run pair-wise encompassing tests comparing all eight forecasting methods. 
	Hence, for each model pair, we run encompassing tests for both hypotheses, i.e.\ that the first forecast encompasses the second, denoted by $\mathbb{H}_0^{(1)}$ and the inverse, denoted by $\mathbb{H}_0^{(2)}$.
	This results in four possible outcomes of these two tests: 
	(1) \textit{non-rejection} (NR) indicates that none of the null hypotheses is rejected and the tests are not helpful.
	(2) \textit{encompassed} (E1) denotes the setting where the first model is encompassed by the competitor model but does not encompass it, i.e.\ $\mathbb{H}_0^{(1)}$ is rejected but $\mathbb{H}_0^{(2)}$ is not, which results in choosing the competitor model.
	(3) \textit{encompassing} (E2) indicates that the first model encompasses the other but is not encompassed by it, i.e.\ $\mathbb{H}_0^{(1)}$ is not rejected but $\mathbb{H}_0^{(2)}$ is, which implies that we choose the first model.
	Finally, (4) \textit{combination} (C) refers to a setting where both null hypotheses are rejected and we consequently opt for a forecast combination.

	\begin{table}[p]
		\centering
		\footnotesize
		\singlespacing
		\caption{Encompassing Test Results for the IBM Stock, the DAX 30 and the S\&P 500 Indices.}
		\label{tab:Freq_EmpAppl}
		\scalebox{0.95}{
		\begin{tabularx}{1.05\linewidth}{l rrrrr  rrrrr rrrrr rrrrr}
			\toprule
			\addlinespace
			& & \multicolumn{4}{c}{Strict ES} & & \multicolumn{4}{c}{Aux ES } & & \multicolumn{4}{c}{Joint VaR ES } & & \multicolumn{4}{c}{VaR }\\
			\cmidrule(lr){3-6} \cmidrule(lr){8-11} \cmidrule(lr){13-16} \cmidrule(lr){18-21}
			Model && NR & E1 & E2 & C & & NR & E1 & E2 & C & & NR & E1 & E2 & C & & NR & E1 & E2 & C  \\  
			\midrule
			\midrule
			\addlinespace
			& & \multicolumn{19}{c}{IBM} \\
			\cmidrule(rr){3-21} 
			\cmidrule(lr){3-6} \cmidrule(lr){8-11} \cmidrule(lr){13-16} \cmidrule(lr){18-21}
			HS &   &   & 57 &   & 43 &   &   & 57 &   & 43 &   &   & 43 &   & 57 &   &   & 43 &   & 57\\
RM &   &   & 57 &   & 43 &   &   & 57 &   & 43 &   &   & 43 &   & 57 &   & 14 & 43 &   & 43\\
GJR &   &   & 57 &   & 43 &   &   & 57 &   & 43 &   &   & 43 &   & 57 &   & 14 & 29 &   & 57\\
GAS &   &   & 43 &   & 57 &   &   & 43 &   & 57 &   &   & 29 &   & 71 &   &   & 29 &   & 71\\
G1F &   & 14 & 29 & 43 & 14 &   & 14 & 29 & 43 & 14 &   &   & 14 &   & 86 &   &   & 14 &   & 86\\
G2F &   & 14 & 29 & 57 &   &   & 14 & 29 & 57 &   &   &   & 29 & 57 & 14 &   &   & 29 & 43 & 29\\
ASES &   & 14 &   & 86 &   &   & 14 &   & 86 &   &   & 14 &   & 57 & 29 &   & 14 &   & 57 & 29\\
SAVES &   & 14 &   & 86 &   &   & 14 &   & 86 &   &   & 14 &   & 86 &   &   & 14 &   & 86 &  \\

			\midrule
			\midrule
			\addlinespace
			& & \multicolumn{19}{c}{S\&P 500} \\
			\cmidrule(rr){3-21} 
			HS &   &   & 86 &   & 14 &   &   & 86 &   & 14 &   &   & 71 &   & 29 &   &   & 86 &   & 14\\
RM &   &   & 71 & 14 & 14 &   &   & 71 & 14 & 14 &   &   & 71 & 14 & 14 &   & 14 & 43 & 14 & 29\\
GJR &   & 29 &   & 71 &   &   & 29 &   & 71 &   &   &   &   & 71 & 29 &   &   &   & 71 & 29\\
GAS &   & 29 & 29 & 29 & 14 &   & 14 & 43 & 29 & 14 &   &   & 57 & 29 & 14 &   & 14 & 43 & 14 & 29\\
G1F &   &   & 29 &   & 71 &   &   & 29 &   & 71 &   &   & 14 &   & 86 &   & 29 & 14 & 29 & 29\\
G2F &   & 29 & 14 & 57 &   &   & 14 & 14 & 71 &   &   &   &   & 57 & 43 &   & 14 &   & 29 & 57\\
ASES &   & 14 &   & 71 & 14 &   & 14 &   & 71 & 14 &   &   &   & 43 & 57 &   &   &   & 29 & 71\\
SAVES &   & 14 & 43 & 29 & 14 &   & 14 & 43 & 29 & 14 &   &   & 43 & 43 & 14 &   & 14 & 29 & 29 & 29\\

			\midrule
			\midrule
			\addlinespace
			& & \multicolumn{19}{c}{DAX 30} \\
			\cmidrule(rr){3-21} 
			HS &   &   & 86 &   & 14 &   &   & 100 &   &   &   &   & 86 &   & 14 &   &   & 86 &   & 14\\
RM &   &   & 43 & 57 &   &   &   & 43 & 57 &   &   &   & 43 & 43 & 14 &   &   & 29 & 14 & 57\\
GJR &   &   & 57 & 43 &   &   &   & 57 & 43 &   &   &   & 14 & 71 & 14 &   & 14 & 14 & 43 & 29\\
GAS &   & 14 & 14 & 71 &   &   & 29 &   & 71 &   &   & 14 & 29 & 43 & 14 &   &   & 14 & 57 & 29\\
G1F &   &   & 71 & 29 &   &   &   & 71 & 29 &   &   &   & 71 & 14 & 14 &   &   & 43 & 14 & 43\\
G2F &   &   & 86 &   & 14 &   &   & 86 & 14 &   &   &   & 57 &   & 43 &   &   & 43 &   & 57\\
ASES &   & 29 &   & 71 &   &   & 29 &   & 71 &   &   &   &   & 86 & 14 &   & 14 &   & 86 &  \\
SAVES &   & 14 &   & 86 &   &   & 29 &   & 71 &   &   & 14 &   & 43 & 43 &   & 29 & 14 & 29 & 29\\

			\bottomrule
			\addlinespace
			\multicolumn{21}{p{1.03\linewidth}} {\textit{Notes:} This table shows results of the pair-wise encompassing test outcomes for the  three ES encompassing tests and the VaR encompassing test of \cite{GiacominiKomunjer2005}, the eight considered models estimated on an in-sample period of $m=2000$ days, a nominal test significance level of 10$\%$ and for the IBM stock, the S\&P 500 and the DAX 30 indices.
			The individual columns report the relative frequencies (in \%) how often the respective test outcome occurs for a model (given in the rows) when this model generates the \textit{first} forecasts in the encompassing tests.
			The column ``NR'' (non-rejection) reports the frequency of cases where neither null hypothesis is rejected, while the column ``C'' (combination) indicates the frequency  that both null hypotheses are rejected.
			The column ``E1'' (encompassed) refers to the case that $\mathbb{H}_0^{(1)}$ is rejected while $\mathbb{H}_0^{(2)}$ is not, i.e. it gives the relative frequencies that this model is encompassed by its competitors and the column ``E2'' (encompassing) refers to the inverse case, i.e. it gives the relative frequencies that this model encompasses its competitors.}
		\end{tabularx}
		}	
	\end{table}

	In Table \ref{tab:Freq_EmpAppl}, we report relative frequencies of test outcomes at the $10\%$ significance level for the different encompassing tests for all three return time series.
	Tables S.7, S.8, S.9 in the supplementary material report the individual $p$-values of the encompassing tests.
	The results can be summarized as follows:
	first, for the IBM stock returns we find many cases of double rejections and hence empirical evidence for using forecast combinations. 
	This implies that the individual models provide additional and exclusive information and hence, a forecast combination is often superior to the stand-alone forecasting models. 
	This finding supports the theoretical advantages of forecast combinations, presented in a general setting e.g., by \cite{GiacominiKomunjer2005}, \cite{Timmermann2006} and \cite{Halbleib2012}, and specifically for the pair VaR and ES by \cite{Taylor2019FCcomb}.
	Second, for the S$\&$P 500 index, and especially for the DAX 30 index, we overall observe less instances of double rejections of the ES encompassing tests.
	While the decrease in cases where the VaR encompassing test opts for a forecast combination is smaller, these rejections have to be considered carefully given that the VaR encompassing test is oversized in all simulation setups in Section \ref{sec:SimulationStudy}, even in large samples.
	This result can be explained by the fact that the S$\&$P 500 and DAX indices are well diversified and the returns fluctuate to a lesser extent and exhibit less extreme outliers than single stock return series.
	Furthermore, the considered VaR and ES forecasts show larger correlations for the indices than for the single stock in Table S.6 in the supplementary material, which negatively influences the tests' power.
	Third, in terms of the frequencies of the cases E1 and E2, we observe recurring patterns over the different models for both time series.
	Especially the ES-specific GAS and CAViaR type models seem to exhibit a superior performance, while the HS, RM, GARCH and GAS-$t$ models tend to be encompassed more often.
	Lastly, the two tests which only focus on testing encompassing of ES forecasts perform almost identically, which supports the conclusion from the simulation study that the potential misspecification does not negatively influence the performance of the strict ES test in realistic financial settings.
	This is encouraging as the strict ES encompassing test can be applied in cases where one does not have VaR forecasts at hand, such as it is currently imposed by the Basel Committee of Banking Supervision \citep{Basel2016, Basel2017}.
	
	Tables S.10, S.11 and S.12 in the supplementary material report the joint VaR and ES losses for the zero-homogeneous loss function in (\ref{eqn:JointLossESReg0Hom}) for forecasts stemming from the stand-alone models and the respective forecast combinations with estimated combination weights from the underlying regressions.
	These results qualitatively confirm the results of the encompassing tests: e.g., for the IBM stock in the first panel of Table \ref{tab:Freq_EmpAppl}, we find that forecast combinations are particularly preferred for the first four models (in the model ordering of the table).
	The average losses in Table S.10 show a similar pattern throughout all three panels. 
	We observe that the optimal forecast combinations exhibit substantially smaller losses compared to the stand-alone models for the first four models while this decrease is of a considerably smaller magnitude for the last four models.

	\FloatBarrier

	\section{Conclusion}
	\label{sec:Conclusion}
	
	With the implementation of the third Basel Accords \citep{Basel2016,Basel2017}, risk managers and regulators currently shift attention towards the risk measure Expected Shortfall (ES), which demonstrates the necessity of forecast evaluation and comparison tools for the ES.
	In this paper, we introduce new forecast encompassing tests for the ES, which are based on a joint loss function and an associated joint regression framework for the ES together with the Value at Risk \citep{Fissler2016, Patton2019, DimiBayer2019}.
	We propose three variants of the ES encompassing test, which allow for testing forecast encompassing for flexible parametric forecast combination methods.
	As one test variant is potentially subject to model misspecification, we extend the existing asymptotic theory of \cite{Patton2019}, \cite{DimiBayer2019} and \cite{BayerDimi2019} to cases of potential model misspecification for flexible parametric models.
	Potential future research on extending the setting presented in this paper includes encompassing tests for convex forecast combinations and forecast combinations which theoretically prevent crossings of the VaR and the ES.
	Both approaches require non-standard asymptotic theory for tests on the boundary of the parameter space.
	
	Tests for forecast encompassing establish a theoretical foundation for forecast combinations of two competing forecasts when both opposing hypotheses of forecast encompassing are rejected.
	This situation corresponds to the case when neither forecast encompasses its competitor.
	Generally, applying forecast combinations can be highly beneficial through the diversification gains stemming from combining different model specifications and underlying information sets.
	This benefit can be particularly pronounced for extreme risk measures such as the ES \citep{Taylor2019FCcomb}, as the stand-alone models are very sensitive to the very little observations in the tails of the return distributions.
	Thus, combining forecasts can be seen as a robustification of the forecasts.
	Our application empirically validates this conjecture, especially pronounced for  daily returns from the IBM stock.
	
	While we propose encompassing tests suitable for the cases that the ES is reported jointly with and without its accompanying VaR forecast, testing encompassing for the VaR and ES jointly is most natural given their joint elicitability.
	Furthermore, this joint elicitability property clearly hints towards reporting ES forecasts jointly with their corresponding VaR forecasts by default.
	In contrast, the auxiliary test can be seen as the first forecast comparison procedure for the ES, which focuses almost entirely on the ES.
	Theoretically, its application is most reasonable in cases where the primary focus of the researcher is on the ES, even though VaR forecasts are available.
	Eventually, the strict ES test is specific to scenarios where only competing ES forecasts are available.
	Furthermore, in cases where competing ES forecasts are built on the same VaR model (forecasts), the strict ES test is applicable while the application of the joint and auxiliary tests is infeasible due to the collinearity of the (identical) VaR forecasts.

	\if0\blind
	{

\section*{Acknowlegdements}
For helpful comments, we thank the editor (Michael McCracken), the associate editor and two referees, as well as Sander Barendse, Sebastian Bayer, Ralf Brüggemann, Joachim Grammig, Alastair Hall, Andrew Patton, James Taylor and the seminar participants at Universität Konstanz, Duke University, Universität Hohenheim, the 2019 QFFE conference in Marseille, the 2019 IAAE conference in Nicosia, the 2019 ESEM conference in Manchester, the 2019 Statistische Woche in Trier and the 2019 Financial Econometrics Workshop in Örebro.
Financial support by the Klaus Tschira Foundation, the Universität Hohenheim, the Graduate School of Decision Sciences (GSDS) and a travel fund of the IAAE is gratefully acknowledged.	

} \fi

	\FloatBarrier
	\singlespacing
	\setlength{\bibsep}{6pt}
	\bibliographystyle{apalike}	
	\bibliography{bib_ESencmp}

	\appendix
		\doublespacing
		
		\renewcommand\thefigure{\thesection.\arabic{figure}}    
		\renewcommand\thetable{\thesection.\arabic{table}}

		\section{Proofs}
		\label{sec:Proofs}

		\begin{proof}[Proof of Proposition \ref{prop:Consistency}] 
			We check that the necessary conditions (i) - (iv) of the basic consistency theorem, given in Theorem 2.1 in \cite{NeweyMcFadden1994}, p.\ 2121 hold, where we consider the objective functions $Q_{n}(\theta)$ and $Q_{n}^0(\theta)$ as defined in (\ref{eqn:DefQn}) and (\ref{eqn:DefQn0}).
			First, notice that condition (ii) holds by imposing condition \ref{cond:ParameterSpace}.
			The unique identification condition (i) holds by assumption \ref{cond:UniqueMinimum}.
			Next, we verify the uniform convergence condition (iv) by applying the uniform weak law of large numbers given in Theorem A.2.5. in \cite{White1994}. For that, we have to show that
			\begin{enumerate} 
				\item 
				\label{cond:LipschitzL1ESReg}
				the map $\theta \mapsto \rho \big( Y_{t+1}, g^q_t(\beta) , g^e_t(\eta)  \big)$ is Lipschitz-$L_1$ on $\Theta$,\footnote{
					See Definition A.2.3 in \cite{White1994} for a definition of Lipschitz-$L_1$.
					Notice that we do not have a double index and thus we suppress the $n$ in the notation of \cite{White1994}. Furthermore, we apply the definition by using the identify function for $a_t^{o}$.
				}
				
				\item 
				\label{cond:LLNlocally}
				For all $\theta^o \in \Theta$, there exists $\delta^o > 0$, such that for all $\delta, 0 < \delta \le \delta^o$, the sequences 
				\begin{align}
				\bar \rho_t(\theta^o,\delta) &:=  \sup_{\theta \in \Theta} \left\{ \left. \rho \big( Y_{t+1}, g^q_t(\beta) , g^e_t(\eta)  \big) \right| || \theta - \theta^o || < \delta \right\} \qquad \text{and} \\
				\ubar{\rho}_t(\theta^o,\delta) &:=  \inf_{\theta \in \Theta} \left\{ \left. \rho \big( Y_{t+1}, g^q_t(\beta) , g^e_t(\eta)  \big) \right| || \theta - \theta^o || < \delta \right\}
				\end{align}
				obey a weak law of large numbers.
			\end{enumerate}
			Condition \ref{cond:LipschitzL1ESReg} follows directly from Lemma S.1 in the supplementary material and we turn to condition \ref{cond:LLNlocally}.
			As the process $Z_t$ is strong mixing of size $-r/(r-2)$ for some $r>2$ by condition \ref{cond:AlphaMixing} and as the functions $\rho \big( Y_{t+1}, g^q_t(\beta) , g^e_t(\eta) \big)$ and the supremum/infimum functions are $\mathcal{F}_t$-measurable for all $t \in \mathbb{N}$, we can conclude that the sequences $\bar \rho_t(\theta^o,\delta)$ and $\ubar{\rho}_t(\theta^o,\delta)$ are also strong mixing of the same size by applying the same theorem.
			
			Furthermore, for $\tilde r > 1$ and for some $\delta > 0$ sufficiently small, $r \ge \tilde r+\delta$ and thus $\mathbb{E} \left[ |\bar{\rho}_t(\theta^o,\delta)|^{\tilde r + \delta} \right] \le \sup_{1\le t \le T}\mathbb{E} \left[ \sup_{\theta \in \Theta} \left| \rho \big( Y_{t+1}, g^q_t(\beta) , g^e_t(\eta)  \big) \right|^{r} \right]$ for all $t, 1\le t \le T, T\ge1$.
			As $\Theta$ is compact, there exists some $c > 0$ such that $\sup_{\theta \in \Theta} || \theta|| \le c$ and thus, for all $t = 1,\dots,T$, it holds that 
			\begin{align}
			&\mathbb{E} \left[ \sup_{\theta \in \Theta} \left| \rho \big( Y_{t+1}, g^q_t(\beta) , g^e_t(\eta)  \big)  \right|^{r} \right] \\ 
			\le \, &4^{r-1} \left\{ 
			1 
			+ \left( \frac{c}{K} \left( 1+ \frac{1}{\alpha} \right) \right) \mathbb{E} || g^q_t(\beta) ||^r
			+ \frac{1}{\alpha K} \mathbb{E} |Y_{t+1}|^r
			+ \sup_{\theta \in \Theta} \mathbb{E} || \log(g^e_t(\eta) )||^r
			\right\},
			\end{align}
			which is bounded by condition \ref{cond:MomentCondition} and as $\log(z) \le z$ for $z$ sufficiently large.
			The same inequality holds for $|\ubar{\rho}_t(\theta^o,\delta)|$.
			Thus, we can apply the weak law of large numbers for strong mixing sequences in Corollary 3.48 in \cite{White2001}, p.\ 49 in order to conclude that for all $\theta^o \in \Theta$ such that $||\theta^o - \theta || \le \delta$, it holds that $\frac{1}{n} \sum_{t=m}^{T-1} \big( \bar{\rho}_t(\theta^o,\delta) -  \mathbb{E} \left[ \bar{\rho}_t(\theta^o,\delta) \right] \big) \stackrel{\mathbb{P}}{\to} 0$ and $\frac{1}{n} \sum_{t=m}^{T-1} \big( \ubar{\rho}_t(\theta^o,\delta) - \mathbb{E} \left[ \ubar{\rho}_t(\theta^o,\delta) \right] \big) \stackrel{\mathbb{P}}{\to} 0$, which shows condition \ref{cond:LLNlocally}.
			Consequently, the uniform convergence condition (iv) holds by applying the uniform weak law of large numbers given in Theorem A.2.5. in \cite{White1994}.
			
			As we have shown that the map $\theta \mapsto \rho \big( Y_{t+1}, g^q_t(\beta) , g^e_t(\eta)  \big)$ is Lipschitz-$L_1$  in Lemma S.1 in the supplementary material, the map $\theta \mapsto Q_{n}^0 = \frac{1}{n} \sum_{t=m}^{T-1} \mathbb{E} \left[  \rho \big( Y_{t+1}, g^q_t(\beta) , g^e_t(\eta)  \big) \right]$ is also continuous which shows condition (iii).
			Thus, we can apply Theorem 2.1. of \cite{NeweyMcFadden1994} which concludes the proof of this proposition.
		\end{proof}

		\begin{proof}[Proof of Proposition \ref{prop:AsymptoticNormality}]
			We define $\Psi_{n}(\theta) = \frac{1}{n} \sum_{t=m}^{T-1} \psi \big( Y_{t+1}, g^q_t(\beta) , g^e_t(\eta)  \big)$ and $\Psi_{n}^0(\theta) = \mathbb{E} [ \Psi_{n}(\theta) ]$.
			From the proof of Lemma S.2 in the supplementary material, we get the mean value expansion (for $\hat \theta_n$ close to $\theta^\ast_n$),
			\begin{align}
			\label{eqn:ProofApplicationMeanValue}
			\Psi_{n}^0(\hat \theta_n) - \Psi_{n}^0(\theta_n^\ast) = \Delta_{n}(\tilde \theta_1, \dots, \tilde \theta_k)  \big( \hat \theta_n - \theta^\ast_n \big),
			\end{align} 
			for (possibly different) values $\tilde \theta_1, \dots , \tilde \theta_k$ somewhere on the line between $\hat \theta_n$ and $\theta_n^\ast$, where the components of $\Delta_{n}(\tilde \theta_1, \dots, \tilde \theta_k)$  are given in Lemma S.2 in the supplementary material, and where $\Psi_{n}^{0}(\theta_n^\ast) = 0$.\footnote{The mean-value theorem cannot be generalized in a straight-forward fashion to vector-valued functions. Thus, we have to consider the mean value expansion in each component separately which gives this more complicated expression.}
			
			Furthermore, it holds that $\Delta_{n}(\theta_n^\ast, \dots, \theta_n^\ast) = \Lambda_n(\theta^\ast_n) $ and $\Delta_{n}(\tilde \theta_1, \dots, \tilde \theta_k)$ is a continuous function in its arguments $\tilde \theta_1,\dots, \tilde \theta_k$.
			Using that $\Lambda_n(\theta^\ast_n) $ has Eigenvalues bounded away from zero (for $n$ large enough), we also get that $\Delta_{n}(\tilde \theta_1, \dots, \tilde \theta_k)$ is non-singular in a neighborhood around $\theta^\ast_n$ (for all arguments) for $n$ large enough as the map which maps the matrix onto its Eigenvalues is continuous.
			As we further know that $\hat \theta_n -  \theta^\ast_n \stackrel{\mathbb{P}}{\to} 0$ and $|| \tilde \theta_j - \theta_n^\ast || \le || \hat \theta_n - \theta^\ast_n ||$ for all $j = 1,\dots,k$, we get from the continuous mapping theorem that
			\begin{align}
			\label{eqn:ContinuousMappingDeltaT}
			\Delta_{n}^{-1}(\tilde \theta_1, \dots, \tilde \theta_k) - \Lambda_n^{-1}(\theta^\ast_n)  \stackrel{\mathbb{P}}{\to} 0.
			\end{align}
			
			In the following, we apply Lemma A.1 in \cite{Weiss1991} (by verifying its assumptions), which extends the iid results of \cite{Huber1967} to strong mixing sequences. 
			Assumption (N1) of Lemma A.1 in \cite{Weiss1991} is satisfied as every almost surely continuous stochastic process is separable in the sense of Doob \citep{GikhmanSkorokhod2004} and the functions $\psi \big( Y_{t+1}, g^q_t(\beta) , g^e_t(\eta)  \big)$ are almost surely continuous for all $t \in \mathbb{N}$.
			Assumption (N2) is satisfied as shown in the proof of Proposition \ref{prop:Consistency}.
			Assumption (N3)(i) is shown in Lemma S.2 in the supplementary material.
			The technical Assumptions (N3)(ii) and (N3)(iii) follow from Lemma 4 and Lemma 5 in the supplemental appendix of \cite{Patton2019}.
			For this, notice that the moment conditions in Assumption 2 (C) and (D) of \cite{Patton2019} are implied by the condition \ref{cond:MomentCondition} in Assumption \ref{assu:AsymptoticTheory}.
			Assumption (N4) follows from the moment conditions \ref{cond:MomentCondition} in Assumption \ref{assu:AsymptoticTheory}  and Assumption (N5) from the strong mixing condition \ref{cond:AlphaMixing}. 
			Furthermore,  Lemma 2 in the supplemental appendix of \cite{Patton2019} implies that $\sqrt{n} \Psi_{n}(\hat \theta_n) \stackrel{\mathbb{P}}{\to} 0$.
			Thus, we can apply Lemma A.1 in \cite{Weiss1991} and get that 
			\begin{align}
			\label{eqn:ResultWeiss1991ESReg}
			\sqrt{n} \Psi_{n}^0(\hat \theta_n) - \sqrt{n} \Psi_{n}(\theta^\ast_n) \stackrel{\mathbb{P}}{\to} 0.
			\end{align}
			Combining (\ref{eqn:ProofApplicationMeanValue}), (\ref{eqn:ContinuousMappingDeltaT}) and (\ref{eqn:ResultWeiss1991ESReg}), we get that
			\begin{align}
			&\sqrt{n} \big( \hat \theta_n - \theta^\ast_n \big) = - \Delta_{n}(\tilde \theta_1, \dots, \tilde \theta_k)^{-1} \sqrt{n} \, \Psi_n^0(\hat \theta_n) \\
			= \, &- \left( \Lambda_n^{-1}(\theta^\ast_n)  + o_p(1) \right) \cdot \left( \sqrt{n} \Psi_n(\theta^\ast_n) + o_p(1)\right)
			= - \Lambda_n^{-1}(\theta^\ast_n)  \cdot \sqrt{n} \Psi_n(\theta^\ast_n) + o_p(1).
			\end{align}
			Furthermore, $\Sigma_n^{-1/2}(\theta^\ast_n) \sqrt{n} \Psi_n(\theta^\ast_n) \stackrel{d}{\to} \mathcal{N} \big( 0, I_k \big)$ by Lemma S.3 in the supplementary material and thus, $\Sigma_n^{-1/2}(\theta^\ast_n) \Lambda_n(\theta^\ast_n) \, \sqrt{n} \big( \hat \theta_n - \theta^\ast_n \big) \stackrel{d}{\to} \mathcal{N} \big( 0, I_k \big)$,
			which concludes the proof of this proposition.
		\end{proof}

		\begin{proof}[Proof of Theorem \ref{thm:DistributionEncmpTestStatistics}]
			
			We first notice that
			\begin{align}
			\widehat \Omega_n^{-1/2}  \sqrt{n} \big( \hat \theta_n - \theta_n^\ast \big)
			=  \Omega_n^{-1/2}  \sqrt{n} \big( \hat \theta_n - \theta_n^\ast \big)
			+  \big( \widehat \Omega_n^{-1/2} - \Omega_n^{-1/2} \big)   \sqrt{n} \big( \hat \theta_n - \theta_n^\ast \big).
			\end{align}
			From Proposition \ref{prop:AsymptoticNormality}, we obtain that  $\Omega_n^{-1/2}  \sqrt{n} \big( \hat \theta_n - \theta_n^\ast \big) \stackrel{d}{\to} \mathcal{N} \big( 0, I_k \big)$.
			Furthermore, as $\big( \widehat \Omega_n^{-1/2} - \Omega_n^{-1/2} \big) = o_P(1)$ by assumption, we apply Slutzky's theorem in order to get that $\big( \widehat \Omega_n^{-1/2} - \Omega_n^{-1/2} \big) \sqrt{n} \big( \hat \theta_n - \theta_n^\ast \big) = o_P(1)$.
			Thus, $\widehat \Omega_n^{-1/2}  \sqrt{n} \big( \hat \theta_n - \theta_n^\ast \big) \stackrel{d}{\to} \mathcal{N} \left( 0, I_k \right)$
			and the result based on the selection matrices $R$ follows by applying the continuous mapping theorem, 	which concludes the proof of this theorem.
		\end{proof}

\def\spacingset#1{\renewcommand{\baselinestretch}%
	{#1}\small\normalsize} \spacingset{1}

\renewcommand{\thetable}{\arabic{table}}   
\renewcommand{\thefigure}{\arabic{figure}}



\newpage
\setcounter{page}{1}
\begin{center}
	SUPPLEMENTARY MATERIAL FOR   \vspace{10pt} \\
	{\Large\bf {Forecast Encompassing Tests for the \\ Expected Shortfall} \vspace{10pt} }\\
	\today \\ 
\end{center}

\setcounter{section}{0}	
\setcounter{page}{1}	
\setcounter{table}{0}
\setcounter{figure}{0}

\renewcommand{\thesection}{S.\arabic{section}}   
\renewcommand{\thepage}{S.\arabic{page}}  
\renewcommand{\thetable}{S.\arabic{table}}   
\renewcommand{\thetheorem}{S.\arabic{theorem}}
\renewcommand{\thefigure}{S.\arabic{figure}}  

\onehalfspacing

\section{Additional Data Generating Processes}	
\label{sec:AdditionalDGPs}

In the following, we describe two additional data generating processes (DGPs) used in the extensions of the simulation study in Section 3.3.1.

\subsection*{The GAS DGP}

We introduce two specifications of the GAS models \citep{Creal2013}, where the second candidate potentially generates model misspecification in the strict ES encompassing test.
For this, we generate $\tilde Y_{1,t+1}$, $\hat q_{1,t}$ and $\hat e_{1,t}$ from a GAS model with Gaussian innovations, which corresponds to the standard GARCH specification given in (3.2).
We obtain the second sequence of forecasts from a GAS model with Student-$t$ residuals with time-varying variance and degrees of freedom, given by
\begin{align}
(\hat \mu_2, \hat \sigma_{2,t}^2 , \hat \nu_{2,t})^\top = \kappa + B \cdot (\hat \mu_2, \hat \sigma_{2,t-1}^2 , \hat \nu_{2,t-1})^\top  + A H_t \nabla_t,
\end{align}
where $H_t \nabla_t$ is the forcing variable of the model, the scaling matrix $H_t$ is the Hessian and $\nabla_t$ the derivative of the log-likelihood function.
We calibrate both models to daily IBM returns resulting in the parameter values $\kappa = ( 0.0659,  0.00599, -1.737)$, $A =  \operatorname{diag}(0, 0.146, 7.563)$ and  $B =  \operatorname{diag}(0, 0.994, 7.381)$. 
This model implies that $\tilde Y_{2,t+1} \sim t_{\hat \nu_{2,t}} \big( \hat \mu_2, \hat \sigma_{2,t}^2 \big)$ and we obtain the VaR and ES forecasts from this $t$-distribution.
In order to simulate returns which follow a convex combination of these two distributions, we simulate Bernoulli draws $\pi_{t+1} \sim \operatorname{Bern}(\pi)$ and let $Y_{t+1}  = (1-\pi_{t+1}) \tilde Y_{1,t+1}  + \pi_{t+1} \tilde Y_{2,t+1}$, as for the VaR/ES-GAS models in Section 3.1.

\subsection*{The VaR/ES CAViaR DGP}

This simulation setup follows the dynamic ES models of \cite{Taylor2019}, which we denote by ES-CAViaR as they augment the CAViaR models of \cite{Engle2004} with a dynamic ES specification.
The asymmetric slope AS-ES-CAViaR model is given by
\begin{align}
\label{eqn:AS-ES-CAViaR}
\hat q_{1,t} &= -0.0003 - 0.05 |\tilde Y_{1,t}| \mathds{1}_{\{ \tilde Y_{1,t} \ge 0 \}} - 0.15 |\tilde Y_{1,t}| \mathds{1}_{\{ \tilde Y_{1,t} < 0 \}} + 0.8 \hat q_{1,t-1}, \qquad \text{ and } \\
\hat e_{1,t} &= \hat q_{1,t} - x_{t}, \qquad \text{ where } \label{eqn:ESCAViaRESSpecification} \\
x_{t} &= 
\begin{cases}
0.00017 + 0.125 (\hat q_{1,t-1} - \tilde Y_{1,t}) + 0.84 \hat q_{1,t-1}  \qquad &\text{ if } \hat q_{1,t-1} \le \tilde Y_{1,t}, \\
x_{t-1} \qquad &\text{ if } \hat q_{1,t-1} > \tilde Y_{1,t}.
\end{cases} \label{eqn:ESCAViaRXSpecification}
\end{align}
The second model variant we consider is the symmetric absolute value SAV-ES-CAViaR model, where the quantile equation is given by
\begin{align}
\label{eqn:CARE_SAV}
\hat q_{2,t} &= -0.0003 - 0.1 |\tilde Y_{2,t}| + 0.8 \hat q_{2,t-1},
\end{align}
and $\hat e_{2,t}$ and $x_{t}$ follow the dynamic specifications in  (\ref{eqn:ESCAViaRESSpecification}) and (\ref{eqn:ESCAViaRXSpecification}).
These parameter choices are slightly modified values of the ones obtained by \cite{Taylor2019}.
In this setup, we simulate data according to the additive model $Y_{t+1} = \big( (1-\pi) \hat e_{1,t} + \pi \hat e_{2,t}  \big) + \varepsilon_{t+1}$, where $\varepsilon_{t+1}\sim \mathcal{N} \big(-\sigma \xi_\alpha, \sigma^2 \big)$, for $\sigma = 0.1$.
This implies that for $\pi = 0$, $ES_\alpha(Y_{t+1} | \mathcal{F}_{t}) = \hat e_{1,t}$ almost surely, and the same holds inversely for $\pi = 1$.
This setup generalizes the CAViaR DGP used in the simulations for the VaR encompassing test of \cite{GiacominiKomunjer2005} to the ES.

\section{Technical Proofs}

\allowdisplaybreaks

\begin{lemma}
	\label{lemma:LipschitzL1}
	Given the conditions from Assumption 2.7, the function $\rho \big( Y_{t+1}, g^q_t(\beta) , g^e_t(\eta)  \big) $ is Lipschitz-$L_1$ on $\Theta$ with $\mathcal{F}_t$-measurable and integrable Lipschitz-constant.  
\end{lemma}

\begin{proof}
	We split the $\rho$-function $\rho \big( Y_{t+1}, g^q_t(\beta) , g^e_t(\eta)  \big) = \rho_1 \big( Y_{t+1}, g^q_t(\beta) , g^e_t(\eta)  \big)  + \rho_2 \big( Y_{t+1}, g^q_t(\beta) , g^e_t(\eta)  \big) $, where
	\begin{align*} 
	\rho_1 \big( Y_{t+1}, g^q_t(\beta) , g^e_t(\eta)  \big)  &= - \mathds{1}_{\{Y_{t+1} \le g^q_t(\beta) \}}  \frac{1}{\alpha g^e_t(\eta) } ( g^q_t(\beta)  - Y_{t+1}) , \\
	\rho_2 \big( Y_{t+1}, g^q_t(\beta) , g^e_t(\eta)  \big) &=  \frac{g^q_t(\beta)  - g^e_t(\eta)  }{g^e_t(\eta)}  - \log(-g^e_t(\eta) ).
	\end{align*}	
	Local Lipschitz continuity of $\rho_2$ follows since it is a continuously differentiable function in $\theta$ (such that $g^e_t(\eta) \not= 0$) and thus (locally) Lipschitz-$L_1$.
	We consequently get that for all $\theta^o \in \Theta$, there exists a $\delta^o > 0$ such that for all $\theta \in U_{\delta^o}(\theta^o) := \big\{ \theta \in \Theta \big| || \theta - \theta^o|| \le \delta^o \big\}$, it holds that
	\begin{align}
	\begin{aligned}
	&\big| \rho_2 \big( Y_{t+1}, g^q_t(\beta^o) , g^e_t(\eta^o)  \big)  - \rho_2 \big( Y_{t+1}, g^q_t(\beta) , g^e_t(\eta)  \big)  \big| \\
	\le \, &\big|\big| \theta - \theta^o \big|\big|
	\cdot \sup_{\theta \in U_{\delta^o}(\theta^o)} \left( \left| \left| \frac{\nabla_\beta g^q_t(\beta)+ \nabla_\eta g^e_t(\eta)}{ g^e_t(\eta)} \right|\right| + \left| \left| \frac{ g^q_t(\beta) \nabla_\eta g^e_t(\eta) }{( g^e_t(\eta))^2} \right|\right| \right),
	\end{aligned}
	\end{align}
	where the sequences $\frac{1}{n} \sum_{t=m}^{T-1}  \mathbb{E} \left[  \left| \left| \frac{\nabla_\beta g^q_t(\beta)+ \nabla_\eta g^e_t(\eta)}{ g^e_t(\eta)} \right|\right| \right]$ and $\frac{1}{n} \sum_{t=m}^{T-1}  \mathbb{E} \left[  \left| \left| \frac{ g^q_t(\beta) \nabla_\eta g^e_t(\eta) }{( g^e_t(\eta))^2} \right|\right| \right]$ are bounded for all $\theta^o \in \Theta$ by the conditions (h) in Assumption 2.7.
	
	For the function $\rho_1$, we consider the following four cases.
	First, let $\Gamma_1 = \big\{ \omega \in \Omega, \theta \in U_{\delta^o}(\theta^o) \, \big| \, g^q_t(\beta^o)(\omega) < Y_{t+1}(\omega) \; \text{ and } \; g^q_t(\beta)(\omega) < Y_{t+1}(\omega) \big\}$.
	Then, on $\Gamma_1$, it holds that,
	\begin{align}
	\rho_1 \big( Y_{t+1}, g^q_t(\beta) , g^e_t(\eta)  \big)  = \rho_1 \big( Y_{t+1}, g^q_t(\beta^o) , g^e_t(\eta^o)  \big)  = 0,
	\end{align}
	which is Lipschitz-$L_1$. 
	
	Second, let $\Gamma_2 = \big\{ \omega \in \Omega, \theta \in U_{\delta^o}(\theta^o) \, \big| \, g^q_t(\beta^o)(\omega) \ge Y_{t+1}(\omega) \; \text{ and } \; g^q_t(\beta)(\omega) \ge Y_{t+1}(\omega) \big\}$.
	On $\Gamma_2$, for both $\tilde \theta \in \{ \theta, \theta^o\}$, it holds that
	\begin{align}
	\rho_1 \big( Y_{t+1}, g^q_t(\tilde \beta) , g^e_t(\tilde \eta)  \big)  = - \frac{1}{\alpha g^e_t(\tilde \eta)}  \big( g^q_t(\tilde \beta) - Y_{t+1} \big),
	\end{align}
	which is a continuously differentiable function. 
	Thus,
	\begin{align}
	\begin{aligned}
	&\big| \rho_1 \big( Y_{t+1}, g^q_t( \beta^o) , g^e_t( \eta^o)  \big) - \rho_1 \big( Y_{t+1}, g^q_t( \beta) , g^e_t( \eta)  \big) \big| \\
	\le \, &\big|\big| \theta^o - \theta \big|\big|
	\cdot  \left( \sup_{\theta \in U_{\delta^o}(\theta^o)} \left|\left|  \frac{ \nabla_\beta g^q_t(\beta) }{\alpha g^e_t(\eta)}   \right|\right|  +  \sup_{\theta \in U_{\delta^o}(\theta^o)}  \left|\left|  \frac{ \nabla_\eta g^e_t(\eta) }{\alpha (g^e_t(\eta) )^2} (g^q_t(\beta) - Y_{t+1})   \right|\right| \right),
	\end{aligned}
	\end{align}
	where the average of the expectations of the suprema sequences in the last two lines are bounded by the conditions  (h) in Assumption 2.7.
	
	Finally, let $\Gamma_3 = \big\{ \omega \in \Omega, \theta \in U_{\delta^o}(\theta^o)  \, \big| \,  g^q_t( \beta)(\omega) < Y_{t+1}(\omega) \le g^q_t( \beta^o)(\omega)  \big\}$.
	As on $\Gamma_3$, $|g^q_t( \beta^o)  - Y_{t+1} | \le | g^q_t( \beta^o)  - g^q_t( \beta)  | $ almost surely, it holds that
	\begin{align*}
	&\big| \rho_1 \big( Y_{t+1}, g^q_t( \beta^o) , g^e_t( \eta^o)  \big) - \rho_1 \big( Y_{t+1}, g^q_t( \beta) , g^e_t( \eta)  \big)  \big| 
	= \left| \frac{1}{\alpha g^e_t( \eta^o) } ( g^q_t( \beta^o) - Y_{t+1}) \right| \\
	\le \,&\left| \frac{1}{\alpha g^e_t( \eta^o) } ( g^q_t( \beta^o)  - g^q_t( \beta)  ) \right| 
	\le \, \big|\big| \theta - \theta^o \big|\big|
	\cdot \sup_{\theta \in U_{\delta^o}(\theta^o)} \left| \left| \frac{\nabla_\beta g^q_t( \beta) }{\alpha g^e_t( \eta) } \right|\right|.
	\end{align*}
	Equivalently as above, the average of the expectations of the suprema sequences in the last two lines are bounded by the conditions (h) in Assumption 2.7.
	An equivalent argument holds for $\Gamma_4 = \big\{ \omega \in \Omega, \theta \in U_{\delta^o}(\theta^o)  \, \big| \,   g^q_t( \beta^o)(\omega) < Y_{t+1}(\omega) \le   g^q_t( \beta)(\omega) \big\}$.
	As $\Omega = \bigcup_{i=1}^4 \Gamma_i$, we can conclude that the function $\rho_1 \big( Y_{t+1}, g^q_t( \beta^o) , g^e_t( \eta^o)  \big)$ is Lipschitz-$L_1$ on $\Theta$.
\end{proof}

\begin{lemma}
	\label{lemma:WeissConditionN3iESReg}
	Given the conditions from Assumption 2.7, there exist constants $a, d_0 > 0$ such that
	\begin{align}
	\big| \big| \Psi_{n}^0(\theta) \big| \big| \ge a || \theta - \theta_n^\ast || \qquad \text{ for any }  \theta \in \Theta \text{ such that }  ||\theta - \theta_n^\ast || \le d_0,
	\end{align}
	and for all $n \ge n_0$, where $n_0 \in \mathbb{N}$ is large enough.
\end{lemma}

\begin{proof}
	Let $\theta \in \Theta$ such that $||\theta - \theta_n^\ast || \le d_0$ for some (small) constant $d_0 > 0$ and define
	\begin{align}
	\Psi^0_{n,q}(\theta) &= \frac{1}{n} \sum_{t=m}^{T-1}  \mathbb{E} \left[ - \frac{\nabla_\beta g^q_t( \beta)}{\alpha g^e_t( \eta)} \big( F_t(g^q_t( \beta)) - \alpha \big) \right] \qquad \text{ and } \\
	\Psi^0_{n,e}(\theta) &=  \frac{1}{n} \sum_{t=m}^{T-1}  \mathbb{E} \left[ \frac{\nabla_\eta g^e_t( \eta)}{(g^e_t( \eta))^2} \left( g^e_t( \eta)  - g^q_t( \beta)+ \frac{1}{ \alpha}  (g^q_t( \beta)- Y_{t+1}) \mathds{1}_{\{Y_{t+1} \le g^q_t( \beta) \}} \right) \right],
	\end{align}
	such that $\Psi^0_{n}(\theta)^\top  = \left( 	\Psi^0_{n,q}(\theta)^\top, \Psi^0_{n,e} (\theta)^\top  \right)$.
	Henceforth, we use the following short notations
	\begin{align}
	G_t^{q}(\beta) &= \nabla_\beta g^q_t( \beta) \nabla_\eta g^q_t( \beta)^\top \\
	G_t^{qe}(\beta,\eta) &= \nabla_\beta g^q_t( \beta) \nabla_\eta g^e_t( \eta)^\top \\
	G_t^{eq}(\beta,\eta) &= \nabla_\eta g^e_t( \eta) \nabla_\beta g^q_t( \beta)^\top \\
	G_t^{e}(\eta) &= \nabla_\eta g^e_t( \eta) \nabla_\eta g^e_t( \eta)^\top,
	\end{align}
	$H^{q}_t( \beta)$ is the $k_\beta \times k_\beta$ Hessian matrix of $g^q_t( \beta)$ and equivalently, $H^{e}_t( \eta)$ is the $k_\eta \times k_\eta$ Hessian matrix of $g^e_t( \eta)$.
	
	In the following, we apply the mean-value theorem to the individual rows of $\Psi^0_{n}(\theta)$ instead of to the complete vector, as the mean-value theorem cannot be generalized directly to vector-valued functions.
	Then, by applying the mean-value theorem to the $j$-th row of of $\Psi^0_{n}(\theta)$ for all $j=1,\dots,k$, we get that
	\begin{align}
	\Psi^0_{n}(\theta) - \Psi^0_{n}(\theta^\ast_n) = \Delta_{n}(\tilde \theta_1, \dots, \tilde \theta_k) \cdot \big( \theta - \theta^\ast_n \big),
	\end{align}
	where 
	\begin{align}
	\Delta_{n}(\tilde \theta_1, \dots, \tilde \theta_k) 
	= \begin{pmatrix}
	\Delta_{n,qq} & \Delta_{n,qe}  \\
	\Delta_{n,eq}  & \Delta_{n,ee}
	\end{pmatrix}.
	\end{align}
	For all $j=1,\dots,k_\beta$, the $j$-th row of $\Delta_{n,qq}$ is given by
	\begin{align}
	\Delta_{n,qq,j}(\tilde \beta_j) =  \frac{1}{n} \sum_{t=m}^{T-1} \mathbb{E}   \left[ 
	\frac{H^q_{t,j} (\tilde \beta_j)}{\alpha g_t^e(\tilde \eta_j)} \big(F_t(g_t^q(\tilde \beta_j)) - \alpha \big)
	+ \frac{ G_t^q(\tilde \beta_j)}{ \alpha g_t^e(\tilde \eta_j)} {h_t}(g_t^q(\tilde \beta_j)) \right],
	\end{align}
	where $H^q_{t,j} (\tilde \beta_j)$ denotes the $j$-th row of $H^q_{t}(\tilde \beta_j)$,
	and the $j$-th row of $\Delta_{n,qe}$ is given by
	\begin{align}
	\Delta_{n,qe,j}(\tilde \theta_j) = \frac{1}{n} \sum_{t=m}^{T-1} \mathbb{E} \left[ 
	-\frac{ G_{t,j}^{qe}(\tilde \beta_j, \tilde \eta_j)}{\alpha g_t^e(\tilde \eta_j)^2} \big(F_t(g_t^q(\tilde \beta_j)) - \alpha \big) \right].
	\end{align}
	For all $j=k_\beta+1,\dots,k_\beta + k_\eta$, the $j$-th row of $\Delta_{n,eq}$ is given by
	\begin{align}
	\Delta_{n,eq,j}(\tilde \theta_j) =  \frac{1}{n} \sum_{t=m}^{T-1} \mathbb{E}   \left[ 
	\frac{G_{t,j}^{eq}(\tilde \beta_j, \tilde \eta_j)}{\alpha g^e_t( \tilde \eta_j)^2} \big( F_t(g_t^q(\tilde \beta_j)) - \alpha \big) \right]
	\end{align}
	and the $j$-th row of $\Delta_{n,ee}$ is given by
	\begin{align*}
	\Delta_{n,ee,j}(\tilde \theta_j) =  \frac{1}{n} \sum_{t=m}^{T-1} \mathbb{E}   \left[ 
	\frac{H^e_{t,j}(\tilde \eta_j)}{g^e_t(\tilde \eta_j)^2} \left( g^e_t( \tilde \eta_j)  - g^q_t( \tilde \beta_j)+ \frac{1}{ \alpha}  (g^q_t( \tilde \beta_j)- Y_{t+1}) \mathds{1}_{\{Y_{t+1} \le g^q_t( \tilde \beta_j) \}} \right) \right. \\
	\left. + \frac{G_{t,j}^{ee}(\tilde \eta_j)}{g^e_t(\tilde \eta_j)^2} 
	-2 \frac{G_{t,j}^{ee}(\tilde \eta_j)}{g^e_t(\tilde \eta_j)^3} \left( g^e_t( \tilde \eta_j)  - g^q_t( \tilde \beta_j)+ \frac{1}{ \alpha}  (g^q_t( \tilde \beta_j)- Y_{t+1}) \mathds{1}_{\{Y_{t+1} \le g^q_t( \tilde \beta_j) \}} \right) \right].
	\end{align*}
	In the following, we show that $\left| \left| \Delta_{n} \big(\tilde \theta_1, \dots, \tilde \theta_{k} \big) - \Lambda_{n}(\theta^\ast_n)  \right| \right| \le c_1 || \theta - \theta^\ast_n ||$ by considering the individual components again.
	For each $j,i = 1,\dots,n_\beta$, (corresponding to the upper-left quantile-specific part of the Hessian matrix)
	\begin{align*}
	&|| \Delta_{n,ji} \big(\tilde \theta_j \big) - \Lambda_{n,ji}(\theta^\ast_n) || \\
	= \, & \left| \frac{1}{n} \sum_{t=m}^{T-1} \mathbb{E} \left[ \frac{H^q_{t,ji} (\tilde \beta_j)}{\alpha g_t^e(\tilde \eta_j)} \big(F_t(g_t^q(\tilde \beta_j)) - \alpha \big) + \frac{ G_{t,ji}^q(\tilde \beta_j)}{\alpha g_t^e(\tilde \eta_j)} {h_t}(g_t^q(\tilde \beta_j)) \right] \right. \\
	&\quad - \, \left. \frac{1}{n} \sum_{t=m}^{T-1} \mathbb{E}   \left[ 
	\frac{H^q_{t,ji} ( \beta^\ast_n)}{\alpha g_t^e(\eta^\ast_n)} \big(F_t(g_t^q(\beta^\ast_n)) - \alpha \big)
	+ \frac{ G_{t,ji}^q( \beta^\ast_n)}{\alpha g_t^e(\eta^\ast_n)} {h_t}(g_t^q( \beta^\ast_n)) \right] \right| \\
	= \, & \left| \left| \frac{1}{n} \sum_{t=m}^{T-1} \mathbb{E} \left[ \frac{\nabla H^q_{t,ji} ( \bar \beta_j )}{\alpha g_t^e(\bar \eta_j)} \big(F_t(g_t^q(\bar \beta_j)) - \alpha \big)  
	- \nabla g_t^e(\bar \eta_j) \frac{1}{\alpha g_t^e(\bar \eta_j)^2} H^q_{t,ji}( \bar \beta_j ) \big(F_t(g_t^q(\bar \beta_j)) - \alpha \big)  \right. \right. \right. \\
	&\quad \qquad \qquad + \left. \left. \left. \nabla_\beta g_t^q(\bar \beta_j) \frac{H^q_{t,ji} (\bar \beta_j)}{\alpha g_t^e(\bar \eta_j)} {h_t}(g_t^q(\bar \beta_j)) \right] \right| \right| \cdot \left| \left| \tilde \theta_j -  \theta_n^\ast \right|\right|,
	\end{align*}
	for some $\bar \theta_j = \big( \bar \beta_j,  \bar \eta_j \big)$ on the line between $\tilde \theta_j$ and $\theta^{\ast}_n$.
	Furthermore, for all $j = 1,\dots,n_\beta$ and $i = n_\beta+1,\dots,n_\beta+n_\eta$ (corresponding to the upper-right quantile/ES-specific part of the Hessian matrix), it holds that
	\begin{align*}
	&|| \Delta_{n,ji} \big(\tilde \theta_j \big) - \Lambda_{n,ji}(\theta^\ast_n) || \\
	= \, & \left| \frac{1}{n} \sum_{t=m}^{T-1} \mathbb{E} \left[ 
	\frac{ G_{t,ji}^{qe}(\tilde \beta_j, \tilde \eta_j)}{\alpha g_t^e(\tilde \eta_j)^2} \big(F_t(g_t^q(\tilde \beta_j)) - \alpha \big) \right]
	- \frac{1}{n} \sum_{t=m}^{T-1} \mathbb{E} \left[  \frac{ G_{t,ji}^{qe}( \beta^\ast_n, \eta^\ast_n)}{\alpha g_t^e(\eta^\ast_n)^2} \big(F_t(g_t^q( \beta^\ast_n)) - \alpha \big)  \right] \right| \\
	= \, & \left| \left| \frac{1}{n} \sum_{t=m}^{T-1} \mathbb{E} \left[ \frac{ \nabla G_{t,ji}^{qe}(\bar \beta_j, \bar \eta_j)}{\alpha g_t^e(\bar \eta_j)^2} \big(F_t(g_t^q(\bar \beta_j)) - \alpha \big) 
	+ \nabla g_t^q(\bar \beta_j) \frac{ G_{t,ji}^{qe}(\bar \beta_j, \bar \eta_j)}{\alpha g_t^e(\bar \eta_j)^2} {h_t}(g_t^q(\bar \beta_j)) \right. \right. \right. \\
	&\quad \qquad \qquad -2 \left. \left. \left. \nabla g_t^e(\bar \eta_j) \frac{ G_{t,ji}^{qe}(\bar \beta_j, \bar \eta_j)}{\alpha g_t^e(\bar \eta_j)^3} \big(F_t(g_t^q(\bar \beta_j)) - \alpha \big) \right] \right| \right| \cdot \left| \left| \tilde \theta_j -  \theta_n^\ast \right|\right|,
	\end{align*}
	for some $\bar \theta_j = \big( \bar \beta_j,  \bar \eta_j \big)$ on the line between $\tilde \theta_j$ and $\theta^{\ast}_n$.
	This holds equivalently for the lower-left block of $\Delta_n$ and $\Lambda_n$.
	Eventually for the lower-right block, i.e.\ for each $j,i = n_\beta+1,\dots,n_\beta+n_\eta$, we get that
	\begin{align*}
	&\big| \Delta_{n,ji} \big(\tilde \theta_j \big) - \Lambda_{n,ji}(\theta^\ast_n) \big| \\
	= \, & \left| \frac{1}{n} \sum_{t=m}^{T-1} \mathbb{E} \left[ 
	\frac{H^e_{t,ji}(\tilde \eta_j)}{(g^e_t(\tilde \eta_j))^2} \left( g^e_t( \tilde \eta_j)  - g^q_t( \tilde \beta_j)+ \frac{1}{ \alpha}  (g^q_t( \tilde \beta_j)- Y_{t+1}) \mathds{1}_{\{Y_{t+1} \le g^q_t( \tilde \beta_j) \}} \right) \right. \right. \\
	& \qquad \left. + \frac{G_{t,ji}^{ee}(\tilde \eta_j)}{(g^e_t(\tilde \eta_j))^2} 
	-2 \frac{G_{t,ji}^{ee}(\tilde \eta_j)}{(g^e_t(\tilde \eta_j))^3} \left( g^e_t( \tilde \eta_j)  - g^q_t( \tilde \beta_j)+ \frac{1}{ \alpha}  (g^q_t( \tilde \beta_j)- Y_{t+1}) \mathds{1}_{\{Y_{t+1} \le g^q_t( \tilde \beta_j) \}} \right) \right] \\
	& - \frac{1}{n} \sum_{t=m}^{T-1} \mathbb{E} \left[ 
	\frac{H^e_{t,ji}(\eta^\ast_n)}{(g^e_t( \eta^\ast_n))^2} \left( g^e_t(  \eta^\ast_n)  - g^q_t(  \beta^\ast_n)+ \frac{1}{ \alpha}  (g^q_t(\beta^\ast_n)- Y_{t+1}) \mathds{1}_{\{Y_{t+1} \le g^q_t( \beta^\ast_n) \}} \right) \right.  \\
	& \qquad \left. \left. + \frac{G_{t,ji}^{ee}(\eta^\ast_n)}{(g^e_t( \eta^\ast_n))^2} 
	-2 \frac{G_{t,ji}^{ee}(\eta^\ast_n)}{(g^e_t( \eta^\ast_n))^3} \left( g^e_t( \eta^\ast_n)  - g^q_t(  \beta^\ast_n)+ \frac{1}{ \alpha}  (g^q_t(\beta^\ast_n)- Y_{t+1}) \mathds{1}_{\{Y_{t+1} \le g^q_t( \beta^\ast_n) \}} \right) \right] \right| \\
	= \, & \left| \left| \frac{1}{n} \sum_{t=m}^{T-1} \mathbb{E} \left[  
	\left\{ \frac{\nabla H^e_{t,ji}(\bar \eta_j)}{(g^e_t(\bar \eta_j))^2}  - 2 \nabla g^e_t(\bar \eta_j) \frac{ H^e_{t,ji}(\bar \eta_j)}{(g^e_t(\bar \eta_j))^3} - 2\frac{\nabla G^{ee}_{t,ji}(\bar \eta_j)}{(g^e_t(\bar \eta_j))^3} + 6 \nabla g^e_t(\bar \eta_j) \frac{ G^{ee}_{t,ji}(\bar \eta_j)}{(g^e_t(\bar \eta_j))^4}\right\} \times \right. \right. \right. \\
	&\qquad \qquad \quad \left\{  g^e_t( \bar \eta_j)  - g^q_t( \bar \beta_j)+ \frac{1}{ \alpha}  (g^q_t( \bar \beta_j)- Y_{t+1}) \mathds{1}_{\{Y_{t+1} \le g^q_t( \bar \beta_j) \}} \right\}  \\
	&\qquad \qquad + \frac{ \nabla G^{ee}_{t,ji}(\bar \eta_j)}{(g^e_t(\bar \eta_j))^2} - 2 \nabla g^e_t(\bar \eta_j) \frac{ G^{ee}_{t,ji}(\bar \eta_j)}{(g^e_t(\bar \eta_j))^3} \\
	&\qquad \qquad + \left. \left. \left. \left\{ \frac{ H^e_{t,ji}(\bar \eta_j)}{(g^e_t(\bar \eta_j))^2} - 2 \frac{ G^{ee}_{t,ji}(\bar \eta_j)}{(g^e_t(\bar \eta_j))^2} \right\} \cdot \left\{  \nabla g^e_t(\bar \eta_j) - \nabla g^q_t(\bar \beta_j) + \frac{1}{\alpha} \nabla g^q_t(\bar \beta_j) F_t(\nabla g^q_t(\bar \beta_j)) \right\} \right] \right| \right| \\
	&\qquad \qquad \cdot \left| \left| \tilde \theta_j -  \theta_n^\ast \right|\right|.
	\end{align*}	
	for some $\bar \theta_j = \big( \bar \beta_j,  \bar \eta_j \big)$ on the line between $\tilde \theta_j$ and $\theta^{\ast}_n$.
	As the respective moments are finite given the moment conditions in (h) in Assumption 2.7 and since $|| \tilde \theta_j - \theta_n^\ast || \le || \theta - \theta_n^\ast ||$ for all $j$, we have shown that for all $n$ sufficiently large enough, there exists a constant $c_1 > 0$ such that
	\begin{align}
	\left| \left| \Delta_{n} \big(\tilde \theta_1, \dots, \tilde \theta_{k} \big) - \Lambda_{n}(\theta^\ast_n)  \right| \right| \le c_1 || \theta - \theta^\ast_n ||.
	\end{align}
	Furthermore, as the matrix $\Lambda_n(\theta^\ast_n)$ has Eigenvalues bounded from below (for $n$ large enough) by assumption, there exists a constant $c_2 > 0$, such that
	\begin{align}
	\left| \left| \Lambda_n(\theta^\ast_n)  \cdot (\theta - \theta_n^\ast) \right| \right| \ge c_2 ||\theta - \theta_n^\ast ||.
	\end{align}
	Thus, we choose $d_0 > 0$ small enough such that $d_0 < \frac{c_2}{2 c_1}$.
	Then $|| \theta - \theta_n^\ast|| \le d_0 < \frac{c_2}{2 c_1}$ and thus, $2c_1 ||\theta - \theta_n^\ast||^2 \le c_2 ||\theta - \theta_n^\ast ||$.
	Consequently, $\left| \left|\big( \Delta_{n} \big(\tilde \theta_1, \dots, \tilde \theta_{k} \big)  - \Lambda_n(\theta^\ast_n)  \big) \cdot (\theta - \theta_n^\ast) \right| \right| \le c_1 ||\theta - \theta_n^\ast||^2 \le c_2/2 ||\theta - \theta_n^\ast||$ and thus 
	\begin{align}
	\begin{aligned}
	\big| \big| \Psi_{n}^0(\theta) \big| \big| 
	&= \big| \big| \Delta_{n} \big(\tilde \theta_1, \dots, \tilde \theta_{k} \big) \cdot (\theta - \theta_n^\ast)\big| \big| \\
	&= \left| \left| \Lambda_n(\theta^\ast_n)  \cdot (\theta - \theta_n^\ast) + \big( \Delta_{n} \big(\tilde \theta_1, \dots, \tilde \theta_{k} \big)  - \Lambda_n(\theta^\ast_n)  \big) \cdot (\theta - \theta_n^\ast )\right| \right| \\
	&\ge \Big| \left| \left| \Lambda_n(\theta^\ast_n)  \cdot (\theta - \theta_n^\ast) \right| \right| - \left| \left| \big( \Delta_{n} \big(\tilde \theta_1, \dots, \tilde \theta_{k} \big)  - \Lambda_n(\theta^\ast_n)  \big) \cdot (\theta - \theta_n^\ast ) \right| \right| \Big| \\
	&\ge \frac{c_2}{2} ||\theta - \theta_n^\ast||,
	\end{aligned}
	\end{align}
	by applying the mean value expansion and the inverse triangular inequality.
\end{proof}

\begin{lemma}
	\label{lemma:AsymptoticNormalityPsiESReg}
	Given the conditions in Assumption 2.7 it holds that
	\begin{align}
	\Sigma_n^{-1/2}(\theta_n^\ast) \,\sqrt{n} \, \Psi_n(\theta^\ast_n) \stackrel{d}{\to} \mathcal{N}(0,I_k).
	\end{align}
\end{lemma}

\begin{proof}
	We show this multivariate result by applying the Cramér–Wold theorem, i.e.\ by showing that the conditions for the univariate CLT for strong mixing sequences given in Theorem 5.20 in \cite{White2001}, p.\ 130 hold for all linear combinations $u^\top  \psi \big( Y_{t+1}, g^q_t(\beta_n^\ast) , g^e_t(\eta_n^\ast)  \big)$ for all $u \in \mathbb{R}^k$ such that $||u|| = 1$.	
	By Theorem 3.49 in \cite{White2001} p.\ 50, we get that the sequences $\psi \big( Y_{t+1}, g^q_t(\beta_n^\ast) , g^e_t(\eta_n^\ast)  \big)$ and $u^\top \psi \big( Y_{t+1}, g^q_t(\beta_n^\ast) , g^e_t(\eta_n^\ast)  \big)$ are strong mixing of size $-r/(r-2)$ for some $r > 2$.	
	Furthermore, for all $t \in \mathbb{N}$, it holds that 
	\begin{align*}
	&\mathbb{E} \left[ \big| u^\top \psi \big( Y_{t+1}, g^q_t(\beta_n^\ast) , g^e_t(\eta_n^\ast)  \big) \big|^{r} \right] 
	\le \mathbb{E} \left[ \big|\big| \psi \big( Y_{t+1}, g^q_t(\beta_n^\ast) , g^e_t(\eta_n^\ast)  \big) \big| \big|^{r} \right] \\
	\le \, &4^{r-1} \left\{
	\max \left( \frac{1-\alpha}{\alpha},1\right)^r \mathbb{E} \left[   \left| \left| \frac{\nabla_\beta g^q_t(\beta_n^\ast)}{g^e_t(\eta_n^\ast)} \right| \right|^r \right]
	+ \mathbb{E} \left[   \left| \left| \frac{\nabla_\eta  g^e_t(\eta_n^\ast) g^e_t(\eta_n^\ast)}{(g^e_t(\eta_n^\ast))^2} \right| \right|^r \right] \right. \\
	&\qquad \qquad + \left. \left( 1 + \frac{1}{\alpha}\right)^r  \mathbb{E} \left[   \left| \left| \frac{\nabla_\eta g^e_t(\eta_n^\ast) g^q_t(\beta_n^\ast)}{(g^e_t(\eta_n^\ast))^2} \right| \right|^r \right]
	+ \mathbb{E} \left[   \left| \left| \frac{\nabla_\eta g^e_t(\eta_n^\ast) Y_{t+1}}{\alpha (g^e_t(\eta_n^\ast))^2} \right| \right|^r \right]
	\right\} \\
	\le \, &4^{r-1} \left\{
	\max \left( \frac{1-\alpha}{\alpha},1\right)^r \frac{1}{K^{r}} \mathbb{E} \left[   \left| \left| \nabla_\beta g^q_t(\beta_n^\ast) \right| \right|^r \right]
	+ \frac{1}{K^{r}} \mathbb{E} \left[ \left| \left| \nabla_\eta g^e_t(\eta_n^\ast) \right| \right|^r \right] \right. \\
	&\qquad \qquad + \left.  \frac{1}{K^{2r}} \left( 1 + \frac{1}{\alpha}\right)^r  \mathbb{E} \left[   \left| \left| \nabla_\eta g^e_t(\eta_n^\ast) g^q_t(\beta_n^\ast) \right| \right|^r \right]
	+ \frac{1}{\alpha K^{2r}} \mathbb{E} \left[   \left| \left| \nabla_\eta g^e_t(\eta_n^\ast) Y_{t+1} \right| \right|^r \right]
	\right\} < \infty,
	\end{align*}
	by applying Jensen's inequality and by the moment conditions (h) in Assumption 2.7, where $r > 2$ (from condition (a)).
	As the sequence $\psi \big( Y_{t+1}, g^q_t(\beta_n^\ast) , g^e_t(\eta_n^\ast)  \big)$ is uncorrelated by condition (c) in Assumption 2.7, we get that
	for all $n \ge 1$,
	\begin{align}
	\begin{aligned}
	&\operatorname{Var} \left( \frac{1}{\sqrt{n}} \sum_{t=m}^{T-1} \psi \big( Y_{t+1}, g^q_t(\beta_n^\ast) , g^e_t(\eta_n^\ast)  \big)  \right) \\
	= \, &\frac{1}{n} \sum_{t=m}^{T-1} \mathbb{E} \left[ \psi \big( Y_{t+1}, g^q_t(\beta_n^\ast), g^e_t(\eta_n^\ast) \big)  \cdot \psi \big( Y_{t+1}, g^q_t(\beta_n^\ast), g^e_t(\eta_n^\ast) \big)^\top \right] = \Sigma_n(\theta_n^\ast).
	\end{aligned}
	\end{align}
	As $\Sigma_n(\theta_n^\ast)$ is real and symmetric and positive definite, it can be diagonalized with a real orthogonal matrix $S$, i.e.\ $S^\top \Sigma_n(\theta_n^\ast) S = D_n$, where $D_n$ is a diagonal matrix containing the Eigenvalues of $\Sigma_n(\theta_n^\ast)$, denoted by  $\{ \lambda_{1,n},\dots,\lambda_{k,n}\}$.
	Consequently, for any $u \in \mathbb{R}^k$,
	\begin{align}
	\begin{aligned}
	\operatorname{Var} \left( \frac{1}{\sqrt{n}} \sum_{t=m}^{T-1} u^\top \psi \big( Y_{t+1}, g^q_t(\beta_n^\ast), g^e_t(\eta_n^\ast) \big) \right) &= u^\top \Sigma_n( \theta_n^\ast ) u 
	= u^\top S^\top D_n S u 
	=  v^\top D_n v \\
	&> \min_{i=1,\dots,k} \lambda_{i,n},
	\end{aligned}
	\end{align}
	where $v = S u$, i.e.\ $||v|| = 1$ as $S$ is orthogonal and where the Eigenvalues $\{ \lambda_{1,n},\dots,\lambda_{k,n}\}$ are bounded away from zero for $n$ sufficiently large.
	Thus, we can apply Theorem 5.20 in \cite{White2001} p.\ 130 for asymptotic normality of the sequences $u^\top\psi \big( Y_{t+1}, g^q_t(\beta_n^\ast), g^e_t(\eta_n^\ast) \big)$ for all $u \in \mathbb{R}^k$ such that $||u|| = 1$. Applying the Cramér-Wold theorem concludes the proof.
\end{proof}

\pagebreak
\section{Additional Tables}
\label{sec:Tables}

\begin{table}[ht]
	\footnotesize
	\centering
	\caption{Empirical Sizes of the Forecast Encompassing Tests.}
	\label{tab:Size1perc}
	\begin{tabularx}{0.95\linewidth}{ll @{\hspace{0.7cm}} RRRR l @{\hspace{0.5cm}} RRRR }
		\toprule
		& & \multicolumn{4}{c}{$\mathbb{H}_0^{(1)}$} & & \multicolumn{4}{c}{$\mathbb{H}_0^{(2)}$} \\
		\cmidrule(lr){3-6} \cmidrule(lr){8-11}
		& & Str ES & Aux ES & VaR ES & VaR & & Str ES & Aux ES & VaR ES & VaR \\ 
		\midrule
		\midrule
		$n$ & & \multicolumn{9}{c}{GARCH} \\
		\cmidrule(lr){3-11}
		$500$  &        &  3.05  &  3.00  &  8.40  & 10.10  &        &  2.80  &  2.80  &  8.20  & 10.10 \\
$1000$ &        &  1.45  &  1.75  &  5.90  &  7.80  &        &  2.20  &  1.95  &  7.70  &  9.30 \\
$2500$ &        &  1.85  &  1.80  &  5.85  &  7.30  &        &  1.85  &  1.75  &  5.45  &  6.45 \\
$5000$ &        &  1.25  &  1.25  &  3.90  &  5.05  &        &  0.80  &  0.80  &  4.10  &  5.15 \\

		\midrule
		\midrule
		$n$ & & \multicolumn{9}{c}{VaR/ES GAS} \\
		\cmidrule(lr){3-11}
		$500$  &        &  5.65  &  5.30  &  9.40  & 11.20  &        &  4.40  &  4.20  &  9.20  & 11.50 \\
$1000$ &        &  4.15  &  4.05  &  6.65  &  7.75  &        &  3.55  &  3.25  &  6.55  &  8.40 \\
$2500$ &        &  2.70  &  2.65  &  4.80  &  5.80  &        &  1.35  &  1.45  &  4.70  &  5.95 \\
$5000$ &        &  1.80  &  1.90  &  3.10  &  4.10  &        &  1.40  &  1.20  &  4.50  &  5.55 \\

		\midrule
		\midrule
		$n$ & & \multicolumn{9}{c}{GAS-$t$} \\
		\cmidrule(lr){3-11}
		$500$  &        & 5.45   & 5.50   & 7.70   & 8.00   &        & 4.90   & 5.30   & 7.75   & 9.15  \\
$1000$ &        & 4.15   & 4.45   & 6.05   & 6.75   &        & 2.00   & 2.25   & 4.75   & 6.00  \\
$2500$ &        & 2.00   & 1.90   & 3.10   & 3.40   &        & 1.25   & 1.35   & 3.10   & 3.90  \\
$5000$ &        & 1.70   & 1.80   & 3.95   & 4.05   &        & 1.00   & 1.05   & 2.15   & 2.70  \\

		\midrule
		\midrule
		$n$ & & \multicolumn{9}{c}{ES-CAViaR} \\
		\cmidrule(lr){3-11}
		$500$  &        & 2.05   & 1.30   & 4.30   & 6.00   &        & 2.35   & 1.45   & 5.05   & 6.50  \\
$1000$ &        & 1.85   & 1.25   & 3.55   & 5.55   &        & 1.65   & 1.25   & 3.10   & 4.85  \\
$2500$ &        & 1.00   & 1.15   & 2.30   & 3.10   &        & 1.00   & 0.90   & 2.05   & 3.00  \\
$5000$ &        & 1.15   & 0.85   & 1.55   & 2.15   &        & 1.10   & 1.15   & 1.15   & 1.85  \\

		\bottomrule 
		\addlinespace
		\multicolumn{11}{p{.93\linewidth}}{\textit{Notes:} This table presents the empirical sizes (in $\%$) of our three forecast encompassing tests for the ES together with the VaR encompassing test of \cite{GiacominiKomunjer2005} for a nominal size of $1\%$.
			The results are shown for the four DGPs described in Section 3.1 and Appendix \ref{sec:AdditionalDGPs} in the horizontal panels, for both tested hypotheses in the vertical panels and for different sample sizes.}
	\end{tabularx}
\end{table}

\begin{table}[ht]
	\footnotesize
	\centering
	\caption{Empirical Sizes of the Forecast Encompassing Tests.}
	\label{tab:Size5perc}
	\begin{tabularx}{0.95\linewidth}{ll @{\hspace{0.7cm}} RRRR l @{\hspace{0.5cm}} RRRR }
		\toprule
		& & \multicolumn{4}{c}{$\mathbb{H}_0^{(1)}$} & & \multicolumn{4}{c}{$\mathbb{H}_0^{(2)}$} \\
		\cmidrule(lr){3-6} \cmidrule(lr){8-11}
		& & Str ES & Aux ES & VaR ES & VaR & & Str ES & Aux ES & VaR ES & VaR \\ 
		\midrule
		\midrule
		$n$ & & \multicolumn{9}{c}{GARCH} \\
		\cmidrule(lr){3-11}
		$500$  &        &  9.20  &  9.30  & 13.90  & 17.50  &        &  8.75  &  9.30  & 13.45  & 17.55 \\
$1000$ &        &  6.90  &  6.45  & 11.40  & 14.40  &        &  6.90  &  6.35  & 12.65  & 17.10 \\
$2500$ &        &  6.35  &  6.40  & 11.10  & 13.55  &        &  5.90  &  5.75  &  9.85  & 12.05 \\
$5000$ &        &  5.65  &  5.25  &  8.65  &  9.75  &        &  5.00  &  5.05  &  9.00  & 10.65 \\

		\midrule
		\midrule
		$n$ & & \multicolumn{9}{c}{VaR/ES GAS} \\
		\cmidrule(lr){3-11}
		$500$  &        & 17.75  & 18.65  & 16.00  & 20.40  &        & 13.35  & 13.20  & 17.35  & 20.60 \\
$1000$ &        & 13.75  & 13.30  & 13.10  & 16.50  &        & 11.00  & 11.05  & 12.65  & 16.55 \\
$2500$ &        &  9.65  &  9.70  & 10.05  & 12.20  &        &  6.85  &  7.10  &  9.90  & 12.95 \\
$5000$ &        &  7.80  &  7.10  &  8.25  &  9.80  &        &  5.45  &  5.70  &  8.65  & 12.05 \\

		\midrule
		\midrule
		$n$ & & \multicolumn{9}{c}{GAS-$t$} \\
		\cmidrule(lr){3-11}
		$500$  &        & 14.50  & 14.30  & 14.40  & 14.50  &        & 11.80  & 11.50  & 13.90  & 16.20 \\
$1000$ &        & 11.80  & 11.75  & 12.30  & 13.75  &        &  7.90  &  8.10  &  9.60  & 11.00 \\
$2500$ &        &  7.00  &  6.85  &  9.85  &  9.75  &        &  6.05  &  6.05  &  7.25  &  9.40 \\
$5000$ &        &  7.05  &  7.05  &  9.50  &  8.85  &        &  5.30  &  5.35  &  6.65  &  7.15 \\

		\midrule
		\midrule
		$n$ & & \multicolumn{9}{c}{ES-CAViaR} \\
		\cmidrule(lr){3-11}
		$500$  &        &  6.75  &  5.95  &  9.15  & 12.90  &        &  7.20  &  5.75  &  9.70  & 13.80 \\
$1000$ &        &  7.00  &  6.15  &  8.40  & 11.05  &        &  6.35  &  5.30  &  7.85  & 10.65 \\
$2500$ &        &  5.10  &  4.45  &  5.60  &  8.05  &        &  5.05  &  4.40  &  6.20  &  8.70 \\
$5000$ &        &  5.40  &  4.80  &  5.20  &  7.15  &        &  5.25  &  5.10  &  5.00  &  6.85 \\

		\bottomrule 
		\addlinespace
		\multicolumn{11}{p{.93\linewidth}}{\textit{Notes:} This table presents the empirical sizes (in $\%$) of our three forecast encompassing tests for the ES together with the VaR encompassing test of \cite{GiacominiKomunjer2005} for a nominal size of $5\%$.
			The results are shown for the four DGPs described in Section 3.1 and Appendix \ref{sec:AdditionalDGPs} in the horizontal panels, for both tested hypotheses in the vertical panels and for different sample sizes.}
	\end{tabularx}
\end{table}

\begin{table}[ht]
	\footnotesize
	\centering
	\caption{Empirical Sizes for Two Additional DGPs}
	\label{tab:SizeExtensionDGPs}
	\begin{tabularx}{0.95\linewidth}{ll @{\hspace{0.7cm}} RRRR l @{\hspace{0.5cm}} RRRR }
		\toprule
		& & \multicolumn{4}{c}{$\mathbb{H}_0^{(1)}$} & & \multicolumn{4}{c}{$\mathbb{H}_0^{(2)}$} \\
		\cmidrule(lr){3-6} \cmidrule(lr){8-11}
		& & Str ES & Aux ES & VaR ES & VaR & & Str ES & Aux ES & VaR ES & VaR \\ 
		\midrule
		\midrule
		$n$ & & \multicolumn{9}{c}{GAS-$t$} \\
		\cmidrule(lr){3-11}
		$500$  &        & 21.95  & 21.75  & 20.25  & 19.50  &        & 18.50  & 18.10  & 18.20  & 21.75 \\
$1000$ &        & 18.75  & 18.35  & 18.75  & 19.35  &        & 14.25  & 13.95  & 13.80  & 17.45 \\
$2500$ &        & 12.40  & 12.05  & 15.45  & 14.90  &        & 11.65  & 11.75  & 12.55  & 15.85 \\
$5000$ &        & 12.50  & 12.35  & 15.25  & 14.50  &        &  9.60  &  9.20  & 11.50  & 12.90 \\

		\midrule
		\midrule
		$n$ & & \multicolumn{9}{c}{ES-CAViaR} \\
		\cmidrule(lr){3-11}
		$500$  &        & 12.95  & 11.80  & 13.55  & 19.00  &        & 13.05  & 11.55  & 15.05  & 19.40 \\
$1000$ &        & 12.30  & 11.70  & 12.70  & 17.25  &        & 11.40  & 10.60  & 11.95  & 16.50 \\
$2500$ &        & 10.35  &  9.45  &  9.65  & 13.75  &        & 10.85  &  9.55  & 10.20  & 13.10 \\
$5000$ &        &  9.65  &  9.65  & 10.65  & 12.65  &        & 10.35  &  9.75  & 10.20  & 11.70 \\

		\bottomrule 
		\addlinespace
		\multicolumn{11}{p{.93\linewidth}}{\textit{Notes:} This table presents the empirical sizes (in $\%$) of our three forecast encompassing tests for the ES together with the VaR encompassing test of \cite{GiacominiKomunjer2005} for a nominal size of $10\%$.
			The results are shown for the two additional DGPs described in Appendix \ref{sec:AdditionalDGPs} in the horizontal panels, for both tested hypotheses in the vertical panels and for different sample sizes.}
	\end{tabularx}
\end{table}

\begin{table}[ht]
	\footnotesize
	\centering
	\caption{Empirical Test Sizes for Different Loss Functions.}
	\label{tab:SizeExtensionLoss}
	\begin{tabularx}{0.82\linewidth}{ll @{\hspace{0.7cm}} RRR l @{\hspace{0.5cm}} RRR }
		\toprule
		& & \multicolumn{3}{c}{$\mathbb{H}_0^{(1)}$} & & \multicolumn{3}{c}{$\mathbb{H}_0^{(2)}$} \\
		\cmidrule(lr){3-5} \cmidrule(lr){7-9}
		& & Str ES & Aux ES & VaR ES & & Str ES & Aux ES & VaR ES \\ 
		\midrule
		\midrule
		$n$ & & \multicolumn{7}{c}{GARCH DGP \quad and \quad $\phi(z) = 1/\sqrt{-z}$} \\
		\cmidrule(lr){3-9}
		$500$  &        & 13.80  & 14.30  & 19.40  &        & 14.60  & 14.90  & 19.10 \\
$1000$ &        & 13.40  & 13.00  & 17.90  &        & 13.00  & 12.70  & 17.50 \\
$2500$ &        & 10.90  & 10.60  & 15.70  &        & 11.20  & 11.00  & 13.20 \\
$5000$ &        & 10.30  & 10.40  & 14.50  &        &  9.30  &  9.70  & 14.00 \\

		\midrule
		\midrule
		$n$ & & \multicolumn{7}{c}{GARCH DGP \quad and  \quad $\phi(z) = -1/z$} \\
		\cmidrule(lr){3-9}
		$500$  &        & 14.80  & 14.60  & 19.30  &        & 14.50  & 14.50  & 18.30 \\
$1000$ &        & 12.40  & 11.60  & 16.00  &        & 14.10  & 14.10  & 18.40 \\
$2500$ &        & 12.20  & 11.60  & 16.10  &        & 10.70  & 10.30  & 14.50 \\
$5000$ &        & 11.00  & 10.80  & 15.30  &        &  9.90  & 10.30  & 13.60 \\

		\midrule
		\midrule
		$n$ & & \multicolumn{7}{c}{VaR/ES GAS DGP \quad and \quad $\phi(z) = 1/\sqrt{-z}$} \\
		\cmidrule(lr){3-9}
		$500$  &        & 27.60  & 27.30  & 24.30  &        & 20.70  & 20.20  & 23.20 \\
$1000$ &        & 22.10  & 21.60  & 19.10  &        & 17.60  & 17.10  & 18.90 \\
$2500$ &        & 15.00  & 16.40  & 14.90  &        & 14.60  & 14.80  & 17.90 \\
$5000$ &        & 13.30  & 13.00  & 14.70  &        & 12.00  & 12.30  & 14.20 \\

		\midrule
		\midrule
		$n$ & & \multicolumn{7}{c}{VaR/ES GAS DGP \quad and \quad $\phi(z) = -1/z$} \\
		\cmidrule(lr){3-9}
		$500$  &        & 30.90  & 31.00  & 26.70  &        & 21.60  & 21.90  & 21.90 \\
$1000$ &        & 22.90  & 22.70  & 22.80  &        & 17.70  & 17.60  & 18.10 \\
$2500$ &        & 15.00  & 14.80  & 16.30  &        & 13.90  & 13.80  & 15.90 \\
$5000$ &        & 11.00  & 11.80  & 12.50  &        & 11.70  & 11.30  & 13.70 \\

		\bottomrule 
		\addlinespace
		\multicolumn{9}{p{.8\linewidth}}{\textit{Notes:} This table presents the empirical sizes (in $\%$) of our three forecast encompassing tests for the ES for a nominal size of $10\%$ for two additional strictly consistent loss functions and for the two DGPs described in Section 3.1.
			The first and third horizontal panels consider the choices $\mathfrak{g}(z) = 0$ and $\phi(z) = 1/\sqrt{-z}$ in the loss function in (3.7) while the second and fourth panel consider the choices $\mathfrak{g}(z) = 0$ and $\phi(z) = -1/{z}$.}
	\end{tabularx}
\end{table}

\begin{table}[ht]
	\footnotesize
	\centering
	\caption{Empirical Sizes for Different Link Functions.}
	\label{tab:SizeExtensionLinkFunctions}
	\begin{tabularx}{0.82\linewidth}{ll @{\hspace{0.7cm}} RRR l @{\hspace{0.5cm}} RRR }
		\toprule
		& & \multicolumn{3}{c}{$\mathbb{H}_0^{(1)}$} & & \multicolumn{3}{c}{$\mathbb{H}_0^{(2)}$} \\
		\cmidrule(lr){3-5} \cmidrule(lr){7-9}
		& & Str ES & Aux ES & VaR ES & & Str ES & Aux ES & VaR ES \\ 
		\midrule
		\midrule
		$n$ & & \multicolumn{7}{c}{Affine Link: GARCH DGP} \\
		\cmidrule(lr){3-9}
		$500$  &        & 12.93  & 11.52  & 13.83  &        & 14.21  & 12.31  & 14.01 \\
$1000$ &        & 14.20  & 12.70  & 14.00  &        & 11.80  &  9.70  & 13.70 \\
$2500$ &        & 13.20  & 11.40  & 13.40  &        & 11.70  &  9.90  & 12.80 \\
$5000$ &        & 11.60  &  9.90  & 11.80  &        & 12.20  &  9.60  & 15.60 \\

		\midrule
		\midrule
		$n$ & & \multicolumn{7}{c}{Affine Link: VaR/ES GAS DGP} \\
		\cmidrule(lr){3-9}
		$500$  &        & 23.52  & 22.72  & 15.82  &        & 14.70  & 13.70  & 13.50 \\
$1000$ &        & 21.30  & 18.20  & 14.60  &        & 10.90  & 10.40  & 13.00 \\
$2500$ &        & 16.20  & 14.30  & 12.80  &        & 10.50  &  9.80  & 12.80 \\
$5000$ &        & 13.80  & 11.00  & 10.70  &        &  9.90  &  9.90  & 12.00 \\

		\midrule
		\midrule
		$n$ & & \multicolumn{7}{c}{Affine Link: GAS-$t$ DGP} \\
		\cmidrule(lr){3-9}
		$500$  &        & 18.63  & 16.72  & 15.21  &        & 12.80  & 10.50  & 11.40 \\
$1000$ &        & 12.90  & 10.30  & 12.00  &        & 11.20  &  9.00  &  8.80 \\
$2500$ &        & 12.80  & 11.20  & 12.20  &        & 10.50  &  8.90  &  9.20 \\
$5000$ &        & 15.10  & 12.90  & 12.70  &        & 12.80  &  9.70  & 10.60 \\

		\midrule
		\midrule
		$n$ & & \multicolumn{7}{c}{Affine Link: ES-CAViaR DGP} \\
		\cmidrule(lr){3-9}
		$500$  &        & 11.88  & 10.39  & 11.56  &        & 11.33  & 10.40  & 11.54 \\
$1000$ &        & 12.45  & 12.75  & 14.26  &        & 12.42  & 10.71  & 11.01 \\
$2500$ &        &  9.80  &  9.40  & 10.30  &        & 12.10  & 11.90  & 12.00 \\
$5000$ &        &  8.50  &  8.60  &  9.80  &        & 10.40  & 10.30  & 10.80 \\

		\midrule
		\midrule
		\\
		$n$ & & \multicolumn{7}{c}{Nonlinear Link: Nonlinear DGP} \\
		\cmidrule(lr){3-9}
		$500$  &        &  8.12  &  8.82  & 13.53  &        &  8.82  &  8.72  & 16.03 \\
$1000$ &        &  8.70  &  9.10  & 14.80  &        &  8.90  &  9.10  & 13.60 \\
$2500$ &        & 10.40  & 10.10  & 11.90  &        &  9.50  & 10.10  & 14.30 \\
$5000$ &        & 11.30  & 11.30  & 17.20  &        & 11.10  & 10.80  & 12.50 \\

		\bottomrule 
		\addlinespace
		\multicolumn{9}{p{.8\linewidth}}{\textit{Notes:} This table presents the empirical sizes (in $\%$) of our three forecast encompassing tests for the ES for a nominal size of $10\%$ for the affine and the nonlinear link functions described in Section 3.3.3.
			The upper four horizontal panels consider the affine link functions for the four DGPs described in Sections 3.1 and in Appendix \ref{sec:AdditionalDGPs}.
			The lowest panel presents results for the nonlinear link functions based on the nonlinear GARCH DGP described in (3.10).}
	\end{tabularx}
\end{table}

\begin{landscape}
	\vspace*{\fill}
	\begin{table}[!htb]
		\caption{Correlation Matrices of the VaR and ES Forecasts in the Empirical Application.}
		\label{tab:CorrelationForecasts_2000}
		\centering
		\scriptsize
		\vspace{-0.2cm}
		\begin{tabular}{ll rrrrrrrr l rrrrrrrr}
			\toprule
			\addlinespace
			\multicolumn{19}{c}{Panel A: IBM Stock} \\
			\cmidrule(lr){1-19} 
			& & \multicolumn{8}{c}{Quantile Forecasts} & & \multicolumn{8}{c}{ES Forecasts} \\
			\cmidrule(lr){3-10} \cmidrule(lr){12-19}
			& & HS  &  RM  &  GJR  & GAS & G1F & G2F & ASES & SAVES  &   & HS & RM & GJR & GAS & G1F & G2F & ASES & SAVES \\
			\cmidrule(lr){3-10} \cmidrule(lr){12-19}
			HS &   & 1.00 & 0.58 & 0.54 & 0.65 & 0.31 & 0.41 & 0.40 & 0.46 &   & 1.00 & 0.58 & 0.54 & 0.65 & 0.31 & 0.41 & 0.40 & 0.46\\
RM &   &    & 1.00 & 0.96 & 0.95 & 0.78 & 0.86 & 0.84 & 0.89 &   &    & 1.00 & 0.96 & 0.95 & 0.78 & 0.86 & 0.84 & 0.89\\
GJR &   &    &    & 1.00 & 0.92 & 0.83 & 0.91 & 0.85 & 0.85 &   &    &    & 1.00 & 0.92 & 0.83 & 0.91 & 0.85 & 0.85\\
GAS &   &    &    &    & 1.00 & 0.70 & 0.82 & 0.81 & 0.87 &   &    &    &    & 1.00 & 0.70 & 0.82 & 0.81 & 0.87\\
G1F &   &    &    &    &    & 1.00 & 0.88 & 0.88 & 0.82 &   &    &    &    &    & 1.00 & 0.88 & 0.88 & 0.82\\
G2F &   &    &    &    &    &    & 1.00 & 0.84 & 0.82 &   &    &    &    &    &    & 1.00 & 0.84 & 0.82\\
ASES &   &    &    &    &    &    &    & 1.00 & 0.96 &   &    &    &    &    &    &    & 1.00 & 0.96\\
SAVES &   &    &    &    &    &    &    &    & 1.00 &   &    &    &    &    &    &    &    & 1.00\\

			\midrule
			\midrule
			\addlinespace
			\multicolumn{19}{c}{Panel B: S$\&$P 500 Index} \\
			\cmidrule(lr){1-19} 
			& & \multicolumn{8}{c}{Quantile Forecasts} & & \multicolumn{8}{c}{ES Forecasts} \\
			\cmidrule(lr){3-10} \cmidrule(lr){12-19}
			HS &   & 1.00 & 0.66 & 0.58 & 0.67 & 0.55 & 0.55 & 0.48 & 0.57 &   & 1.00 & 0.66 & 0.58 & 0.67 & 0.55 & 0.55 & 0.48 & 0.57\\
RM &   &    & 1.00 & 0.98 & 1.00 & 0.92 & 0.91 & 0.90 & 0.97 &   &    & 1.00 & 0.98 & 1.00 & 0.92 & 0.91 & 0.90 & 0.97\\
GJR &   &    &    & 1.00 & 0.97 & 0.95 & 0.94 & 0.96 & 0.97 &   &    &    & 1.00 & 0.97 & 0.95 & 0.94 & 0.96 & 0.97\\
GAS &   &    &    &    & 1.00 & 0.90 & 0.90 & 0.90 & 0.97 &   &    &    &    & 1.00 & 0.90 & 0.90 & 0.90 & 0.97\\
G1F &   &    &    &    &    & 1.00 & 0.96 & 0.89 & 0.90 &   &    &    &    &    & 1.00 & 0.96 & 0.89 & 0.90\\
G2F &   &    &    &    &    &    & 1.00 & 0.89 & 0.90 &   &    &    &    &    &    & 1.00 & 0.89 & 0.90\\
ASES &   &    &    &    &    &    &    & 1.00 & 0.94 &   &    &    &    &    &    &    & 1.00 & 0.94\\
SAVES &   &    &    &    &    &    &    &    & 1.00 &   &    &    &    &    &    &    &    & 1.00\\

			\midrule
			\midrule
			\addlinespace
			\multicolumn{19}{c}{Panel C: DAX 30 Index} \\
			\cmidrule(lr){1-19} 
			& & \multicolumn{8}{c}{Quantile Forecasts} & & \multicolumn{8}{c}{ES Forecasts} \\
			\cmidrule(lr){3-10} \cmidrule(lr){12-19}
			HS &   & 1.00 & 0.68 & 0.53 & 0.62 & 0.67 & 0.60 & 0.57 & 0.62 &   & 1.00 & 0.68 & 0.53 & 0.62 & 0.67 & 0.60 & 0.57 & 0.62\\
RM &   &    & 1.00 & 0.92 & 0.98 & 0.93 & 0.94 & 0.95 & 0.98 &   &    & 1.00 & 0.92 & 0.98 & 0.93 & 0.94 & 0.95 & 0.98\\
GJR &   &    &    & 1.00 & 0.96 & 0.91 & 0.95 & 0.99 & 0.95 &   &    &    & 1.00 & 0.96 & 0.91 & 0.95 & 0.99 & 0.95\\
GAS &   &    &    &    & 1.00 & 0.93 & 0.95 & 0.97 & 0.99 &   &    &    &    & 1.00 & 0.93 & 0.95 & 0.97 & 0.99\\
G1F &   &    &    &    &    & 1.00 & 0.96 & 0.90 & 0.91 &   &    &    &    &    & 1.00 & 0.96 & 0.90 & 0.91\\
G2F &   &    &    &    &    &    & 1.00 & 0.94 & 0.93 &   &    &    &    &    &    & 1.00 & 0.94 & 0.93\\
ASES &   &    &    &    &    &    &    & 1.00 & 0.98 &   &    &    &    &    &    &    & 1.00 & 0.98\\
SAVES &   &    &    &    &    &    &    &    & 1.00 &   &    &    &    &    &    &    &    & 1.00\\

			\bottomrule
			\addlinespace
			\multicolumn{19}{p{0.84\linewidth}}{\textit{Notes:} This table shows the correlations of the respective quantile and ES forecasts obtained from the eight forecasting models described in Section 4 for the IBM stock, the S\&P 500 and the DAX 30 indices estimated on an in-sample window of $m=2000$ days.
				The models are abbreviated as follows: Historical simulation (HS), RiskMetrics (RM), GAS-t model (GAS), GAS one factor (G1F) and GAS two factor (G2F) model, dynamic AS-ES-CAViaR (ASES) and dynamic SAV-ES-CAViaR (SAVES).}
		\end{tabular}
	\end{table}
	\vspace*{\fill}
\end{landscape}

\begin{landscape}
	\vspace*{\fill}
	\begin{table}[!htb]
		\caption{Forecast Encompassing Test $p$-values for the IBM Stock.}
		\label{tab:IBM_pval_2000}
		\centering
		\scriptsize
		\begin{tabular}{llllllllll l lllllllll}
			\toprule
			\\
			& & \multicolumn{8}{c}{Joint VaR ES encompassing test} & & \multicolumn{8}{c}{VaR encompassing test} \\
			\cmidrule(rr){3-10} \cmidrule(rr){12-19}
			& & HS  &  RM  &  GJR  & GAS & G1F & G2F & ASES & SAVES  &   & HS & RM & GJR & GAS & G1F & G2F & ASES & SAVES \\
			\cmidrule(rr){3-10} \cmidrule(rr){12-19}
			HS &   &   & 0.000* & 0.000* & 0.000* & 0.000* & 0.000  & 0.000  & 0.000  &   &   & 0.000* & 0.000* & 0.000* & 0.000* & 0.000  & 0.000  & 0.000 \\
RM &   & 0.035* &   & 0.057* & 0.000* & 0.002* & 0.001  & 0.000  & 0.000  &   & 0.079* &   & 0.111  & 0.000* & 0.001* & 0.001  & 0.000  & 0.000 \\
GJR &   & 0.006* & 0.002* &   & 0.000* & 0.000* & 0.000  & 0.000  & 0.000  &   & 0.002* & 0.002  &   & 0.004* & 0.000* & 0.000* & 0.000  & 0.000 \\
GAS &   & 0.000* & 0.000* & 0.000* &   & 0.000* & 0.000  & 0.000* & 0.000  &   & 0.000* & 0.000* & 0.000* &   & 0.000* & 0.000  & 0.000* & 0.000 \\
G1F &   & 0.000* & 0.000* & 0.000* & 0.000* &   & 0.000* & 0.000* & 0.000  &   & 0.000* & 0.000* & 0.000* & 0.000* &   & 0.000* & 0.000* & 0.000 \\
G2F &   & 0.574  & 0.397  & 0.200  & 0.475  & 0.091* &   & 0.000  & 0.000  &   & 0.827  & 0.175  & 0.087* & 0.338  & 0.017* &   & 0.000  & 0.000 \\
ASES &   & 0.577  & 0.648  & 0.621  & 0.024* & 0.006* & 0.419  &   & 0.341  &   & 0.809  & 0.368  & 0.737  & 0.009* & 0.004* & 0.302  &   & 0.115 \\
SAVES &   & 0.393  & 0.352  & 0.265  & 0.750  & 0.111  & 0.173  & 0.730  &   &   & 0.335  & 0.920  & 0.930  & 0.766  & 0.172  & 0.522  & 0.564  &  \\

			\midrule
			\midrule
			\\
			& & \multicolumn{8}{c}{Auxiliary ES encompassing test} & & \multicolumn{8}{c}{Strict ES encompassing test} \\
			\cmidrule(rr){3-10} \cmidrule(rr){12-19}
			& & HS  &  RM  &  GJR  & GAS & G1F & G2F & ASES & SAVES  &   & HS & RM & GJR & GAS & G1F & G2F & ASES & SAVES \\
			\cmidrule(rr){3-10} \cmidrule(rr){12-19}
			HS &   &   & 0.000* & 0.000* & 0.000* & 0.001  & 0.000  & 0.000  & 0.000  &   &   & 0.000* & 0.000* & 0.000* & 0.000  & 0.000  & 0.000  & 0.000 \\
RM &   & 0.016* &   & 0.040* & 0.002* & 0.055  & 0.033  & 0.000  & 0.000  &   & 0.013* &   & 0.043* & 0.003* & 0.053  & 0.057  & 0.000  & 0.000 \\
GJR &   & 0.023* & 0.011* &   & 0.000* & 0.052  & 0.019  & 0.000  & 0.000  &   & 0.023* & 0.011* &   & 0.000* & 0.048  & 0.014  & 0.000  & 0.000 \\
GAS &   & 0.002* & 0.020* & 0.007* &   & 0.029* & 0.018  & 0.001  & 0.001  &   & 0.001* & 0.022* & 0.007* &   & 0.034* & 0.034  & 0.000  & 0.001 \\
G1F &   & 0.298  & 0.260  & 0.388  & 0.076* &   & 0.242  & 0.000  & 0.001  &   & 0.307  & 0.244  & 0.391  & 0.073* &   & 0.192  & 0.000  & 0.001 \\
G2F &   & 0.265  & 0.710  & 0.788  & 0.276  & 0.825  &   & 0.002  & 0.002  &   & 0.250  & 0.729  & 0.756  & 0.282  & 0.793  &   & 0.006  & 0.005 \\
ASES &   & 0.497  & 0.713  & 0.528  & 0.890  & 0.446  & 0.472  &   & 0.824  &   & 0.485  & 0.589  & 0.408  & 0.866  & 0.338  & 0.580  &   & 0.729 \\
SAVES &   & 0.185  & 0.183  & 0.136  & 0.570  & 0.254  & 0.212  & 0.473  &   &   & 0.255  & 0.170  & 0.125  & 0.563  & 0.222  & 0.310  & 0.454  &  \\

			\bottomrule
			\addlinespace
			\multicolumn{19}{p{.99\linewidth}}{\textit{Notes:} This table reports the p-values of the three different ES encompassing tests introduced in Section 2 and of the VaR encompassing test of \cite{GiacominiKomunjer2005} applied to returns from the IBM stock and an in-sample period of $m=2000$ days.
				The p-value in the $i$-th row and $j$-th column of the respective matrices corresponds to testing $\mathbb{H}_0$: the forecasts from model $i$ encompass the forecasts from model $j$.
				The symbol $^\ast$ denotes model pairs where both encompassing tests (test whether model $i$ encompasses model $j$ and vice versa) are significant at the $10\%$ level. 
				The models are abbreviated as follows: Historical simulation (HS), RiskMetrics (RM), GAS-t model (GAS), GAS one factor (G1F) and GAS two factor (G2F) model, dynamic AS-ES-CAViaR (ASES) and dynamic SAV-ES-CAViaR (SAVES). }
		\end{tabular}
	\end{table}
	\vspace*{\fill}
\end{landscape}

\begin{landscape}
	\vspace*{\fill}
	\begin{table}[!htb]
		\caption{Forecast Encompassing Test $p$-values for the S\&P 500 Index.}
		\label{tab:SP500_pval_2000}
		\centering
		\scriptsize
		\begin{tabular}{llllllllll l lllllllll}
			\toprule
			\\
			& & \multicolumn{8}{c}{Joint VaR ES encompassing test} & & \multicolumn{8}{c}{VaR encompassing test} \\
			\cmidrule(rr){3-10} \cmidrule(rr){12-19}
			& & HS  &  RM  &  GJR  & GAS & G1F & G2F & ASES & SAVES  &   & HS & RM & GJR & GAS & G1F & G2F & ASES & SAVES \\
			\cmidrule(rr){3-10} \cmidrule(rr){12-19}
			HS &   &   & 0.000  & 0.000  & 0.000  & 0.000* & 0.000  & 0.000* & 0.000  &   &   & 0.000  & 0.000  & 0.000  & 0.000  & 0.000  & 0.000* & 0.000 \\
RM &   & 0.389  &   & 0.000  & 0.067  & 0.000* & 0.001  & 0.000  & 0.099  &   & 0.809  &   & 0.000  & 0.206  & 0.000* & 0.001  & 0.000  & 0.070*\\
GJR &   & 0.255  & 0.281  &   & 0.682  & 0.193  & 0.035* & 0.000* & 0.860  &   & 0.505  & 0.412  &   & 0.540  & 0.177  & 0.005* & 0.000* & 0.736 \\
GAS &   & 0.128  & 0.226  & 0.000  &   & 0.000* & 0.010  & 0.000  & 0.047  &   & 0.106  & 0.134  & 0.001  &   & 0.000  & 0.003* & 0.000* & 0.013 \\
G1F &   & 0.057* & 0.092* & 0.009  & 0.056* &   & 0.039* & 0.000* & 0.027* &   & 0.170  & 0.082* & 0.003  & 0.121  &   & 0.107  & 0.001* & 0.125 \\
G2F &   & 0.309  & 0.406  & 0.000* & 0.156  & 0.096* &   & 0.005* & 0.119  &   & 0.127  & 0.166  & 0.000* & 0.056* & 0.051  &   & 0.007* & 0.094*\\
ASES &   & 0.011* & 0.479  & 0.002* & 0.421  & 0.006* & 0.048* &   & 0.662  &   & 0.008* & 0.144  & 0.000* & 0.055* & 0.003* & 0.007* &   & 0.395 \\
SAVES &   & 0.556  & 0.131  & 0.007  & 0.467  & 0.000* & 0.002  & 0.000  &   &   & 0.473  & 0.055* & 0.002  & 0.154  & 0.000  & 0.000* & 0.000  &  \\

			\midrule
			\midrule
			\\
			& & \multicolumn{8}{c}{Auxiliary ES encompassing test} & & \multicolumn{8}{c}{Strict ES encompassing test} \\
			\cmidrule(rr){3-10} \cmidrule(rr){12-19}
			& & HS  &  RM  &  GJR  & GAS & G1F & G2F & ASES & SAVES  &   & HS & RM & GJR & GAS & G1F & G2F & ASES & SAVES \\
			\cmidrule(rr){3-10} \cmidrule(rr){12-19}
			HS &   &   & 0.000  & 0.000  & 0.000  & 0.000* & 0.000  & 0.000  & 0.000  &   &   & 0.000  & 0.000  & 0.000  & 0.000* & 0.000  & 0.000  & 0.000 \\
RM &   & 0.131  &   & 0.000  & 0.033  & 0.003* & 0.020  & 0.002  & 0.072  &   & 0.119  &   & 0.000  & 0.015  & 0.002* & 0.019  & 0.002  & 0.059 \\
GJR &   & 0.544  & 0.110  &   & 0.476  & 0.252  & 0.575  & 0.376  & 0.853  &   & 0.580  & 0.108  &   & 0.432  & 0.239  & 0.621  & 0.364  & 0.900 \\
GAS &   & 0.222  & 0.181  & 0.003  &   & 0.015* & 0.100  & 0.020  & 0.794  &   & 0.220  & 0.158  & 0.003  &   & 0.017* & 0.133  & 0.019  & 0.811 \\
G1F &   & 0.039* & 0.028* & 0.008  & 0.015* &   & 0.015  & 0.000* & 0.007* &   & 0.031* & 0.035* & 0.009  & 0.019* &   & 0.016  & 0.001* & 0.008*\\
G2F &   & 0.632  & 0.516  & 0.123  & 0.242  & 0.131  &   & 0.043  & 0.293  &   & 0.501  & 0.429  & 0.101  & 0.172  & 0.157  &   & 0.029  & 0.258 \\
ASES &   & 0.779  & 0.674  & 0.454  & 0.500  & 0.056* & 0.410  &   & 0.817  &   & 0.778  & 0.651  & 0.464  & 0.497  & 0.053* & 0.515  &   & 0.802 \\
SAVES &   & 0.819  & 0.936  & 0.034  & 0.486  & 0.009* & 0.089  & 0.024  &   &   & 0.821  & 0.945  & 0.039  & 0.465  & 0.013* & 0.088  & 0.021  &  \\

			\bottomrule
			\addlinespace
			\multicolumn{19}{p{.99\linewidth}}{\textit{Notes:} This table reports the p-values of the three different ES encompassing tests introduced in Section 2 and of the VaR encompassing test of \cite{GiacominiKomunjer2005} applied to returns from the S\&P 500 index and an in-sample period of $m=2000$ days.
				The p-value in the $i$-th row and $j$-th column of the respective matrices corresponds to testing $\mathbb{H}_0$: the forecasts from model $i$ encompass the forecasts from model $j$.
				The symbol $^\ast$ denotes model pairs where both encompassing tests (test whether model $i$ encompasses model $j$ and vice versa) are significant at the $10\%$ level. 
				The models are abbreviated as follows: Historical simulation (HS), RiskMetrics (RM), GAS-t model (GAS), GAS one factor (G1F) and GAS two factor (G2F) model, dynamic AS-ES-CAViaR (ASES) and dynamic SAV-ES-CAViaR (SAVES). }
		\end{tabular}
	\end{table}
	\vspace*{\fill}
\end{landscape}

\begin{landscape}
	\vspace*{\fill}
	\begin{table}[!htb]
		\caption{Forecast Encompassing Test $p$-values for the DAX 30 Index.}
		\label{tab:DAX_pval_2000}
		\centering
		\scriptsize
		\begin{tabular}{llllllllll l lllllllll}
			\toprule
			\\
			& & \multicolumn{8}{c}{Joint VaR ES encompassing test} & & \multicolumn{8}{c}{VaR encompassing test} \\
			\cmidrule(rr){3-10} \cmidrule(rr){12-19}
			& & HS  &  RM  &  GJR  & GAS & G1F & G2F & ASES & SAVES  &   & HS & RM & GJR & GAS & G1F & G2F & ASES & SAVES \\
			\cmidrule(rr){3-10} \cmidrule(rr){12-19}
			HS &   &   & 0.000  & 0.000  & 0.000  & 0.000  & 0.000* & 0.000  & 0.000  &   &   & 0.000  & 0.000  & 0.000  & 0.000  & 0.000* & 0.000  & 0.000 \\
RM &   & 0.587  &   & 0.000  & 0.000* & 0.160  & 0.125  & 0.000  & 0.004  &   & 0.252  &   & 0.000* & 0.000* & 0.020* & 0.030* & 0.000  & 0.025 \\
GJR &   & 0.375  & 0.122  &   & 0.175  & 0.455  & 0.775  & 0.000  & 0.085* &   & 0.279  & 0.046* &   & 0.073* & 0.456  & 0.917  & 0.000  & 0.144 \\
GAS &   & 0.936  & 0.096* & 0.049  &   & 0.832  & 0.647  & 0.005  & 0.252  &   & 0.889  & 0.042* & 0.080* &   & 0.692  & 0.304  & 0.001  & 0.854 \\
G1F &   & 0.732  & 0.011  & 0.000  & 0.000  &   & 0.037* & 0.000  & 0.000  &   & 0.590  & 0.036* & 0.000  & 0.000  &   & 0.013* & 0.000  & 0.000*\\
G2F &   & 0.012* & 0.001  & 0.000  & 0.000  & 0.020* &   & 0.000  & 0.000* &   & 0.022* & 0.004* & 0.000  & 0.001  & 0.097* &   & 0.000  & 0.000*\\
ASES &   & 0.869  & 0.905  & 0.537  & 0.867  & 0.965  & 0.997  &   & 0.079* &   & 0.848  & 0.964  & 0.793  & 0.986  & 0.927  & 0.964  &   & 0.543 \\
SAVES &   & 0.673  & 0.476  & 0.047* & 0.239  & 0.273  & 0.068* & 0.000* &   &   & 0.742  & 0.760  & 0.022  & 0.072  & 0.028* & 0.029* & 0.000  &  \\

			\midrule
			\midrule
			\\
			& & \multicolumn{8}{c}{Auxiliary ES encompassing test} & & \multicolumn{8}{c}{Strict ES encompassing test} \\
			\cmidrule(rr){3-10} \cmidrule(rr){12-19}
			& & HS  &  RM  &  GJR  & GAS & G1F & G2F & ASES & SAVES  &   & HS & RM & GJR & GAS & G1F & G2F & ASES & SAVES \\
			\cmidrule(rr){3-10} \cmidrule(rr){12-19}
			HS &   &   & 0.000  & 0.000  & 0.000  & 0.000  & 0.000  & 0.000  & 0.000  &   &   & 0.000  & 0.000  & 0.000  & 0.000  & 0.000* & 0.000  & 0.000 \\
RM &   & 0.518  &   & 0.152  & 0.043  & 0.271  & 0.229  & 0.034  & 0.001  &   & 0.459  &   & 0.161  & 0.036  & 0.247  & 0.278  & 0.027  & 0.001 \\
GJR &   & 0.188  & 0.093  &   & 0.052  & 0.170  & 0.421  & 0.011  & 0.025  &   & 0.148  & 0.070  &   & 0.057  & 0.156  & 0.465  & 0.005  & 0.007 \\
GAS &   & 0.994  & 0.796  & 0.838  &   & 0.892  & 0.951  & 0.333  & 0.101  &   & 0.967  & 0.808  & 0.813  &   & 0.880  & 0.945  & 0.339  & 0.041 \\
G1F &   & 0.929  & 0.048  & 0.047  & 0.001  &   & 0.290  & 0.000  & 0.000  &   & 0.923  & 0.059  & 0.051  & 0.001  &   & 0.359  & 0.000  & 0.000 \\
G2F &   & 0.131  & 0.014  & 0.026  & 0.001  & 0.035  &   & 0.000  & 0.000  &   & 0.087* & 0.012  & 0.033  & 0.001  & 0.021  &   & 0.000  & 0.000 \\
ASES &   & 0.625  & 0.746  & 0.445  & 0.736  & 0.837  & 0.982  &   & 0.268  &   & 0.600  & 0.743  & 0.303  & 0.695  & 0.820  & 0.960  &   & 0.208 \\
SAVES &   & 0.809  & 0.662  & 0.898  & 0.800  & 0.520  & 0.929  & 0.830  &   &   & 0.766  & 0.680  & 0.991  & 0.817  & 0.612  & 0.846  & 0.694  &  \\

			\bottomrule
			\addlinespace
			\multicolumn{19}{p{.99\linewidth}}{\textit{Notes:} This table reports the p-values of the three different ES encompassing tests introduced in Section 2 and of the VaR encompassing test of \cite{GiacominiKomunjer2005} applied to returns from the DAX 30 index and an in-sample period of $m=2000$ days.
				The p-value in the $i$-th row and $j$-th column of the respective matrices corresponds to testing $\mathbb{H}_0$: the forecasts from model $i$ encompass the forecasts from model $j$.
				The symbol $^\ast$ denotes model pairs where both encompassing tests (test whether model $i$ encompasses model $j$ and vice versa) are significant at the $10\%$ level. 
				The models are abbreviated as follows: Historical simulation (HS), RiskMetrics (RM), GAS-t model (GAS), GAS one factor (G1F) and GAS two factor (G2F) model, dynamic AS-ES-CAViaR (ASES) and dynamic SAV-ES-CAViaR (SAVES). }
		\end{tabular}
	\end{table}
	\vspace*{\fill}
\end{landscape}

\begin{table}[p]
	\caption{Average Losses for the VaR and ES Forecasts for the IBM Stock.}
	\label{tab:score_IBM_2000}
	\centering
	\footnotesize
	\begin{tabularx}{\linewidth}{lc @{\hspace{0.4cm}} cc @{\hspace{0.4cm}}  CCCCCCCC }
		\toprule
		\\
		& & & & \multicolumn{8}{c}{Panel A: Forecast Combination Losses \& Joint VaR ES Encompassing Weights}  \\
		\cmidrule(lr){1-3} \cmidrule(lr){5-12} 
		Model & & Loss & & HS & RM & GJR & GAS & G1F & G2F & ASES & SAVES  \\  
		\cmidrule(lr){1-3} \cmidrule(lr){5-12} 
		\addlinespace
		HS &   & 1.544 &   &     & 1.369 & 1.381 & 1.358 & 1.380 & 1.375 & 1.326 & 1.324\\
RM &   & 1.504 &   &     &     & 1.370 & 1.365 & 1.368 & 1.370 & 1.325 & 1.321\\
GJR &   & 1.476 &   &     &     &     & 1.363 & 1.375 & 1.376 & 1.324 & 1.320\\
GAS &   & 1.420 &   &     &     &     &     & 1.360 & 1.363 & 1.325 & 1.324\\
G1F &   & 1.403 &   &     &     &     &     &     & 1.375 & 1.323 & 1.318\\
G2F &   & 1.382 &   &     &     &     &     &     &     & 1.322 & 1.316\\
ASES &   & 1.329 &   &     &     &     &     &     &     &     & 1.322\\
SAVES &   & 1.325 &   &     &     &     &     &     &     &     &    \\

		\midrule
		\midrule
		\\
		& & & & \multicolumn{8}{c}{Panel B: Forecast Combination Losses \& Strict ES Encompassing Weights}  \\
		\cmidrule(lr){1-3} \cmidrule(lr){5-12} 
		Model & & Loss	& & HS & RM & GJR & GAS & G1F & G2F & ASES & SAVES  \\  
		\cmidrule(lr){1-3} \cmidrule(lr){5-12} 
		\addlinespace
		HS &   & 1.544 &   &     & 1.369 & 1.378 & 1.355 & 1.380 & 1.367 & 1.334 & 1.323\\
RM &   & 1.504 &   &     &     & 1.370 & 1.364 & 1.369 & 1.372 & 1.324 & 1.320\\
GJR &   & 1.476 &   &     &     &     & 1.361 & 1.375 & 1.376 & 1.323 & 1.320\\
GAS &   & 1.420 &   &     &     &     &     & 1.359 & 1.364 & 1.326 & 1.323\\
G1F &   & 1.403 &   &     &     &     &     &     & 1.370 & 1.325 & 1.317\\
G2F &   & 1.382 &   &     &     &     &     &     &     & 1.323 & 1.317\\
ASES &   & 1.329 &   &     &     &     &     &     &     &     & 1.320\\
SAVES &   & 1.325 &   &     &     &     &     &     &     &     &    \\

		\midrule
		\midrule
		\\
		& & & & \multicolumn{8}{c}{Panel C: Forecast Combination Losses \& VaR Encompassing Weights}  \\
		\cmidrule(lr){1-3} \cmidrule(lr){5-12} 
		Model & & Loss	& & HS & RM & GJR & GAS & G1F & G2F & ASES & SAVES  \\  
		\cmidrule(lr){1-3} \cmidrule(lr){5-12} 
		\addlinespace
		HS &   & 1.120 &   &     & 0.993 & 1.002 & 0.988 & 0.998 & 0.997 & 0.971 & 0.969\\
RM &   & 1.008 &   &     &     & 0.998 & 0.994 & 0.988 & 0.989 & 0.970 & 0.970\\
GJR &   & 1.010 &   &     &     &     & 0.996 & 0.992 & 0.994 & 0.971 & 0.970\\
GAS &   & 1.008 &   &     &     &     &     & 0.985 & 0.988 & 0.969 & 0.970\\
G1F &   & 1.022 &   &     &     &     &     &     & 0.993 & 0.971 & 0.969\\
G2F &   & 0.998 &   &     &     &     &     &     &     & 0.971 & 0.969\\
ASES &   & 0.976 &   &     &     &     &     &     &     &     & 0.968\\
SAVES &   & 0.970 &   &     &     &     &     &     &     &     &    \\

		\bottomrule
		\addlinespace
		\multicolumn{12}{p{.98\linewidth}} {\textit{Notes:} This table shows average out-of-sample forecast losses for the eight stand-alone models described in Section 4 and the respective linear forecast combinations with estimated weights for the IBM stock.
			The  column labeled "Loss" reports the average losses for the individual forecasting models and the remaining eight columns report the average losses of the forecast combinations with the respective estimated combination parameters.
			In Panel A, the estimated weights are obtained from the underlying regression of the "joint VaR and ES encompassing test" (or equivalently from the "auxiliary ES encompassing test"), in Panel B from the "strict ES encompassing test" and in Panel C from the "VaR encompassing test". The values in Panel C are multiplied by 10.
			The losses given in Panel A and B refer to the joint VaR and ES loss function given in (2.8) whereas in Panel C, the values are obtained by using the quantile-specific piecewise linear loss function.
		}
	\end{tabularx}
\end{table}

\begin{table}[p]
	\caption{Average Losses for the VaR and ES Forecasts for the S\&P 500 Index.}
	\label{tab:score_SP500_2000}
	\centering
	\footnotesize
	\begin{tabularx}{\linewidth}{lc @{\hspace{0.4cm}} cc @{\hspace{0.4cm}}  CCCCCCCC }
		\toprule
		\\
		& & & & \multicolumn{8}{c}{Panel A: Forecast Combination Losses \& Joint VaR ES Encompassing Weights}  \\
		\cmidrule(lr){1-3} \cmidrule(lr){5-12} 
		Model & & Loss & & HS & RM & GJR & GAS & G1F & G2F & ASES & SAVES  \\  
		\cmidrule(lr){1-3} \cmidrule(lr){5-12} 
		\addlinespace
		HS &   & 1.267 &   &     & 1.031 & 0.985 & 1.020 & 0.978 & 0.995 & 0.948 & 1.009\\
RM &   & 1.157 &   &     &     & 0.976 & 1.014 & 0.982 & 0.996 & 0.960 & 1.011\\
GJR &   & 1.070 &   &     &     &     & 0.980 & 0.976 & 0.981 & 0.962 & 0.986\\
GAS &   & 1.090 &   &     &     &     &     & 0.981 & 0.993 & 0.961 & 1.011\\
G1F &   & 1.001 &   &     &     &     &     &     & 0.981 & 0.955 & 0.977\\
G2F &   & 1.009 &   &     &     &     &     &     &     & 0.961 & 0.989\\
ASES &   & 0.960 &   &     &     &     &     &     &     &     & 0.960\\
SAVES &   & 1.013 &   &     &     &     &     &     &     &     &    \\

		\midrule
		\midrule
		\\
		& & & & \multicolumn{8}{c}{Panel B: Forecast Combination Losses \& Strict ES Encompassing Weights}  \\
		\cmidrule(lr){1-3} \cmidrule(lr){5-12} 
		Model & & Loss	& & HS & RM & GJR & GAS & G1F & G2F & ASES & SAVES  \\  
		\cmidrule(lr){1-3} \cmidrule(lr){5-12} 
		\addlinespace
		HS &   & 1.267 &   &     & 1.031 & 0.986 & 1.015 & 0.979 & 0.998 & 0.952 & 1.011\\
RM &   & 1.157 &   &     &     & 0.976 & 1.003 & 0.982 & 0.996 & 0.960 & 1.011\\
GJR &   & 1.070 &   &     &     &     & 0.980 & 0.976 & 0.981 & 0.962 & 0.985\\
GAS &   & 1.090 &   &     &     &     &     & 0.981 & 0.992 & 0.961 & 1.009\\
G1F &   & 1.001 &   &     &     &     &     &     & 0.981 & 0.955 & 0.977\\
G2F &   & 1.009 &   &     &     &     &     &     &     & 0.960 & 0.990\\
ASES &   & 0.960 &   &     &     &     &     &     &     &     & 0.960\\
SAVES &   & 1.013 &   &     &     &     &     &     &     &     &    \\

		\midrule
		\midrule
		\\
		& & & & \multicolumn{8}{c}{Panel C: Forecast Combination Losses \& VaR Encompassing Weights}  \\
		\cmidrule(lr){1-3} \cmidrule(lr){5-12} 
		Model & & Loss	& & HS & RM & GJR & GAS & G1F & G2F & ASES & SAVES  \\  
		\cmidrule(lr){1-3} \cmidrule(lr){5-12} 
		\addlinespace
		HS &   & 9.426 &   &     & 7.644 & 7.407 & 7.592 & 7.420 & 7.464 & 7.307 & 7.581\\
RM &   & 7.876 &   &     &     & 7.382 & 7.591 & 7.414 & 7.442 & 7.368 & 7.566\\
GJR &   & 7.630 &   &     &     &     & 7.399 & 7.381 & 7.369 & 7.357 & 7.411\\
GAS &   & 7.876 &   &     &     &     &     & 7.400 & 7.427 & 7.364 & 7.557\\
G1F &   & 7.463 &   &     &     &     &     &     & 7.404 & 7.305 & 7.391\\
G2F &   & 7.539 &   &     &     &     &     &     &     & 7.319 & 7.436\\
ASES &   & 7.414 &   &     &     &     &     &     &     &     & 7.389\\
SAVES &   & 7.608 &   &     &     &     &     &     &     &     &    \\

		\bottomrule
		\addlinespace
		\multicolumn{12}{p{.98\linewidth}} {\textit{Notes:} This table shows average out-of-sample forecast losses for the eight stand-alone models described in Section 4 and the respective linear forecast combinations with estimated weights for the S\&P 500 index.
			The  column labeled "Loss" reports the average losses for the individual forecasting models and the remaining eight columns report the average losses of the forecast combinations with the respective estimated combination parameters.
			In Panel A, the estimated weights are obtained from the underlying regression of the "joint VaR and ES encompassing test" (or equivalently from the "auxiliary ES encompassing test"), in Panel B from the "strict ES encompassing test" and in Panel C from the "VaR encompassing test". The values in Panel C are multiplied by 100.
			The losses given in Panel A and B refer to the joint VaR and ES loss function given in (2.8) whereas in Panel C, the values are obtained by using the quantile-specific piecewise linear loss function.
		}
	\end{tabularx}
\end{table}

\begin{table}[p]
	\caption{Average Losses for the VaR and ES Forecasts for the DAX 30 Index.}
	\label{tab:score_DAX_2000}
	\centering
	\footnotesize
	\begin{tabularx}{\linewidth}{lc @{\hspace{0.4cm}} cc @{\hspace{0.4cm}}  CCCCCCCC }
		\toprule
		\\
		& & & & \multicolumn{8}{c}{Panel A: Forecast Combination Losses \& Joint VaR ES Encompassing Weights}  \\
		\cmidrule(lr){1-3} \cmidrule(lr){5-12} 
		Model & & Loss & & HS & RM & GJR & GAS & G1F & G2F & ASES & SAVES  \\  
		\cmidrule(lr){1-3} \cmidrule(lr){5-12} 
		\addlinespace
		HS &   & 1.331 &   &     & 1.140 & 1.117 & 1.125 & 1.150 & 1.150 & 1.107 & 1.121\\
RM &   & 1.202 &   &     &     & 1.116 & 1.124 & 1.138 & 1.136 & 1.108 & 1.121\\
GJR &   & 1.159 &   &     &     &     & 1.116 & 1.121 & 1.124 & 1.106 & 1.111\\
GAS &   & 1.166 &   &     &     &     &     & 1.124 & 1.124 & 1.108 & 1.120\\
G1F &   & 1.152 &   &     &     &     &     &     & 1.151 & 1.108 & 1.120\\
G2F &   & 1.173 &   &     &     &     &     &     &     & 1.108 & 1.118\\
ASES &   & 1.109 &   &     &     &     &     &     &     &     & 1.105\\
SAVES &   & 1.121 &   &     &     &     &     &     &     &     &    \\

		\midrule
		\midrule
		\\
		& & & & \multicolumn{8}{c}{Panel B: Forecast Combination Losses \& Strict ES Encompassing Weights}  \\
		\cmidrule(lr){1-3} \cmidrule(lr){5-12} 
		Model & & Loss	& & HS & RM & GJR & GAS & G1F & G2F & ASES & SAVES  \\  
		\cmidrule(lr){1-3} \cmidrule(lr){5-12} 
		\addlinespace
		HS &   & 1.331 &   &     & 1.138 & 1.118 & 1.125 & 1.151 & 1.148 & 1.110 & 1.112\\
RM &   & 1.202 &   &     &     & 1.116 & 1.124 & 1.138 & 1.131 & 1.110 & 1.111\\
GJR &   & 1.159 &   &     &     &     & 1.116 & 1.121 & 1.124 & 1.108 & 1.104\\
GAS &   & 1.166 &   &     &     &     &     & 1.124 & 1.123 & 1.110 & 1.110\\
G1F &   & 1.152 &   &     &     &     &     &     & 1.151 & 1.110 & 1.111\\
G2F &   & 1.173 &   &     &     &     &     &     &     & 1.109 & 1.106\\
ASES &   & 1.109 &   &     &     &     &     &     &     &     & 1.106\\
SAVES &   & 1.121 &   &     &     &     &     &     &     &     &    \\

		\midrule
		\midrule
		\\
		& & & & \multicolumn{8}{c}{Panel C: Forecast Combination Losses \& VaR Encompassing Weights}  \\
		\cmidrule(lr){1-3} \cmidrule(lr){5-12} 
		Model & & Loss	& & HS & RM & GJR & GAS & G1F & G2F & ASES & SAVES  \\  
		\cmidrule(lr){1-3} \cmidrule(lr){5-12} 
		\addlinespace
		HS &   & 9.853 &   &     & 8.362 & 8.142 & 8.237 & 8.416 & 8.370 & 8.097 & 8.276\\
RM &   & 8.622 &   &     &     & 8.134 & 8.215 & 8.330 & 8.291 & 8.097 & 8.277\\
GJR &   & 8.290 &   &     &     &     & 8.144 & 8.167 & 8.174 & 8.098 & 8.140\\
GAS &   & 8.447 &   &     &     &     &     & 8.235 & 8.225 & 8.097 & 8.235\\
G1F &   & 8.431 &   &     &     &     &     &     & 8.392 & 8.097 & 8.246\\
G2F &   & 8.561 &   &     &     &     &     &     &     & 8.097 & 8.217\\
ASES &   & 8.105 &   &     &     &     &     &     &     &     & 8.084\\
SAVES &   & 8.283 &   &     &     &     &     &     &     &     &    \\

		\bottomrule
		\addlinespace
		\multicolumn{12}{p{.98\linewidth}} {\textit{Notes:} This table shows average out-of-sample forecast losses for the eight stand-alone models described in Section 4 and the respective linear forecast combinations with estimated weights for the DAX 30 index.
			The  column labeled "Loss" reports the average losses for the individual forecasting models and the remaining eight columns report the average losses of the forecast combinations with the respective estimated combination parameters.
			In Panel A, the estimated weights are obtained from the underlying regression of the "joint VaR and ES encompassing test" (or equivalently from the "auxiliary ES encompassing test"), in Panel B from the "strict ES encompassing test" and in Panel C from the "VaR encompassing test". The values in Panel C are multiplied by 100.
			The losses given in Panel A and B refer to the joint VaR and ES loss function given in (2.8) whereas in Panel C, the values are obtained by using the quantile-specific piecewise linear loss function.
		}
	\end{tabularx}
\end{table}

\section{Additional Figures}
\label{sec:Figures}

\begin{figure}[hptb]
	\includegraphics[width=\linewidth]{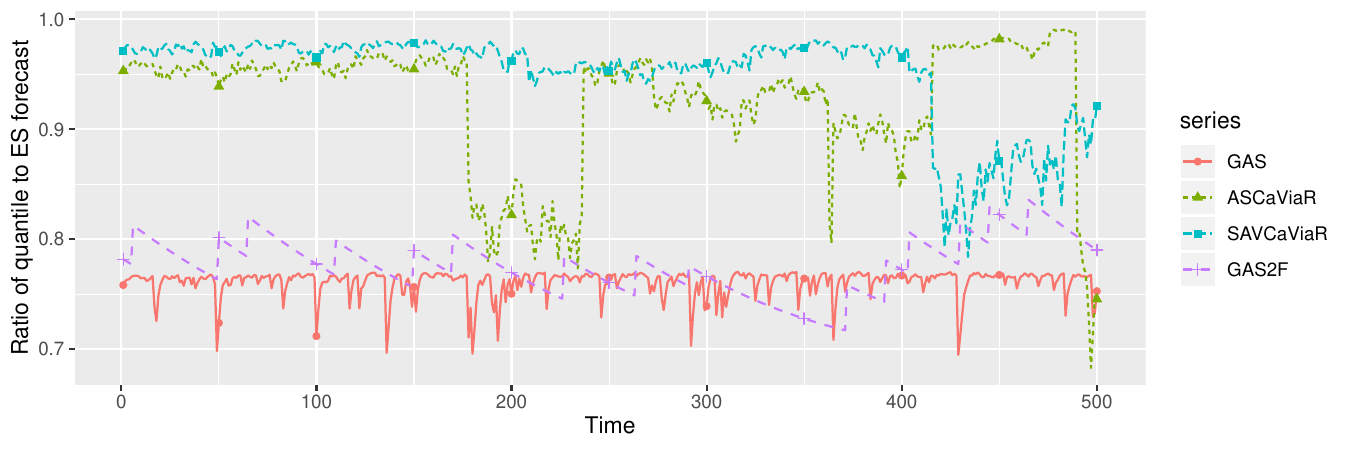}
	\caption{This figure shows the ratio of the true VaR and ES forecasts $\hat{q}_t/\hat{e}_t$ for simulated paths of the GAS-$t$ model of \cite{Creal2013}, the two factor GAS model for the VaR and ES introduced in \cite{Patton2019}, the AS-ES-CAViaR model and the SAV-ES-CAViaR model proposed by \cite{Taylor2019}.}		
	\label{fig:VaRESDegreeMisspec}
\end{figure}

\end{document}